%% file: dynamic-clustering.tex
\documentclass[english, 11pt, letterpaper]{article}
\usepackage[T1]{fontenc}
\usepackage[utf8]{inputenc}
\usepackage{babel}
\usepackage{csquotes}
\usepackage{booktabs}
\usepackage[inline]{enumitem}
\newlist{inenum}{enumerate*}{1}
\setlist[inenum]{label=(\roman*)}
\newlist{compitem}{itemize}{3}
\setlist[compitem]{label=\textbullet,nosep}
\setlist{nosep}

\usepackage{titlesec}

\usepackage[
  margin=1in
]{geometry}
\usepackage{microtype}
\usepackage[dvipsnames]{xcolor}

\usepackage{times}
\usepackage{tablefootnote}

\usepackage{hyperref}
\usepackage[noend]{algorithm2e}
\usepackage{amsmath}
\usepackage{amssymb}
\usepackage{amsthm} %
\usepackage{thmtools,thm-restate} %
\usepackage[capitalise]{cleveref} %
\usepackage{crossreftools}

\usepackage{etoolbox}
\usepackage{xparse}

\input{preamble.tex}

\usepackage[draft,footnote,marginclue]{fixme}
\newcommand{\mhr}[1]{}
\newcommand{\hfr}[1]{}

\newcommand\aw[1]{{#1}}
\newcommand\awtwo[1]{{#1}}
\newcommand\hf[1]{{#1}}
\def\eps{\epsilon}
\global\long\def\N{\mathbb{N}}%
\global\long\def\R{\mathbb{R}}%
\global\long\def\OPT{\mathrm{OPT}}%
\global\long\def\ALG{\mathsf{apx}}%
\global\long\def\apx{\mathsf{apx}}

\usepackage{authblk}
\author{Hendrik Fichtenberger\thanks{fichtenberger@google.com}}
\affil{Google Research, Zurich}
\author{Monika Henzinger\thanks{monika.henzinger@univie.ac.at}}
\affil{University of Vienna}
\author{Andreas Wiese}
\affil{Vrije Universiteit Amsterdam\thanks{a.wiese@vu.nl}}

\begin{document}
\title{On fully dynamic constant-factor approximation algorithms for clustering problems}

\maketitle
\thispagestyle{empty}

\begin{abstract}
Clustering is an important task with applications in many fields of
computer science. We study the fully dynamic setting in which we want
to maintain good clusters efficiently when input points (from a metric space) can be inserted
and deleted. Many clustering problems are $\mathsf{APX}$-hard but admit
polynomial time $O(1)$-approximation algorithms.
Thus, it is a natural question whether we can maintain $O(1)$-approximate
solutions 
for them %
in subpolynomial update time, against
adaptive and oblivious adversaries.
Only a few results are known that give partial answers to this question.
There are dynamic algorithms for $k$-center,
$k$-means, and $k$-median that maintain constant factor approximations
in expected $\tilde{O}(k^{2})$ update time against an oblivious adversary. However, 
for these problems 
there are no algorithms known with an update time that is subpolynomial in $k$, and
against an adaptive adversary there are even no (non-trivial) dynamic algorithms known at all.
Also, for the $k$-sum-of-radii and the $k$-sum-of-diameters problems 
it is open whether there is any dynamic algorithm, against either type of adversary.

In this paper, we complete the picture of the question above for all
these clustering problems. 
\begin{itemize}
\item We show that there is no fully dynamic $O(1)$-approximation algorithm
for \emph{any }of the classic clustering problems above with an update time
in $n^{o(1)}h(k)$ against an adaptive adversary, for an arbitrary 
function $h$. So in particular, there are no such deterministic algorithms. 
\item We give a lower bound of $\Omega(k)$ on the update time for
{each} of the above problems, even against an \emph{oblivious} adversary.
This rules out update times with subpolynomial dependence on $k$.
\item We give the first $O(1)$-approximate fully dynamic algorithms for
$k$-sum-of-radii and for $k$-sum-of-diameters with expected update time of $\tilde{O}(k^{O(1)})$  against an oblivious
adversary.
\item Finally, for $k$-center we present a fully dynamic $(6+\eps)$-approximation
algorithm with an expected update time of $\tilde{O}(k)$ against
an oblivious adversary. This is the first dynamic $O(1)$-approximation
algorithm for \emph{any} of the clustering problems above
whose update time is $\tilde{O}(k)$
and in particular the first whose update time is asymptotically optimal
in $k$. 
\end{itemize}

\end{abstract}
\newpage{}

\setcounter{page}{1}

\section{Introduction}

Clustering is a central task in data science and in many fields of
computer science such as machine learning, information retrieval,
computer vision, data compression, and resource allocation. The goal
is to partition a set of points into non-overlapping subsets so that
``close'' points belong to the same subset. To compare the quality
of different clusterings, each partition is assigned a \emph{cost}.
\emph{Static} and also \emph{online} clustering (where data points
arrive in an online manner and earlier decisions cannot be revoked)
are well-studied, e.g., \cite{hsu1979easy,DBLP:journals/tcs/Gonzalez85,gibson2010metric,DoddiMRTW00,AwasthiCKS15,DBLP:conf/stoc/JainMS02,ahmadiannsw_17,byrkaprst_17}.

However, in various settings \emph{dynamic} clustering is needed,
e.g.,~for dynamically evolving data set or interactive data analysis.
Research on dynamic clustering was already initiated in 1987~\cite{can1987dynamic},
but there are still many fundamental open questions. In \emph{fully
dynamic clustering} data points can be inserted as well as deleted
and the goal is to maintain a set of clusters whose cost is within
a (small) multiplicative factor of the optimal clustering for the
current set of points. Dynamic clustering algorithms %
process a sequence of update operations and 
they can answer two types of queries: in a \emph{value-query} an approximation
$\apx$ of the optimal clustering cost is returned, in a \emph{solution-query}
a solution is returned.
Such an algorithm is an $\alpha$-approximation
algorithm if it always holds that $\OPT\le\apx\le\alpha\OPT$ and each returned
solution has cost at most $\alpha\OPT$. It is evaluated based on
\begin{inenum}

\item its approximation ratio $\alpha$ and

\item its running time per operation\end{inenum}.

We study the following dynamic setting: at any point in time, there
is a set $P$ of active points. We use $n$ to denote the current
size of~$P$. In each update, the adversary adds or deletes a point.
The algorithm can query distances between pairs of points from the
adversary. The reported distances stem from a metric on the set of points consisting of $P$ and the
already deleted points, and an upper bound~$\Delta$ on the aspect
ratio of this metric is given to the algorithm at the beginning.
For a given $k$, the algorithm needs to maintain a set $C$ of $k$
centers.  For each point $p$ denote by $d(p,C)$ the distance of
$p$ to its closest center in $C$. All points that are assigned to
the same center $c\in C$ (i.e., for which $c$ is the closest center)
form a \emph{cluster}. 

There are different clustering problems with different cost functions
that one seeks to minimize. 
In the\emph{ $k$-center} \emph{problem},
we want to minimize the maximum distance of a point $p$ to its assigned
cluster center, i.e., we minimize the $L_{\infty}$-norm of the vector
$\left(d(p,C)\right)_{p\in P}$. In the \emph{$k$-median} and\emph{
$k$-means} \emph{problems} we instead minimize $\sum_{p\in S}d(p,C)$
and $\sum_{p\in S}d(p,C)^{2}$, respectively, which is equivalent
to minimizing the $L_{1}$- and the $L_{2}$-norms of $\left(d(p,C)\right)_{p\in P}$.
More generally, in the \emph{$(k,p)$-clustering problem} we minimize
$\sum_{p\in S}d(p,C)^{p}$ (corresponding to the $L_{p}$-norm). Another
way to interpret the $k$-center problem is that each cluster center
$c_{i}$ has a radius $r_{i}$ assigned to it, $c_{i}$ ``covers'' all
points within radius $r_{i}$, every point needs to be covered by a center, and we seek to minimize the maximum
of these radii, i.e., $\max_{i}r_{i}$. A natural variation of this
is the \emph{$k$-sum-of-radii problem} (also known as the \emph{$k$-cover
problem}) in which we assign again each center $c_{i}$ a radius $r_{i}$
but seek instead to minimize their sum $\sum_{i}r_{i}$. Related to
this is the \emph{$k$-sum-of-diameters} problem in which instead we
seek to minimize the sum of the diameters of the clusters. For $k$-center, $k$-median, and $k$-means we assume that the solution returned by an algorithm is the set of centers, for $k$-sum-of-radii the centers together with the radii, and for $k$-sum-of-diameter an actual partitioning of the points into clusters. %

All above problems are NP-hard,
and $k$-center, $k$-median, $k$-means, and $k$-sum-of-diameters are also known to be APX-hard~\cite{DoddiMRTW00,DBLP:journals/tcs/Gonzalez85,gibson2010metric,lee2017improved,jain2002new,feder1988optimal}. Hence,
we aim for $O(1)$-approximation algorithms for them. In the offline
setting, such algorithms with polynomial running time are known for
all problems above~\cite{DBLP:journals/tcs/Gonzalez85,hochbaum1985best,hochbaum1986unified,ahmadian2016approximation}. In this paper, we study clustering problems
in the dynamic setting. Of course, after each update we could recompute
the solution from scratch using a known offline algorithm, which
leads to an update time that is polynomial in $n$ and $k$. But can
we do better, i.e., achieve subpolynomial update time? For solution-queries the size of the output can be linear in $k$ (since all centers are returned), but for 
value-queries
no such obvious lower bound exists. In this paper,
we investigate the following basic question. 
\begin{center}
\emph{Can we maintain $O(1)$-approximate solutions for clustering
problems dynamically,\\
with update times that are subpolynomial in 
$n$ and/or $k$?}\\
 
\par\end{center}

For this we need to distinguish between an \emph{adaptive adversary
}that can choose each operation based on all the algorithm's query answers
so far (e.g., which points are centers) and an \emph{oblivious adversary
}that knows the algorithm, but does not see the actual answers
of the algorithm (which might depend on random bits). 
So far there
are only dynamic $O(1)$-approximation algorithms known for $k$-center,
$k$-median, and $k$-means with expected update times of $\tilde{O}(k^{2} / \epsilon^{O(1)})$
against an oblivious adversary~\cite{ChanGS18,DBLP:conf/esa/HenzingerK20}
(where $\tilde{O}$ suppresses factors that are polylogarithmic in $n$, $k$, and $\Delta$).
In particular, for none
of the problems listed above there is a dynamic algorithm known that
works against an adaptive adversary. Considering the recent progress on designing fast dynamic graph algorithms
against an \emph{adaptive adversary}~\cite{DBLP:conf/stoc/NanongkaiS17,DBLP:journals/corr/abs-2004-08432,DBLP:journals/corr/abs-2011-00977,DBLP:conf/stoc/Wajc20,DBLP:conf/soda/GutenbergW20a,DBLP:conf/icalp/EvaldFGW21}, it is an obvious questions whether such fast algorithms are also possible for dynamic clustering.
Furthermore, it is not known
whether their quadratic dependence on $k$ in the running time  against an oblivious adversary is
necessary, or whether we could improve it to, e.g., $k$ or $\log k$.
Moreover, for $k$-sum-of-radii and $k$-sum-of-diameters there is
even no non-trivial dynamic algorithm known at all (for either type of adversary). This leaves many open questions. 

\vspace{-0.5ex}
\subsection{Our results}
\vspace{-0.5ex}
In this paper we answer the above questions
for each of the classic clustering problems defined earlier. To this
end, we present novel upper and lower bounds for the remaining open
settings, see Table~\ref{tab:results} for an overview of all our
results.

\begin{table}
\begin{centering}
\begin{tabular}{ccc}
\toprule 
\textbf{Problem} & \textbf{Oblivious adversary} & \textbf{Adaptive adversary}\tabularnewline
\midrule 
$k$-center & \textcolor{red}{$\tilde{O}(k/
\eps)$ for $(6+\eps)$-approx.} & \textcolor{red}{$\tilde{O}(k)$ for $O(\min\{k,\log(n/k)\})$-approx.}\tabularnewline
 & \textcolor{red}{$\Omega(k)$ for $(\Delta-\eps)$-approx.}  & \textcolor{red}{$n^{\Omega(1)}h(k)$ for $O(1)$-approx.}\tabularnewline
 & $\tilde{O}(k^{2}/\eps)$ for $(2+\eps)$-approx.~\cite{ChanGS18} & \tabularnewline
\midrule 
$k$-median & \textcolor{red}{$\Omega(k)$ for $(\Delta-\eps)$-approx.} & \textcolor{red}{$n^{\Omega(1)}h(k)$ for $O(1)$-approx.}\tabularnewline 
 & $\tilde{O}(k^{2}/\eps^{O(1)})$ for $(5.3+\eps)$-approx.~\cite{DBLP:conf/esa/HenzingerK20} & \tabularnewline
\midrule 
$k$-means & \textcolor{red}{$\Omega(k)$ for $(\Delta^{2}-\eps)$-approx.} & \textcolor{red}{$n^{\Omega(1)}h(k)$ for $O(1)$-approx.}\tabularnewline 
 & $\tilde{O}(k^{2}/\eps^{O(1)})$ for $(36+\eps)$-approx.~\cite{DBLP:conf/esa/HenzingerK20} & \tabularnewline
\midrule 
$k$-sum-of-radii & \textcolor{red}{$k^{O(1)}\log\Delta$ for $O(1)$-approx.} & \textcolor{red}{$n^{\Omega(1)}h(k)$ for $O(1)$-approx.}\tabularnewline 
 & \textcolor{red}{$\Omega(k)$ for $(\Delta-\eps)$-approx.} & \tabularnewline
\midrule 
$k$-sum-of-diameters & \textcolor{red}{$k^{O(1)}\log\Delta$ for $O(1)$-approx.} & \textcolor{red}{$n^{\Omega(1)}h(k)$ for $O(1)$-approx.}\tabularnewline 
 & \textcolor{red}{$\Omega(k)$ for $(\Delta-\eps)$-approx.} & \tabularnewline
\midrule 
$(k,p)$-clustering & \textcolor{red}{$\Omega(k)$ for $(\Delta^{p}-\eps)$-approx.} & \textcolor{red}{$n^{\Omega(1)}h(k)$ for $O(1)$-approx.}\tabularnewline
\bottomrule
\end{tabular}
\par\end{centering}
\caption{\label{tab:results}An overview of our running time bounds (in red) in comparison
to previous work. Our lower bounds against an adaptive adversary hold
for an arbitrary function $h$, also if $k=1$, and if the algorithm can open $O(k)$ centers instead
of $k$. Apart from $k$-center, they hold even if the algorithm needs
to report only its current solution but not its (estimated) cost $\apx$.
All our lower bounds hold additionally in the case where the algorithm and the optimal solution are allowed
to select already deleted points as centers.}
\end{table}

\vspace{-2ex}
\paragraph{Lower bounds.}
\emph{Adaptive adversary.}
We present the first lower bounds for dynamic clustering problems
against an \emph{adaptive} adversary. Thus, they imply the same bounds
for deterministic algorithms against an arbitrary adversary. Our bounds
are very strong: they imply that for \emph{none }of the above problems
we can maintain $O(1)$-approximations with an update time that is
subpolynomial in $n$, even for the case that\emph{ $k=1$. }So in\emph{
}particular this holds independently of how the running time depends
on $k$\,! Also, our bounds hold even if the algorithm can open $O(k)$
centers instead of only $k$. Formally, we prove the following tradeoff:
\emph{any algorithm with an update time of $f(k,n)$ that opens up to $O(k)$
centers must have an approximation ratio (compared to the optimal
solution with only $k$ centers) of
$\Omega(\log n/(\log f(1,2n))$ for $k$-center, $k$-sum-of-radii, $k$-diameter,
of $\Omega(\log n/(10+2\log f(1,2n))$ 
and $k$-median,
of $\Omega((\log n/(12+2\log f(1,2n)))^{2})$ for 
$k$-means, and of 
$\Omega((\log n/(2p+8+2\log f(1,2n)))^{p})$ for 
$(k,p)$-clustering,}
see Theorem~\ref{thm:det-1lb}. Thus, to achieve an $O(1)$-approximation, the update time needs to be
$n^{\Omega(1)}h(k)$ for any function $h$. 

For $k$-center, if $k=1$, then one obtains a 2-approximation if one
outputs simply an arbitrary point; our lower bound above hinges on
the fact that the algorithm also needs to output an approximation
$\apx$ for its cost. However, even if we require the algorithm to
output only its current solution, we prove a lower bound of $\Omega(\min\{k,\log n/(k\log f(k,2n))\})$
on the approximation ratio
for any $k$. Therefore, for \emph{sufficiently large} $k$ we again need an update
time of $n^{\Omega(1)}h(k)$ in order to maintain a $O(1)$-approximation.
We show that our lower bound is (asymptotically) tight by giving a deterministic algorithm
with $O(k\log n\log\Delta/\eps)$ update time that achieves a $O(\min\{k,\log(n/k)\})$  %
approximation for arbitrary $k$. Our lower bounds for all other problems hold directly
also if the algorithm needs to output only its current solution.
\aw{Even more, all above bounds hold already under the weaker assumption that only the amortized number of distance queries of the algorithm
per operation is bounded by $f(k,n)$, rather than its running time.
}

\emph{Oblivious adversary.}
Also, we provide a lower bound of $\Omega(k)$ for the update time
of a $(\Delta-\eps)$-approximation algorithm for \emph{any }of the
above problems against an \emph{oblivious} adversary (and hence also against
an adaptive adversary). For $k$-means and $(k,p)$-clustering this
lower bound holds even for maintaining a $(\Delta^{2}-\eps)$-approximation
or a $(\Delta^{p}-\eps)$-approximation, respectively.

All our lower bounds hold in the setting where the algorithm can use
as centers only points in $P$, but also when additionally
it can also use any point that it has seen so far, including points that are already deleted. %

\vspace{-1ex}
\paragraph{Upper bounds.}
\emph{$k$-sum-of-radii and $k$-sum-of-diameters.}
As discussed above, the $k$-sum-of-radii and $k$-sum-of-diameters
problems are classical clustering problems for which no non-trivial
dynamic approximation algorithm is known. We fill this gap by giving
a fully dynamic $(13.008+\epsilon)$-approximation algorithm for $k$-sum-of-radii
with an expected update time of $k^{O(1/\eps)}\log\Delta$, which
implies a dynamic $(26.016+\epsilon)$-approximation algorithm for
$k$-sum-of-diameters with the same update time. This completes the picture
that all above clustering problems admit fully dynamic $O(1)$-approximation
algorithms with update time $\tilde{O}(k^{O(1)})$.

Our algorithms for $k$-sum-of-radii and $k$-sum-of-diameters use
a new idea in the setting of dynamic clustering algorithms: we maintain
a \emph{bi-criteria approximation}, i.e., an $O(1)$-approximate solution
that might open more than $k$ centers, up to $O(k/\epsilon)$ many.
After each update, we transform this larger set of centers to a solution with only $k$ centers,
losing only a constant factor in this step. For the last step, our
machinery allows to use any static offline $\alpha$-approximation
algorithm for $k$-sum-of-radii in a black-box manner, yielding a
dynamic $(6+2\alpha+\epsilon)$-approximation algorithm for the problem
(we use the algorithm due to Charikar and Panigrahy~\cite{CharikarP04}
for which $\alpha=3.504+\eps$). Inserting all points of a fixed set $P$ into our
data structure even leads to the fastest known \emph{static} $O(1)$-approximation algorithm  
for $k$-sum-of-radii.

\emph{$k$-center.}
We then demonstrate that our idea of using bi-criteria approximations has
more applications: we apply it to $k$-center for which we give a
fully dynamic $(6+\eps)$-approximation algorithm with an expected
amortized update time of $O(k\log^{2}n\log\Delta/\eps)=\tilde{O}(k/\epsilon)$.
This is the first dynamic $O(1)$-approximation algorithm for
any of the classic clustering problems above with an update time whose
dependence on $k$ is asymptotically optimal (while being polylogarithmic
in $n$ and $\Delta$). 
In fact, before not even a dynamic $O(1)$-approximation
algorithm with an update time in $\tilde{o}(k^{2})$ was known for
the simpler deletion-only case of $k$-center or any other of the
clustering problems above. %

Note that improving the dependency on $k$ to linear (even with an additional dependency on $\log n$) is beneficial in all settings where $k$ has a  poly-logarithmic or larger dependency on $n$. This is the case
in a variety of applications.
For
example, setting the number $k$ of cluster centers to be super-logarithmic
or even near-linear in the number of input points can improve the
quality of spam and abuse detection \cite{qian2010case,sheikhalishahi2015fast}
and near-duplicate detection \cite{hassanzadeh2009framework}. %

\vspace{-0.2cm}

\subsection{Technical overview\label{sec:technical-overview}}

\paragraph{Lower bounds against an adaptive adversary.}

To sketch the idea behind the lower bounds we discuss here a simplified setting in which the algorithm \begin{inenum}

\item can query only the distances between points that are currently
in $P$ (i.e., that have been introduced already but not yet removed)
and

\item has a worst-case update time of $f(k,n)$ (rather than amortized
update time) for some function $f$\end{inenum}. (These assumptions are not needed for the lower bound proof presented in the later part of the paper.) Our goal is to show
that the dynamic algorithm cannot maintain $O(1)$-approximate solutions
in subpolynomial update time. 

The adversary starts adding points into $P$ and maintains an auxiliary
graph $G=(V,E)$ with one vertex $v_{p}$ for each point $p$. Whenever
the algorithm queries the distance between two points $p,p'\in P$,
the adversary reports that they have a distance of 1 and adds an edge
$\{v_{p},v_{p'}\}$ of length 1 to $G$. Intuitively, the adversary
uses $G$ to keep track of the previously reported distances.
Whenever there is a vertex $v_{p}$ with a degree of at least $100f(k,n)$,
in the next operation the adversary deletes the corresponding point
$p$. There could be several such vertices, and then the adversary
deletes them one after the other. Thanks to assumption (i), once a point $p$ is deleted, the
degree of $v_{p}$ cannot increase further. Hence, the degrees of
the vertices in $G$ cannot grow arbitrarily. More precisely, one
can show that the degree of each vertex can grow to at most $O(\log n\cdot f(k,n))$.
Thus, for each vertex $v_{p}$ and each $\ell\in\mathbb{N}$,
there are at most $O(\log^{\ell}n\cdot f(k,n)^{\ell})$ vertices at
distance $\ell$ to $v_{p}$. In particular, at least half of the
vertices in $G$ are at distance at least $\Omega(\log_{O(\log n\cdot f(k,n))}n)=\Omega\left(\log n\;/\;[\,\log\log n\cdot\log(f(k,n))\,]\right)$
to $v_{p}$.

Now observe that the algorithm knows only the point distances that
it queried, i.e., those that correspond to edges in $G$. For all other distances
the triangle inequality imposes only an upper bound for any pairs of points whose corresponding
vertices are connected in $G$. Thus,
the algorithm cannot distinguish the setting where the underlying metric is
the shortest path metric in $G$, from the setting where all points
are at distance 1 to each other. If $k=1$ for any of the problems
under consideration, then for the selected center $c$ the algorithm
cannot distinguish whether all points are at distance 1 to $c$ or
if half of the points in $P$ are at distance $\Omega\left(\log n\;/\;[\,\log\log n\cdot\log(f(1,n))\,]\right)$
to $c$. Therefore, the reported cost $\apx$ can be larger than $\OPT$ by a factor of
up to $\Omega\left(\log n\;/\;[\,\log\log n\cdot\log(f(1,n))\,]\right)$.

We improve the above construction so that it works also for algorithms
with amortized update times, which are allowed to query
distances to points that are already deleted and which may output $O(k)$ instead of $k$ centers,
and also already if only the amortized number of distance queries per operation is bounded by $f(k,n)$. 
To this end, we adjust
the construction so that for points with degree of at least $100f(k,n)$
we report distances that might be larger than~1 and that still allow
us to use the same argumentation as above. At the same time, we remove
the factor of $\log\log n$ in the denominator. %

\vspace{-.2cm}

\paragraph{Lower bounds against an oblivious adversary.}

We sketch our lower bound of $\Omega(k)$ for the needed update time
in order to maintain solutions with an approximation ratio of $\Delta-\epsilon$
(and of $\Delta^{2}-\eps$ and $\Delta^{p}-\eps$ for $k$-means and
$(k,p)$-clustering, respectively). First, the adversary introduces
$k$ points $P_{0}$ that are at pairwise distance $\Delta$ to each
other. Then, the adversary introduces a point $p_{1}$ and queries
the cost of the current solution. The intuition is that the algorithm
needs $\Omega(k)$ queries in order to be able to distinguish whether
$p_{1}$ is close to some point in $P_{0}$, e.g., at distance~1 (in
which case the cost of $\OPT$ is 1) or whether $p_{1}$ is at distance
$\Delta$ from each point in $P_{0}$ (in which case the cost of $\OPT$
is $\Delta$, $\Delta^{2}$, or $\Delta^{p}$, depending on the problem).
Note that the algorithm needs to be able to distinguish these two
cases in order to achieve its approximation guarantees.

Then the adversary removes $p_{1}$, introduces another point $p_{2}$,
and repeats this process. Using this idea, we can generate a distribution over
inputs and then apply Yao's minmax principle to it. In this way we show that $\Omega(k)$
queries per update are needed, even for a randomized algorithm against
an oblivious adversary, in order to be able to distinguish the two
cases above for every point $p_{i}$ (and hence in order to obtain
non-trivial approximation ratios of $\Delta-\epsilon$, $\Delta^{2}-\epsilon$,
and $\Delta^{p}-\epsilon$, respectively).

\vspace{-.2cm}

\paragraph{Upper bounds via bi-criteria approximations.}

Our algorithms against an oblivious adversary for $k$-sum-of-radii,
$k$-sum-of-diameters, and $k$-center use the following paradigm:
we maintain bi-criteria approximations, i.e., solutions with a small
approximation ratio that might use more than $k$ centers (at most
$O(k/\epsilon)$ or $O(k\log^{2}n)$ centers, depending on the problem).
In a second step, we use the centers of this solution as the input
to an auxiliary dynamic instance for which we maintain a solution
with only $k$ centers. These centers then form our solution to the
actual problem, by increasing their radii appropriately. Since the
input to our auxiliary instance is much smaller than the input to
the original instance, we can afford to use algorithms for it whose
update times have a much higher dependence on $n$. Depending on the
desired overall update time, we can afford update times of $O(n)$
here, or even recompute the whole solution from scratch. 

In all our algorithms we first ``guess'' a $(1+\eps)$-approximate
estimate $\OPT'$ on the value $\OPT$ of the optimal solution. For the simplicity of presentation, we assume that $d(p,q) \geq 1$ for any $p,q$ in the metric space. If this is not the case, we simply scale all distances linearly. More formally,
for %
$O((\log\Delta)/\eps)$ values $\OPT'$ that are powers of $1+\eps$, we build a
dynamic clustering data structure that uses $\OPT'$ as an estimate of $\OPT$.
It has to either maintain a solution whose
cost is within a constant factor of $\OPT'$ or certify that $\OPT > \OPT'$.
We describe our approaches for $k$-sum-of-radii and $k$-center below in more detail.

\vspace{-.2cm}

\paragraph{$k$-sum-of-radii.}

Recall that in the static offline setting there is a $(3.504+\eps)$-approximation
algorithm with running time $n^{O(1/\eps)}$ for $k$-sum-of-radii~\cite{CharikarP04}.
We provide a black-box reduction that transforms any such offline
algorithm with running time of $g(n)$ and an approximation ratio of
$\alpha$ into a fully dynamic $(6+2\alpha+\eps)$-approximation algorithm
with an update time of $O(((k/\epsilon)^{5}+g(\mathrm{poly}(k/\epsilon)))\log\Delta)$.
To this end, we introduce a dynamic primal-dual algorithm that maintains
a bi-criteria approximation with up to $O(k/\epsilon)$ centers. After
each update, we use these $O(k/\epsilon)$ centers as input for the
offline $\alpha$-approximation algorithm, run it from scratch, and
increase the radii of the computed centers appropriately such that
they cover all points of the original input instance.

For computing the needed bi-criteria approximation, there is a polynomial-time
offline primal-dual $O(1)$-approximation algorithm with a running
time of $\Omega(n^{2}k)$~\cite{CharikarP04} from which we borrow
ideas. Note that a dynamic algorithm with an update time of $(k\log n)^{O(1)}$
yields an offline algorithm with a running time of $n(k\log n)^{O(1)}$
(by inserting all points one after another) and no offline algorithm
is known for the problem that is that fast. Hence, we need new ideas.
First, we formulate the problem as an LP $(P)$ which has a variable
$x_{p}^{(r)}$ for each combination of a point $p$ and a radius~$r$
from a set of (suitably discretized) radii $R$, and a constraint
for each point $p$. Let $(D)$ denote its dual LP, see below, where
$z:=\epsilon\OPT'/k$.
\begin{alignat*}{9}
\min & \,\,\, & \sum_{p\in P}\sum_{r\in R}x_{p}^{(r)}(r+z) &  &  &  &  &  &  &  & \max & \,\,\, & \sum_{p\in P}y_{p}\,\,\,\,\,\,\,\,\,\,\,\\
\mathrm{s.t.} &  & \sum_{p'\in P}\sum_{r:d(p,p')\le r}x_{p'}^{(r)} & \ge1 & \,\,\, & \forall p\in P & \,\,\,\,\,\,\,\: & (P) & \,\,\,\,\,\,\,\,\,\,\,\,\,\,\, &  & \mathrm{s.t.} &  & \sum_{p'\in P:d(p,p')\le r}y_{p'} & \le r+z & \,\,\, & \forall p\in P,r\in R & \,\,\,\,\,\,\,\:(D)\\
 &  & x_{p}^{(r)} & \ge0 &  & \forall p\in P\,\,\forall r\in R &  &  &  &  &  &  & y_{p} & \ge0 &  & \forall p\in P
\end{alignat*}

We select a point $c$ randomly and raise its dual variable $y_{c}$.
Unlike \cite{CharikarP04}, we raise $y_{c}$ only until the constraint
for~$c$ and some radius $r$ becomes \emph{half-tight, }i.e., $\sum_{p':d(c,p')\le r}y_{p'}=r/2+z$.\emph{
}We show that then we can guarantee that no other dual constraint
is violated, which saves running time since we need to check only
the (few, i.e. $|R|$ many) dual constraints for~$c$, and not the constraints for
all other input points when we raise $y_{c}$. We add $c$ to our
solution and assign it a radius of~$2r$. Since the constraint for
$(c,r)$ is half-tight, our dual can still ``pay'' for including $c$
with radius $2r$ in our solution. %

In primal-dual algorithms, one often raises all primal variables whose
constraints have become tight, in particular in the corresponding
routine in~\cite{CharikarP04}. However, since we assign to $c$
a radius of $2r$, we argue that we do not need to raise the up to
$\Omega(n)$ many primal variables corresponding to dual constraints
that have become (half-)tight. Again, this saves a considerable amount
of running time. Then, we consider the points that are not covered
by $c$ (with radius $2r$), select one of these points uniformly
at random, and iterate. After a suitable pruning routine (which \aw{is
faster than the} corresponding routine in~\cite{CharikarP04}) this
process opens $k'= O(k/\eps)$ centers at cost $(6+\eps)\OPT'$,
or asserts that $\OPT>\OPT'$.

We maintain the above primal-dual solution dynamically. The most interesting
update operation occurs when a point $p$ is deleted whose dual variable
$y_{p}$ was raised at some point. Then, we ``warm-start'' the algorithm,
starting at the point where $y_{p}$ was raised, and select instead
another point $p'$ randomly to increase its dual variable $y_{p'}$.
Suppose that $U_{p}$ are the set of points from which $p$ was selected
randomly. The work of the warm-start is at most $O(|U_{p}|g(k'))$.
In expectation the adversary needs to delete $\Omega(|U_{p}|)$ points
of $U_{p}$ before deleting $p$. Thus, we charge the cost of the
warm-start to these points, resulting in a $g(k')$ time per operation.

As mentioned above, after each update we feed the centers of the bi-criteria
approximation as input points to our (arbitrary) static offline $\alpha$-approximation
algorithm for $k$-sum-of-radii and run it. Finally, we translate
its solution to a $(6+2\alpha+\eps)$-approximate solution to the original
instance. For the best known offline approximation algorithm \cite{CharikarP04}
it holds that $\alpha=3.504$, and thus we obtain a ratio of $13.008+\eps$
overall. For $k$-sum-of-diameters the same solution yields a $(26.016+2\epsilon)$-approximation.
\vspace{-.2cm}

\paragraph*{$k$-center.}

In our $(6+\eps)$-approximation algorithm for $k$-center, we maintain
a bi-criteria solution with cost at most $4\OPT'$ that uses up to
$O(k\log^{2}n)$ centers. We feed its centers into an auxiliary data
structure with an update time of $O(n+k)$ that maintains a solution
with cost at most $2\OPT'$; so unlike for $k$-sum-of-radii we do not simply recompute the
solution from scratch. Note that our update time of $O(n+k)$ is still
faster than recomputing a 2-approximate solution from scratch with
the fastest known algorithm, which needs $O(nk)$ time~\cite{hsu1979easy}.
We argue that each input point is at distance at most $6\OPT'$ from
a point computed by the auxiliary data structure, which yields a $(6+\eps)$-approximation. 

We describe first the auxiliary data structure with update time $O(n+k)$.
We say that a center $c$ \emph{covers} a point $p$ if $d(c,p)\le2\OPT'$.
We maintain a counter $a_{p}$ for each point $p$ of the number of
current centers that cover $p$. If the counter $a_{p}$ becomes zero,
then $p$ is defined as a new center and, hence, we increment the
counters of each point $p'$ that $p$ covers. This yields $O(n)$
increments. If a center $c$ is deleted, we decrement the counters
of all points that are covered by~$c$, which yields $O(n)$ decrements
and check for each of these points whether it has to become a center
now. As $n$ is the current size of $P$, if a point that is inserted
and deleted when $n$ is small, it cannot ``pay'' for the cost of
$O(n)$ if it becomes a center when $n$ is large. Instead
we use an amortized analysis that charges each update operation $\Theta(n+k)$
for this cost. Our centers are at distance at least $2\OPT'$ to each
other. Hence, if at least $k+1$ centers exist, they assert that $\OPT>\OPT'$.
If there are at most $k$ centers, they cover all points and, thus,
represent a valid 2-approximation.

Next, we describe a deletion-only data structure maintaining a solution
with cost at most $4\OPT'$ which is %
allowed to use $O(k\log n)$ centers, instead of only $k$. We will
use it as a subroutine to obtain a fully dynamic data structure that
uses $O(k\log^{2}n)$ centers. Already for the deletion-only case,
no data structure was known that maintains a $O(1)$-approximation
with update time $k(\log n)^{O(1)}$. Similarly as Chan et al.~\cite{ChanGS18},
we initialize our data structure by selecting the first center $c_{1}$
uniformly at random among all points in $P=:U_{1}$ and assign all
points to $c_{1}$ that are covered by~$c_{1}$. We say that they
form a \emph{cluster} with center $c_{1}$. We define $U_{2}$ to
be the remaining points and continue similarly by selecting a point
$c_{2}$ uniformly at random from $U_{2}$, etc. In contrast to \cite{ChanGS18},
we assign the computed centers into \emph{buckets} such that for all
points $c_{i}$ in the same bucket, the size of the corresponding
set $U_{i}$ differs by at most a factor of 2. We stop if we have
covered all points or if a bucket $B^{*}$ has $2k$ centers; hence,
$B^{*}$ must contain the last $2k$ chosen centers. This defines
$O(k\log n)$ centers since there can be only $\log n$ buckets.

If a point $p$ is deleted and $p$ is not a center, we do not change
our centers. If a center $c$ is deleted, we replace $c$ by an arbitrary
point $p$ in the cluster of $c$. This update can be done in time
$O(1)$, and then $p$ becomes the center corresponding to this cluster.
Since the diameter of the cluster is at most $4\OPT'$, $p$ covers
all points in the cluster if we assign it a radius of $4\OPT'$. In
particular, this is still the case if later we replace $p$ by some
other point in the same cluster in the same fashion.

As long as the bucket $B^{*}$ contains at least $k+1$ of its initially
chosen centers, this asserts that $\OPT>\OPT'$. Once $k$ of the
original centers from $B^{*}$ are deleted, we ``warm-start'' the
above computation at the point where the first center $c_{i^{*}}\in B^{*}$
was chosen, discarding all centers from $c_{i^{*}}$ onwards and recomputing
them with the current set $P$. Since each center $c_{i}$ was chosen
uniformly at random from its respective set $U_{i}$, in expectation
the adversary needs to delete $\Omega(|U_{i}|)$ points from $U_{i}$
before deleting $c_{i}$. Since we deleted $k$ centers from $B^{*}$
and the sizes of their respective sets $U_{i}$ differ by at most
a factor of $2$, in expectation the adversary must delete $\Omega(k|U_{i^{*}}|)$
points before the warm-start. The set $U_{i^{*}}$ contains all points
that are assigned to some center in~$B^{*}$. Hence, we can charge
the cost of the warm-start of $O(k|U_{i^{*}}|)$ to these deletions.

Using a technique from Overmars~\cite{DBLP:books/sp/Overmars83}
we turn our deletion-only data structure into a fully dynamic one,
at the expense of increasing the number of centers by a factor of
$O(\log n)$. This yields a set $C$ with $O(k\log^{2}n)$ centers
and a cost of at most $4\OPT'$. As mentioned above, we feed $C$
into the auxiliary data structure with $O(n+k)$ update time. This
yields a solution with cost at most $6\OPT'\le(6+\epsilon)\OPT$-approximation
overall. As one update in $P$ might lead to many changes in $C$
this requires a careful amortized analysis to achieve an update time
$O(k\log^{2}n\log\Delta)$. 

Related work and all missing proofs can be found in the appendix.

\section{Bounds against adaptive adversaries}

In this section we present our upper and lower bounds on the approximation
guarantees of dynamic algorithms for $k$-center, $k$-median, $k$-means,
$k$-sum-of-radii, $k$-sum-of-diameters, and $(k,p)$-clustering against
an \emph{adaptive }adversary. We start with the lower bounds. First,
we define a generic strategy for an adversary that in each operation
creates or deletes a point and answers the distance queries of the
algorithm. Then we show how to derive lower bounds for all of our
problems from this single strategy. Finally, we show that our lower
bounds for $k$-center are asymptotically tight by giving a matching
upper bound.

\vspace{-0.2cm}

\paragraph*{Strategy of the adaptive adversary.}

Let $f(k,n)$ be a positive function that is, for every fixed $k$, non-decreasing in $n$. Suppose that there is an algorithm (for any of the problems under
consideration) that in an amortized sense queries the distances
between at most $f(k,n)$ pairs of points per operation of the adversary,
where $n$ denotes the number of points at the beginning of the respective
operation. Note that any algorithm with an amortized update time
of at most $f(k,n)$ fulfills this condition. To determine the distance
between two points the algorithm asks a distance query to the adversary. 
We present now an adversary ${\cal A}$ whose goal
is to maximize the approximation ratio of the algorithm. To
record past answers and to give consistent answers, ${\cal A}$ maintains
a graph $G=(V,E)$ which contains a vertex $v_{p}\in V$ for each
point $p$ that has been inserted previously (including points that
have been deleted already). Intuitively, with the edges in $E$
the adversary keeps track of previous answers to distance queries.
Each vertex $v_{p}$ is labeled as \emph{active, passive, }or \emph{off}.
If a point $p$ has not been deleted yet, then its vertex $v_{p}$
is labeled as active or passive. Once point $p$ is deleted, then
$v_{p}$ is labeled as off. Intuitively, if $p$ has not been deleted
yet, then $v_{p}$ is active if it has small degree and passive if
it has large degree. In the latter case, ${\cal A}$ will delete $p$
soon. All edges in $G$ have length 1, and for two vertices $v,v'\in V$
we denote by $d_{G}(v,v')$ their distance in $G$.

When choosing the next update operation, ${\cal A}$ checks whether there
is a passive vertex $v_{p}$. If yes, ${\cal A}$ picks an arbitrary
passive vertex $v_{p}$, deletes the corresponding point $p$, and
labels the vertex $v_{p}$ as off. Otherwise, it adds a new point
$p$, adds a corresponding vertex $v_{p}$ to $G$, and labels $v_{p}$
as active.

Suppose now that the algorithm queries the distance $d(p,p')$ for
two points $p,p'$ while processing an operation. Note that $p$ and/or
$p'$ might have been deleted already. If both $v_{p}$ and $v_{p'}$
are active then ${\cal A}$ reports to the algorithm that $d(p,p')=1$
and adds an edge $\{v_{p},v_{p'}\}$ to $E$. Intuitively, due to
the edge $\{v_{p},v_{p'}\}$ the adversary remembers that it reported
the distance $d(p,p')=1$ before and ensures that in the future it
will report distances consistently. Otherwise, ${\cal A}$ considers
an augmented graph $G'$ which consists of $G$ and has
in addition an edge $\{v_{\bar{p}},v_{\bar{p}'}\}$ of length 1
between any pair of active vertices $\bar{p},\bar{p}'$. The adversary
computes the shortest path $P$ between $p$ and $p'$ in $G'$ and
reports that $d(p,p')$ equals the length of $P$. If $P$ uses an
edge between two active vertices $\bar{p},\bar{p}'$, then ${\cal A}$
adds the edge $\{v_{\bar{p}},v_{\bar{p}'}\}$ to $G$. Note that $P$
can contain at most one edge between two active vertices since it
is a shortest path in a graph in which all pairs of active vertices
have distance~1. %
Observe that if both $v_{p}$ and $v_{p'}$ are active then this procedure
reports that $d(p,p')=1$ and adds an edge $\{v_{p},v_{p'}\}$ to
$E$ which is consistent with our definition above for this
case. If a vertex $v_{p}$ has degree at least $100f(k,t)$ for the current operation $t$, then $v_{p}$ is labeled as passive. A passive vertex never becomes active again.

In the next lemma we prove some properties about this strategy of
${\cal A}$. For each operation $t$, denote by $G_{t}=(V_{t},E_{t})$
the graph $G$ at the beginning of the operation $t$. 
Recall that the value of $n$ right before the operation $t$ (which is  the number of
current points) equals the number of active and passive
vertices in $G_{t}$.
We show that
the the number of active vertices is $\Theta(t)$, each vertex has
bounded degree, and there exist arbitrarily large values $t$
such that in $G_{t}$ there are no passive vertices (i.e., only active
and off vertices). \begin{lemma} \label{lem:mpt-det-low-adv-prop}
For every operation $t>0$ the strategy of the adversary ensures the
following properties for $G_{t}$
\begin{enumerate}
\item the number of active vertices in $G_{t}$ is at least $96t/100$,
\label{it::mpt-det-low-adv-prop-numac} 
\item each vertex in $G_{t}$ has a degree of at most $100f(k,n)$, \label{it::mpt-det-low-adv-prop-deg} 
\item there exists an operation $t'$ with $t < t' \leq 2t$ such that $G_{t'}$
contains only active and off vertices, but no passive vertices. \label{it::mpt-det-low-adv-prop-clean} 
\end{enumerate}
\end{lemma} We say that an operation $t\in\N$ is a \emph{clean operation}
if in $G_{t}$ there are no passive vertices. For any clean operation
$t$, denote by $\bar{G}_{t}=(\bar{V}_{t},\bar{E}_{t})$ the subgraph
of $G_{t}$ induced by the active vertices in $V_{t}$.

\vspace{-0.2cm}

\paragraph*{Consistent metrics. }

The algorithm does not necessarily know the complete metric of the
given points, it knows only the distances reported by the adversary.
In particular, there might be many possible metrics that are consistent
with the reported distances. For each $t\in\N$ denote by $Q_{t}$
the points that were inserted before operation $t$, including
all points that were deleted before operation $t$, and let
$P_{t}\subseteq Q_{t}$ denote the points in $Q_{t}$ that are not
deleted. Given a metric $M$ on any point set $P'$, for all pairs
of points $p,p'\in P'$ we denote by $d_{M}(p,p')\ge0$ the distance
between $p$ and $p'$ according to $M$. For any $t\in\N$ we say
that a metric $M$ for the point set $Q_{t}$ is \emph{consistent}
if for any pair of points $p,p'\in Q_{t}$ for which the adversary
reported the distance $d(p,p')$ before operation~$t$, it
holds that $d(p,p')=d_{M}(p,p')$. In particular, any consistent metric
might be the true underlying metric for the point set $Q_{t}$.

The key insight is that for each clean operation $t$, we can build
a consistent metric with the following procedure. Take the graph $G_{t}$
and insert an arbitrary set of edges of length $1$ between pairs
of active vertices (but no edges that are incident to off vertices),
and let $G'_{t}$ denote the resulting graph. Let $M$ be the shortest
path metric according to $G'_{t}$. If a metric $M$ for $Q_{t}$
is constructed in this way, we say that $M$ is an \emph{augmented
graph metric for }$t$.

\begin{lemma} \label{lem:metric-consistent}Let $t\in\N$ be a clean
operation and let $M$ be an augmented graph metric for $t$. Then
$M$ is consistent. 
\end{lemma} 
In particular, there are no shortcuts
via off vertices in $G_{t}$ that could make the metric $M$
inconsistent.

We fix a clean operation $t\in\N$. We define some metrics that are consistent
with $Q_{t}$ that we will use later for the lower bounds for our
specific problems. The first one is the ``uniform'' metric $M_{\mathrm{uni}}$
that we obtain by adding to $G_{t}$ an edge between \emph{each} pair
of active vertices in $G_{t}$. As a result, $d_{M_{\mathrm{uni}}}(p,p')=1$
for any $p,p'\in P_{t}$. 
\begin{lemma} \label{lem:muni-consistent}
For each clean operation $t$ the corresponding metric $M_{\mathrm{uni}}$
is consistent. 
\end{lemma} 
In contrast to $M_{\mathrm{uni}}$, our
next metric ensures that there are distances of up to $\Omega(\log n)$
between some pairs of points. Let $p^{*}\in P_{t}$ be a point such
that $v_{p^{*}}$ is active. For each $i\in\N$ let $V^{(i)}\subseteq V_{t}$
denote the active vertices $v\in V_{t}$ with $d_{G_{t}}(v_{p^{*}},v)=i$,
and let $V^{(n)}\subseteq V_{t}$ denote the vertices in $G_{t}$
that are in a different connected component than $p^{*}$. Since
the vertices in $G_{t}$ have degree at most $100 f(k,n)$, more than half of all
vertices are in sets $V^{(i)}$ with $i\ge\Omega(\log n/ \log f(k,n))$.

We define now a metric $M(p^{*})$ as the shortest path metric in
the graph defined as follows. We start with~$G_{t}$; for each $i,i'\in\N$
we add to $G_{t}$ an edge $\{v_{p},v_{p'}\}$ between any pair of
vertices $v_{p}\in V^{(i)}$, $v_{p'}\in V^{(i')}$ such that $|i-i'|\le1$.
As a result, for any $i,i'\in\N$ and any $v_{p}\in V^{(i)}$, $v_{p'}\in V^{(i')}$
we have that $d_{M(p^{*})}(p,p')=\max\{|i-i'|,1\}$, i.e., $d_{M(p^{*})}(p,p')=1$
if $i=i'$ and $d_{M(p^{*})}(p,p')=|i-i'|$ otherwise. 
\begin{lemma}
\label{lem:mstar-consistent} For each clean operation $t$ and each
point $p^{*}\in V_{t}$ the metric $M(p^{*})$ is consistent. 
\end{lemma}
For any two thresholds $\ell_{1},\ell_{2}\in\N_{0}$ with $\ell_1 < \ell_2$  we define a metric
$M_{\ell_{1},\ell_{2}}(p^{*})$ (which is a variation of $M(p^{*})$)
as the shortest path metric in the following graph. Intuitively,
we group the vertices in $\bigcup_{i=0}^{\ell_{1}}V^{(i)}$ to one
large group and similarly the vertices in $\bigcup_{i=\ell_{2}}^{\infty}V^{(i)}$.
Formally, in addition to the edges defined for $M(p^{*})$, for each
pair of vertices $v_{p}\in V^{(i)}$, $v_{p'}\in V^{(i')}$ we add
an edge $\{v_{p},v_{p'}\}$ if $i\le i'\le\ell_{1}$ or $\ell_{2}\le i\le i'$.%

\begin{lemma} \label{lem:mstarrange-consistent} For each clean operation
$t$, each point $p^{*}\in V_{t}$, and each $\ell_{1},\ell_{2}\in\N_{0}$
the metric $M_{\ell_{1},\ell_{2}}(p^{*})$ is consistent. \end{lemma}

\paragraph*{Lower bounds. }

Consider a clean operation $t$. The algorithm cannot distinguish
between $M_{\mathrm{uni}}$ and $M(p^{*})$ for any $p^{*}\in P_{t}$.
In particular, for the case that $k=1$ (for any of our clustering
problems) the algorithm selects a point $p^{*}$ as the center, and
then for each $i$ it cannot determine whether the distance
of the points in $V^{(i)}$ to $p^{*}$ equals 1 or $i$.
However, there are at least $n/2$ points in sets $V^{(i)}$
with $i\ge\Omega(\log n/\log f(k,n)))$ and hence they contribute a large amount
to the objective function value. This yields the following lower bounds,
already for the case that $k=1$. With more effort, we can show them even for bi-criteria approximations,
i.e., for algorithms that may output $O(k)$ centers, but where the approximation
ratio is still calculated with respect to the optimal cost on $k$
centers. \begin{theorem} \label{thm:det-1lb} \aw{Against an adaptive
adversary}, any dynamic algorithm that queries amortized at most $f(k,n)$
distances per operation and outputs at most $O(k)$ centers 
\begin{itemize}
\item for $k$-center, $k$-sum-of-radii or $k$-sum-of-diameters has an
approximation ratio of $\Omega(\frac{\log(n)}{\log f(1,2n)})$, 
\item for $k$-median has an approximation ratio of $\Omega(\frac{\log(n)}{10+2\log f(1,2n)})$, 
\item for $k$-means has an approximation ratio of $\Omega((\frac{\log n}{12+2\log f(1,2n)})^{2})$, 
\item for $(k,p)$-clustering has an approximation ratio of $\Omega((\frac{\log n}{2p+8+2\log f(1,2n)})^{p})$. 
\end{itemize}
\end{theorem}

For $k$-center the situation changes
if the algorithm does
not need to be able to report an upper bound of the value of its computed
solution (but only the solution itself), since if $k=1$, then any
point is a 2-approximation. However, for arbitrary $k$ we can argue
that there must be $3k$ consecutive sets $V^{(i)},V^{(i+1)},...,V^{(i+3k-1)}$
such that the algorithm does not place any center on any point corresponding
to the vertices in these sets and hence incurs a cost of at least
$3k/2$. On the other hand, for the metric $M_{i,i+3k-1}(p^{*})$ the
optimal solution selects one center from each set $V^{(i+1)},V^{(i+4)},V^{(i+7)},...$
which yields a cost of only 1. %
\begin{theorem} \label{lem:mpt-det-low-kcen} \aw{Against an adaptive
adversary}, any dynamic algorithm for $k$-center that queries
amortized at most $f(k,n)$ distances per operation
and outputs at most $O(k)$ centers
has an approximation factor
of at least $\Omega\left(\min\left\{ k,\frac{\log n}{k\log f(k,2n)}\right\} \right)$,
even if at query time it needs to output only a center set, but no
estimate $\apx$ on $\OPT$.\end{theorem}

We show that our lower bound for $k$-center is asymptotically tight
\aw{by giving an algorithm that matches these bounds.
Our algorithm is deterministic, and hence its guarantees hold against an adaptive adversary.}
\begin{theorem}[see \cref{sec:app-det-upper}]
\label{thm:mpt-det-upper} There is a dynamic deterministic algorithm
for $k$-center with update time of $O(k\log n\log\Delta/\eps)$ and
an approximation factor of $(1+\eps)\min\{4k,4\log(n/k)\}$. \end{theorem}%

\section{Algorithms for $k$-sum-of-radii and $k$-sum-of-diameters}

\label{sec:primal-dual}

In this section we present our (randomized) dynamic $(13.008+\epsilon)$-approximation
algorithm for $k$-sum-of-radii with an amortized update time of $k^{O(1/\epsilon)}\log\Delta$,
against an oblivious adversary. Our strategy is to maintain a 
bi-criteria approximation with $O(k/\epsilon)$ clusters whose sum of radii is at most $(6+\epsilon)\OPT$
(and which covers all input points). We show how to use an arbitrary
offline $\alpha$-approximation algorithm to turn this solution into
a $(6+2\alpha+\eps)$-approximate solution with only $k$ clusters.
Using the algorithm in \cite{CharikarP04} (for which $\alpha=3.504+\eps$),
this yields a dynamic $(13.008+\epsilon)$-approximation for $k$-sum-of-radii,
and hence a $(26.015+\eps)$-approximation for $k$-sum-of-diameters. %

Assume that we are given an $\epsilon>0$ such that w.l.o.g.~it holds
that $1/\epsilon\in\N$. We maintain one data structure for each value
$\OPT'$ that is a power of $1+\epsilon$ in $[1,\Delta]$, i.e.,
$O(\log_{1+\epsilon}\Delta)$ many. The data structure for each such
value $\OPT'$ outputs a solution of cost at most $(13.008+\epsilon)\OPT'$
or asserts that $\OPT>\OPT'$. We output the solution with smallest
cost which is hence a $(13.008+O(\epsilon))$-approximation.
We describe first how to maintain the mentioned bi-criteria approximation.
We define $z:=\epsilon\OPT'/k$. Our strategy is to maintain a solution
for an auxiliary problem based on a Lagrangian relaxation-type approach.
More specifically, we are allowed to select an arbitrarily
large number of clusters, however, for each cluster we need to pay
a fixed cost of $z$ plus the radius of the cluster.
For the radii we allow only integral multiples of $z$ that are bounded
by $\OPT'$, i.e., only radii in the set $R=\{z,2z,...,\OPT'-z,\OPT'\}$.

This problem can be modeled by an integer program whose LP-relaxation
is given in Section~\ref{sec:technical-overview}; we denote it by $(P)$. We describe first an offline
primal-dual algorithm that computes an integral solution to~$(P)$ and
a fractional solution to the dual $(D)$, such that their costs differ
by at most a factor $6$. By weak duality, this implies that our solution
for $(P)$ is a $6$-approximation. 
We initialize $x\equiv0$ and $y\equiv0$, define $U_{1}:=P$, and
we say that a pair $(p',r')$ \emph{covers a point }$p$ if $d(p',p)\le r'$.
Our algorithm works in iterations. %

At the beginning of the $i$-th iteration, assume that we are given
a set of points $U_{i}$, a vector $\big(x_{p}^{(r)}\big)_{p\in P,r\in R}$,
and a vector $\left(y_{p}\right)_{p\in P}$, such that $U_{i}$ contains
all points in $P$ that are not covered by any pair $(p,r)$ for which
$x_{p}^{(r)}=1$ (which is clearly the case for $i=1$ since $U_{1}=P$
and $x\equiv0$). We select a point $p\in U_{i}$ uniformly at random
among all points in $U_{i}$. We raise its dual variable $y_{p}$
until there is a value $r\in R$ such that the dual constraint for
$(p,r)$ becomes \emph{half-tight, }meaning that %
$\sum_{p':d(p,p')\le r}y_{p'}=r/2+z.$ %

Note that it might be that the constraint is already at least half-tight
at the beginning of iteration $i$ in which case we do not raise $y_{p}$
but still perform the following operations. Assume that $r$ is the
largest value for which the constraint for $(p,r)$ is at least half-tight.
We define $x_{p}^{(2r)}:=1$, $(p_{i},r_{i}):=(p,2r)$ and set $U_{i+1}$
to be all points in $U_{i}$ that are not covered by $(p,2r)$. 
\begin{lemma}
\label{lem:pd-iteration-time} \aw{Each iteration $i$ needs a running
time of $O(\lvert U_{i}\rvert+ik/\epsilon)$}. \end{lemma} 
We stop
when in some iteration $i^{*}$ it holds that $U_{i^{*}}=\emptyset$
or if we completed the $(2k/\epsilon)^{2}$-th iteration and $U_{(2k/\epsilon)^{2}+1}\ne\emptyset$.
Suppose that $x$ and $y$ are the primal and dual vectors after the
last iteration. In case that $U_{(2k/\epsilon)^{2}}\ne\emptyset$
we can guarantee that our dual solution has a value of more than $\OPT'$
from which we can conclude that $\OPT>\OPT'$; therefore, we stop
the computation for the estimated value $\OPT'$ in this case. 
\begin{lemma}[{restate=[name=]lemPdIterationsBound}]
\label{lem:pd-iterations-bound}
If $U_{(2k/\epsilon)^2 + (2k/\epsilon)}\ne\emptyset$ then $\OPT>\OPT'$. 
\end{lemma}
If $U_{(2k/\epsilon)^{2}}=\emptyset$
we perform a pruning step in order to transform $x$ into a
solution whose cost is at most by a factor $3$ larger than the cost
of $y$. We initialize $\bar{S}:=\emptyset$.
Let $S=\{(p_{1},r_{1}),(p_{2},r_{2}),...\}$ denote the set
of pairs $(p,r)$ with $x_{p}^{(r)}=1$. We sort the pairs in $S$
non-increasingly by their respective radius $r$. Consider a pair
$(p_{j},r_{j})$. We insert the cluster $(p_{j},3r_{j})$ in our solution $\bar{S}$
and delete from $S$ all pairs $(p_{j'},r_{j'})$ such that $j'>j$
and $d(p_{j},p_{j'})<r_{j}+r_{j'}$. Note that $(p_{j},3r_{j})$ covers
all points that are covered by any deleted pair $(p_{j'},r_{j'})$
due to our ordering of the pairs. Let $\bar{S}$ denote
the solution obained in this way and let $\bar{x}$ denote the corresponding
solution to $(P)$, i.e., $\bar{x}_{p}^{(r)}=1$ if and only if $(p,r)\in\bar{S}$.
We will show below that $\bar{S}$ is a feasible solution to~$(P)$
with at most $O(k/\epsilon)$ clusters. Let $\bar{C}\subseteq P$
denote their centers, i.e., $\bar{C}=\{p\mid\exists r\in R:(p,r)\in\bar{S}\}$.
\begin{lemma} \label{lem:mpt-pd-pruning-time} Given $x$ we can
compute $\bar{x}$ in time $O((k/\epsilon)^{4})$ and $\bar{x}$ selects
at most $O(k/\epsilon)$ centers. %
\end{lemma} 
We transform now the bi-criteria approximate solution $\bar{x}$ into a feasible solution
$\tilde{x}$ with only $k$ clusters. To this end, we invoke the offline
$(3.504+\epsilon)$-approximation algorithm from \cite{CharikarP04}
on the input points $\bar{C}$. Let $\hat{S}$ denote the set of pairs
$(\hat{p},\hat{r})$ that it outputs (and note that not necessarily
$\hat{r}\in R$ since we use the algorithm as a black-box). Note that
$\hat{S}$ covers only the points in $\bar{C}$, and not necessarily
all points in $P$. 
\aw{On the other hand, the solution $\hat{S}$ has a cost of at most $2\OPT$ since
we can always find a solution with this cost covering $\bar{C}$, even if we are only allowed
to select centers from $\bar{C}$.
}
Thus, based on $\bar{S}$ and $\hat{S}$ we compute
a solution $\tilde{S}$ with at most $k$ clusters that covers $P$.
We initialize $\tilde{S}:=\emptyset$. For each pair $(\hat{p},\hat{r})\in\hat{S}$
we consider the points $\bar{C}'$ in $\bar{C}$ that are covered
by $(\hat{p},\hat{r})$. Among all these points, let $\bar{p}$ be
the point with maximum radius $\bar{r}$ such that $(\bar{p},\bar{r})\in\bar{S}$.
We add to $\tilde{S}$ the pair $(\hat{p},\hat{r}+\bar{r})$ and remove
$\bar{C}'$ from $\bar{C}$. We do this operation with each pair $(\hat{p},\hat{r})\in\hat{S}$.
Let $\tilde{S}$ denote the resulting set of pairs. 
\begin{lemma}[{restate=[name=]lemPdOfflineTime}]
\label{lem:pd-offline-time}
Given $\bar{S}$ and $\hat{S}$ we can compute $\tilde{S}$ in time
$O(k^3/\epsilon^2)$. 
\end{lemma}

We show that $\tilde{S}$ is a feasible solution with small cost.
We start by bounding the cost of $\bar{x}$ via $y$. 
\begin{lemma}[{restate=[name=]lemPdCostA}]
\label{lem:pd-cost-a}
We have that $\bar{x}$ and $y$ are feasible solutions to (P) and
(D), respectively, for which we have that 
$
\sum_{p\in P}\sum_{r\in R}\bar{x}_{p}^{(r)}(r+z)\le 6\cdot\sum_{p\in P}y_{p}\le6\OPT'.
$
\end{lemma}
 Next, we argue that $\tilde{S}$ is feasible and bound
its cost by the cost of $S'$ and the cost of $\bar{x}$. 
\begin{lemma}[{restate=[name=]lemPdCostB}]
\label{lem:pd-cost-b}
We have that $\tilde{S}$ is a feasible solution with cost at most
$\sum_{(\hat{p},\hat{r})\in \hat{S}}\hat{r}+\sum_{p\in P}\sum_{r\in R}\bar{x}_{p}^{(r)}(r+z)\le(13.008+\epsilon)\OPT'$. 
\end{lemma}

\subsection{Dynamic algorithm}

We describe now how we maintain the solutions $x,\bar{x},y,S,\bar{S}$,
and $\tilde{S}$ dynamically when points are inserted or deleted.
Our strategy is similar to \cite{ChanGS18}.

Suppose that a point $p$ is inserted. For each $i\in\{1,...,2k/\epsilon+1\}$
we insert $p$ into the set $U_{i}$ if $U_{i}\ne\emptyset$ and $p$
is not covered by a pair $(p_{j},r_{j})$ with $j\in\{1,...,i-1\}$.
If there is an index $i\in\{1,...,2k/\epsilon+1\}$ such that $U_{i-1}\ne\emptyset$
(assume again that $U_{0}\ne\emptyset$), $U_{i}=\emptyset$, and
$p$ is not covered by any pair $(p_{j},r_{j})$ with $j\in\{1,...,i-1\}$,
we start the above algorithm in the iteration $i$, being initialized
with $U_{i}=\{p\}$ and the solutions $x,y$ as being computed previously.

Suppose now that a point $p$ is deleted. We remove $p$ from each
set $U_{i}$ that contains $p$. If there is no $r\in R$ such that
$(p,r)\in S$ then we do not do anything else. Assume now that $(p,r)\in S$
for some $r\in R$. %
The intuition is that this does not happen very often since in each
iteration $i$ we choose a point uniformly at random from $U_{i}$.
More precisely, in expectation the adversary needs to delete a constant
fraction of the points in $U_{i}$ before deleting $p$. Consider
the index $i$ such that $(p,r)=(p_{i},r_{i})$. We restart the algorithm
from iteration $i$. More precisely, we initialize $y$ to the values
that they have after raising the dual variables $y_{p_{1}},...,y_{p_{i-1}}$
in this order as described above until a constraint for the respective
point $p_{j}$ becomes half-tight. We initialize $x$ to the corresponding
primal variables, i.e., $x_{p_{j}}^{(2r_{j})}=1$ for each $j\in\{1,...,i-1\}$
and $x_{p'}^{(r')}=0$ for all other values of $p',r'$. Also, we
initialize the set $U_{i}$ to be the obtained set after removing
$p$. With this initialization, we start the algorithm above in iteration
$i$, and thus compute like above the (final) vectors $y,x$ and based
on them $S,\bar{S}$, and $\tilde{S}$.

When we restart the algorithm in some iteration $i$ then it takes
time $O(|U_{i}|k^{2})$ to compute the new set $S$. We can charge
this to the points that were already in $U_{i}$ when $U_{i}$ was
recomputed the last time %
and to the points that were inserted into $U_{i}$ later. After a
point $p$ was inserted, it is charged at most $O(k/\epsilon)$ times
in the latter manner since it appears in at most $O(k/\epsilon)$
sets $U_{i}$. Finally, given $S$, we can compute the sets $\bar{S}$,
$\hat{S}$, and $\tilde{S}$ in time $O(k^{3}/\epsilon^{2})$. The
algorithm from \cite{CharikarP04} takes time $n^{O(1/\epsilon)}$ if the input has size $n$.
One can show that this yields an update time of $k^{O(1/\epsilon)}+(k/\epsilon)^{4}$
for each value $\OPT'$. Finally, the same set $\tilde{S}$ yields
a solution for $k$-sum-of-diameters, increasing the approximation
ratio by a factor of~2. 

\begin{theorem}[{restate=[name=]thmPdMain}]
\label{thm:pd-main}
There are dynamic algorithms for the $k$-sum-of-radii and the $k$-sum-of-diameters
problems with update time
$k^{O(1/\epsilon)} \log\Delta$ and with approximation ratios of $13.008+\epsilon$ and $26.016+\epsilon$, respectively, against an oblivious adversary.
\end{theorem}

\section{\label{sec:random-k-cluster}Algorithm for $k$-center}

We present our algorithm for the $k$-center
problem which maintains a $(6+\epsilon)$-approximate solution with an amortized update time of $O(k\log^{2}n\log\Delta)$ against an adaptive adversary. 
As subroutines, it uses 
a dynamic $(2+\epsilon)$-approximation algorithm with an update time of $O(n+k)$, and, for the deletion-only case, 
a bi-criteria $(4+\epsilon)$-approximation algorithm that is allowed to use $O(k\log n)$ centers.
Again, we run $O(\log_{1+\epsilon}\Delta)$ data structures in parallel, one
for each value $\OPT'$ that is a power of $1+\epsilon$ in $[1,\Delta]$ and the algorithm outputs the 
solution of the data structure corresponding to the
minimum value $\OPT'$ among all data structures that output a solution (rather than asserting that the respective $\OPT' < \OPT$).
In the appendix, we provide the details of the fully dynamic algorithm with linear update time that we sketched in the technical overview, corresponding to the following lemma. %
\begin{lemma}[{restate=[name=]lemRndKcenLinup}]
\label{lem:mpt-rnd-kcen-linup}
For any given value $\OPT'$ there is
a fully dynamic algorithm that, started on a (possibly empty) set $P$ with $|P| = n_0$ for some integer $n_0$ has preprocessing time $O(k n_0)$ and
amortized update time $O(n+k)$ such that after each
operation it either returns a solution $C$ of value at most $2\OPT'$ or asserts
that $\OPT>\OPT'$, where $n$ is the current size of $P$. Any point in $P$ becomes a center
at most once, and then stays a center until it is deleted.
\end{lemma}

\subsection{\label{subsec:deletion-only}Deletion-only algorithm with $O(k \log n)$ centers}

Suppose we are given a set of $n$ points $P$ and in each operation
the adversary deletes a point in $P$, but the adversary cannot insert
new points into $P$. We present a data structure for this case which we  preprocess as follows.
We define a set of centers $C$ and a partition $\mathcal C$ of $C$ into \emph{buckets}
$C_{\log n}, C_{\log n - 1}, ...,C_{1}$ which we initialize by $C:=\emptyset$ and
$C_{j}:=\emptyset$ for each $j\in\{1,...,\log n\}$. At the beginning
of each iteration $i$ we are given a set of points $U_{i}$ where
for $i=1$ we define $U_{1}:=P$. 
In each iteration $i$, we select
a point $c_{i}\in U_{i}$ uniformly at random from $U_{i}$, add $c_{i}$
to $C$ (i.e., make it a center). Note that this induces a natural ordering on the centers, i.e.,~$c_i$ becomes the $i$-th center.
We also assign $c_{i}$ to the bucket $C_{j}$ such that $2^{j-1}<|U_{i}|\le2^{j}$.
Also, we assign to $c_{i}$ all points in $P(c_{i}):=\{p\in U_{i} \mid d(p,c_{i})\le2\OPT'\}$
which form the \emph{cluster} of $c_{i}$. We define $U_{i+1}:=\{p\in U_{i} \mid d(c_{i},p)>2\OPT'\}$
and observe that $|U_{i}|>|U_{i+1}|$. Note that therefore every center $c_j$ with
${1} \le j \le {i-1}$ was assigned to a bucket in $C_{\log n},C_{\log n-1},...,C_{1}$.
We stop at the beginning of an iteration $i$ if $U_{i}=\emptyset$.
Also, we stop after some iteration $i$ if the respective bucket $C_{j}$
of iteration $i$ contains $2k$ centers after adding $c_{i}$ to
it. In particular, this ensures that $|C| = \sum_{j=1}^{\log n} |C_j| \le2k\log n$. 
We say that a bucket $C_{j}$ is \emph{small} if $|C_{j}| \le k$ and \emph{big}
otherwise. If $|C_{j}|= 2 k$ the bucket is \emph{very big}, and no bucket has ever more than $2k$ centers. 
Whenever a bucket $C_j$ has received $2k$ centers, the while loop stops and we call $C_j$ the (very big) bucket $C^*$.
There can be at most one very big bucket $C^*$
and if so, $C^*$ must be the last bucket into which centers were added and out of all non-empty buckets the one with smallest index. Note also that as long as $C^*$ is big or very big, it provides a ``witness'' that no solution for the $k$-center problem with value $\OPT'$ exists.

As data structure we keep (1) an array $\mathcal{C}$ of size $\log n$ with each entry pointing to the corresponding bucket if it is non-empty, and a nil-pointer otherwise, (2) for each bucket a list of its centers, (3) for each center $c$ the index of its bucket as well as a list to all points in $P(c)$, and for every point that is not a center, a nil-pointer instead of $P(c)$, (4) a pointer to the index of $C^*$ in the array $\mathcal{C}$, (5) the number of centers in $C^*$, (6)
a binary flag $F$, described below, and if $F$ is set,
a list $U^*$ which contains all points that are not within distance $2 \OPT'$ of any  center as well as the sets $P(c)$ of some already deleted centers $c$ of $C^*$.
\emph{A set flag $F$ indicates that no solution for the $k$-center problem with value $\OPT'$ exists since the last bucket is big. If $F$ is not set, the set $C$ contains a set of at most $2k \log n$ centers that cover $P$ and $F$ will never set again.}

\vspace{-.2cm}
\paragraph{Update operation.}
Suppose now that the adversary deletes a point $p\in P$. \aw{Intuitively, we do the following: if $p$ is not a center then we do not need to do anything. If $p$ is a center and $F$ is not set is then we simply replace $p$ by some other point in its cluster. We do the same if $F$ is set but $p$ is not a center from the last bucket $C^*$. If $F$ is set and $p \in C^*$ then we might perform a \emph{partial rebuild}.
Formally, }we update the clustering lazily as follows:

\noindent\emph{Case 1:} If $p\notin C$ then we simply delete $p$ from the set $P(c)$ that it belongs to.

\noindent\emph{Case 2:} If $p\in C$ and $p$ does not belong to $C^*$, we select another point $p'\in P(p)$ instead of $p$ and
set $P(p')$ to $P(p)$; if $P(p)=\emptyset$, we  do nothing and the size of $C$ decreases by one.

\noindent\emph{Case 3:} If $p$ belongs to $C^*$ and $F$ is not set, we perform the same update as in Case 2.

\noindent\emph{Case 4:} If $p$ belongs to $C^*$ and $F$ is set, we simply delete $p$ as center $C^*$, add its set $P(p)$ to  $U^*$, and decrease the center counter of $C^*$. If $C^*$ 
becomes small due to this deletion, let 
 $c_{i}\in C^*$ denote the first center that was inserted
into $C^*$, i.e., it is the center with the smallest index in $C^*$. We define $U_{i}:=\{p\in P \mid \forall i'<i : d(p,c_{i'})>2\OPT' \}$. Note that this set can be created by taking the union of
the sets $P(c)$ stored in the list of pointers of $C^*$ with $U^*$. Let $C:=\bigcup_{j'=1}^{j-1}C_{j'}$ which can be created by removing all centers $c_i, c_{i+1}, \dots$ (stored at $C^*$) from $C$. Then we perform a partial rebuild by calling a partial rebuild routine $PR(U_i, C)$.

We describe the partial rebuild routine $PR(U_i, C)$ now.
We start the routine by initializing it with $U_i$ and $C$. Then we 
select a point $c_{i}$ uniformly at random from $U_{i}$, add $c_{i}$
to $C$, and assign $c_{i}$ to the bucket $C_{j}$ such that $2^{j-1}<|U_{i}|\le2^{j}$.
As we are in the deletion-only case, $C_j$ is necessarily the non-empty bucket (after $c_i$ was added) with smallest index.
Also, by processing all points in $U_i$ we create the set  $P(c_{i}):=\{p\in U_{i} \mid d(p,c_{i})\le2\OPT'\}$, assign it to $c_{i}$, and remove them from $U_i$. Let us call the resulting set $U_{i+1}$, i.e.
 $U_{i+1}:= U_i \setminus P(c_{i})$,
and iterate. Let $C_j$ be the bucket into which $c_i$ was placed.
Like before, we stop if $U_{i}=\emptyset$ at the beginning of some
iteration $i$ or if after some iteration $i$ the bucket $C_j$
contains $2k$ centers. In the former case we unset the flag $F$, and then no further partial rebuild will happen and also no very big bucket will ever exist again.
In the latter case we keep the flag $F$ set and denote $C_j$ by $C^*$.

\vspace{-.2cm}
\paragraph{Update time.}
We want to show that the algorithm has an amortized update time of
$O(k)$ %
per operation. 
To argue this, we partition the operations into \emph{phases }$Q_{1},Q_{2},...$ such that whenever an operation $t$ causes us to partially rebuild, the current phase ends. Note that each update in a phase, except the last one, only 
performs $O(1)$ work.
Thus, if no partial rebuild happens at the end of a phase $Q$ (because the flag $F$ is not set or because after the last operation in $Q$ the adversary does not delete any more points), then
we only perform $O(1)$ work per update.
Thus, consider a phase $Q$ at whose end we perform a partial rebuild.
\aw{Then, at the beginning of 
$Q$ there is exactly one very big bucket, namely $C^*$,  
which contains exactly $2k$ centers at the beginning of $Q$ and exactly $k$ at the end.}
Thus, exactly $k$ centers were deleted from $C^*$ during $Q$.
Let
$t$ be  the last operation of that phase. %
For each operation $t\in Q$ we define a budget $a_{t}$ as follows. 
Let $P^{(t)}$ denote the set of existing points before operation $t$ and for each center $c_{i'}\in C^*$ we define $U_{i'}^{(t)}:=\{p\in P^{(t)} \mid \forall i''<i' : d(p,c_{i''})>2\OPT'\}$.
As we are in the deletion-only case, $U_{i'}^{(t)} \subseteq U_{i'}^{(t')}$ for $t' \le t$.
If operation $t$ does not delete a center in $C^*$, then we define
$a_{t}:=k$. Otherwise, suppose that operation $t$ deletes a center
$c_{i}\in C^*$. We set $a_{t}:=|U_{i}^{(t)}| + k$. 
For each of the
at most $2k$ centers $c_{i'}\in C^*$ we know that the adversary
deletes $c_{i'}$ with probability at most $1/|U_{i'}^{(t)}|$. This holds since
$c_{i'}$ was chosen uniformly at random from the respective set $U_{i'}$
when $c_{i'}$ was selected the last time in a partial rebuild
or during preprocessing. In case that $c_{i'}$ is deleted it holds that $a_{t}=|U_{i'}^{(t)}| + k$.
Since $|C^*|\le2k$ one can show that $\mathbb{E}[a_{t}]\le 1 \cdot k + \sum_{i \in [2k]} 1 \leq O(k)$. 
\begin{lemma}
\label{lem:a_t} For every $t \ge 0$ we have that $\mathbb{E}[a_{t}]\le O(k)$. 
\end{lemma}
Let $c_{i}$ be the first center that was added to $C^*$  in phase $Q$ and let $U_{i}$
be the corresponding set during the last partial rebuild. The
overall work during phase $Q$ is bounded by $O(|Q|+k|U_{i}|)$, where $|Q|$ denotes the number of iterations during $Q$:
\hf{for the at most $2k \log n$ new centers chosen during the partial rebuild at the very end,
it takes time $\Theta(|U_i| + |U_{i+1}| + \ldots)$ to compare, for all $j \geq i$,
every point in $U_j$ to determine $P(c_j)$. Since $\sum_{j=0}^{k \log n} |U_{i + j}| \leq \sum_{j=0}^{\log n} k|U_{i + jk}| \leq \sum_{j=0}^{\log n} k|U_{i}| / 2^{j-i} \leq 2k|U_i|$, the partial rebuild takes time  $\Theta(k|U_{i}|)$.}
All other operations in $Q$ cause $O(1)$ work each. Now the key insight
is that $|Q|+k|U_{i}|\le O(\sum_{t\in Q}a_{t})$. If $|Q|\ge|U_{i}|/4$,
this is immediate since each $a_t \ge k$.  
Otherwise, we observe that if a center $c_{i'}$ is deleted by operation $t$, then $a_t = |U_{i'}^{(t)}|\ge|U_{i}|/4$ as at most $|Q| < |U_{i}|/4$ many points where deleted during phase $Q$ and at the beginning of the phase
$|U_{i'}| > |U_i|/2$.
Since $k$ centers are deleted before a partial rebuild, we obtain the following lemma.

\begin{lemma}[{restate=[name=]lemMptBoundQ}]
\label{lem:bound-Q}The running time during phase $Q$ is bounded by $O(\sum_{t\in Q}a_{t})$. 
\end{lemma}
Now Lemmas~\ref{lem:a_t} and \ref{lem:bound-Q} together imply that
the expected running time during phase $Q$ is bounded by $O(|Q|k)$
which yields an amortized expected update time of $O(k)$ per operation.
\begin{lemma}
The algorithm has an amortized expected update time of $O(k)$ per operation.
\end{lemma}

For the fully dynamic algorithm in the next subsection we also need to bound the amortized number of changes in the set $C$.

\begin{lemma}\label{lem:amortchanges}
There are at least $k$ deletions between two partial rebuilds.
If there are more than 2 changes to $C$ after an update operation, then there was a partial rebuild at this
operation.
\end{lemma}
\begin{proof}
Each phase consists of all deletions between two consecutive rebuilds. At the beginning of a phase $C^*$ 
contains $2k$ centers, at the end of the phase it only contains $k$ centers. Thus there are at least $k$ deletions during a phase, i.e., between two partial rebuilds.

 Each update operation that is not the last operation of a phase deletes at most one and creates at most one new center. Thus, after a deletion there are only two changes to $C$, except if a partial rebuild was performed.
\end{proof}

\vspace{-.2cm}
\paragraph{Query operation.}
The full details appear in the appendix (see the technical overview for a short version).
\begin{lemma}[{restate=[name=]lemRndKcenDelOnly}]
\label{lem:mpt-rnd-kcen-delonly}
There is an algorithm for the deletion-only case with an expected
update time of $O(k)$ that either outputs a solution of value $4\OPT'$
with $O(k\log n)$ centers, or asserts that there is no solution with
value $\OPT'$ that uses at most $k$ centers.
\end{lemma}

\subsection{General case}

We use the algorithms from Lemma~\ref{lem:mpt-rnd-kcen-linup} and 
Section~\ref{subsec:deletion-only}
and techniques from Overmars~\cite{DBLP:books/sp/Overmars83} in order to obtain a fully dynamic
data structure with $\tilde O(k)$ time per operation. We partition the points $P$ into $O(\log n)$ groups
$P_{1},...,P_{O(\log n)}$. This partition is dynamically created by the algorithm. The data structure
consists of three parts:
(1) For $P_{1}$ we use the fully dynamic data structure with linear update time from Lemma~\ref{lem:mpt-rnd-kcen-linup}, to which we refer to as \emph{$ \mathcal{D}_1 $}. 
(2) For each $P_{j}$ with $j\ge2$ we maintain the data structure from Section~\ref{subsec:deletion-only},
which we will refer to as \emph{$ \mathcal{D}_{2,j} $}. As this data structure can only handle deletions,
we describe below how we handle insertions. %
(3) For each $j\in\{1,...,\log n\}$,
let $C_{j}\subseteq P_{j}$ denote the centers selected by the data
structure for $P_{j}$. Let $C:=\bigcup_{j}C{}_{j}$. 
Note that $|C| = O(k \log^2 n)$ at all times.
We use another instance of the data structure from 
Lemma~\ref{lem:mpt-rnd-kcen-linup}, to which we refer to as \emph{$ \mathcal{D}_3 $}, with $C$ as input points. Let $C'\subseteq C$ denote
the centers selected by this data structure. If $|C'| \le k$, the algorithm outputs $C'$, otherwise $|C'| =k+1$ and
it outputs that $\OPT' < \OPT$.

\vspace{-.2cm}
\paragraph{Update operation.}
Suppose that a point $p$ is inserted. We add $p$ into the data structure $ \mathcal{D}_1 $
for $P_{1}$. If $p\in C_{1}$, then we add $p$ to $ \mathcal{D}_3 $. 
Then we iteratively do the following ``overflow operation''
for each $j=1,...,\log n$ in this order. If $|P_{j}|>2^{j}k \log n$
then
\begin{inenum}
\item we remove $C_{j}$ and $C_{j+1}$ from $C$ and $ \mathcal{D}_3 $ accordingly,
\item we add $P_{j}$ to $P_{j+1}$, set $P_{j}:=\emptyset$, and reinitialize
the data structure $\mathcal{D}_{2,j+1}$ for $P_{j+1}$, obtaining a set of $O(k\log n)$
centers $C_{j+1}$,
\item then we add $C_{j+1}$ to $C$. 
\end{inenum}
If at least one overflow operation was executed, we perform all necessary overflow operations and afterwards we reinitialize $\mathcal{D}_3$ with the new set $C$.

Suppose that a point $p$ is deleted and assume that $p\in P_{j}$
for some $j$. 
If $j=1$ then we delete $p$ from $\mathcal{D}_{1}$ and update $\mathcal{D}_3$ accordingly. Suppose now that $j>1$.
Then we delete $p$ from $\mathcal{D}_{2,j}$.
Suppose that due to this deletion operation a set $C_{j}^{+}\subseteq P_{j}$
is added to $C_{j}$ and a set $C_{j}^{-}\subseteq C_{j}$ is deleted
from $C_{j}$. Then we add $C_{j}^{+}$ to $C$, and remove $C_{j}^{-}$
from $C$. If there are at most 2 changes to $C$ (e.g., this happens when one point is deleted and possibly replaced by some other point inside one data structure $\mathcal{D}_{2,j}$), we update $ \mathcal{D}_3$ using two update operations.
If there are more changes to $C$ (this happens if one of the data structures $\mathcal{D}_{2,j}$ performs a partial rebuild), we rebuild $\mathcal{D}_3$ from scratch. %

\vspace{-.2cm}
\paragraph{Update time.}
We want to analyze the amortized update time. 
Note that (i) as $|P_1|\le k \log n$ and
$|C| = O(k\log^2n)$, the cost per update operation in $ \mathcal{D}_1 $ and $ \mathcal{D}_3 $ is $O(k \log n)$ and $O(k \log^2 n)$, respectively, and (ii)
the amortized expected cost for each update operation in each $ \mathcal{D}_{2,j} $ is $O(k)$.

Each deletion leads to one update operation in $ \mathcal{D}_1 $ or $ \mathcal{D}_{2,j} $ for one $j$.
Suppose first that we update $ \mathcal{D}_1 $. Each point can become at most once a center in $ \mathcal{D}_1 $ and then stays a center until it is deleted or moved to $\mathcal{D}_2$. Hence, we charge the update cost of $O(k \log^2 n)$  in $\mathcal{D}_3$  to the respective point that becomes a center in $\mathcal{D}_1 $ or is deleted from $\mathcal{D}_1$. 

Suppose now that we do an update operation in $ \mathcal{D}_{2,j} $ for some $j$.
If there are only 2 changes in $C$, we update $\mathcal{D}_3$ with two update operations in $O(k \log^2 n)$ time. If there are more changes in $C$, we rebuild $\mathcal{D}_3$ from scratch in time $O(k^2 \log^2 n)$.
By Lemma~\ref{lem:amortchanges} there must have been a rebuild of $\mathcal{D}_{2,j}$ at this deletion and there were at least $k$ update operations in $\mathcal{D}_{2,j}$
 since the last rebuild. Thus we can amortize the cost of $O(k^2 \log^2 n)$ for the rebuild of
$\mathcal{D}_3$ over these deletions. As each point is deleted only once and each deletion belongs only to one interval between two rebuilds (i.e. phases), each deletion is only charged once in this way. Thus, this amortized charge adds 
an amortized cost of  $O(k \log^2 n)$ to the running time of each deletion.
 
Each insertion leads to (a) one update operation in $ \mathcal{D}_1 $,
(b) potentially multiple reinitialization of $ \mathcal{D}_{2,j} $ for different sets $P_j$, and (c)  if at least one $\mathcal{D}_{2,j}$ was reinitialized, then one reinitialization of  $ \mathcal{D}_3 $.
Thus,  it remains to bound (b) and (c).

Consider a point $p$.
For each $j\in\{2,...,\log n\}$, the point $p$ can participate at
most once in an overflow operation from $j-1$ to $j$, i.e., 
$p$ is moved at most once from $P_{j-1}$ to $P_{j}$. 
In such a case, the cost
for reinitializing the data structure for the new set $P_{j}$ is $O(k|P_{j}|)$.
Additionally, reinitializing $\mathcal{D}_3$ costs
time $O(k |C|) = O(k^2 \log^2 n) = O(k |P_j| \log n)$ as $P_j \ge k \log n$.
This cost is charged to the points that were previously in $P_{j-1}$. As at that time
$|P_{j-1}| > 2^{j-1} k  \log n \ge |P_j|/2$, charging $\Theta(k \log n)$ credits to each point 
in $P_{j-1}$ provides $O(k|P_{j}| \log n)$ credits, which suffices to pay for the 
reinitialization of $P_j$ and $ \mathcal{D}_3 $.

\begin{lemma}[{restate=[name=]lemKcenUpdtTime}]
	\label{lem:mpt-kcen-update-time}
	The algorithm has $O(k n_0)$ preprocessing time, where $n_0$ is the initial number of points, and an expected amortized update time of $O(k\log^{2}n)$.
\end{lemma}

\vspace{-.2cm}
\paragraph{Query operation.}
Suppose that a query operation takes place and we need to output our
solution. If one of the data structures for $C,P_{1},...,P_{O(\log n)}$
asserts that $\OPT>\OPT'$, then we output that $\OPT>\OPT'$. Otherwise,
we output the centers in $C'$. In this case, the solution value of $C'$ is at most $6\OPT'$: for each $j$ and each point $p\in P_{j}$,
there is a point $c\in C_j$ with $d(p,c)\le4\OPT'$ by \cref{lem:mpt-rnd-kcen-delonly}. By definition,
$c\in C$, and hence there is a point $c'\in C'$ with $d(c,c')\le2\OPT'$ by \cref{lem:mpt-rnd-kcen-linup}.
This implies that $d(p,c')\le6\OPT'$.
\begin{theorem}
There is a randomized dynamic $(6+\epsilon)$-approximation algorithm
for the $k$-center problem with amortized update time of $O(k\log\Delta\log^{2}n)$ against an oblivious adversary.
\end{theorem}

We prove that our update time is tight up to logarithmic factors because an algorithm needs to make at least amortized $\Omega(k)$ queries in expectation to obtain any non-trivial approximation ratio. This holds for $k$-center and also for all other clustering problems studied in this paper.

\begin{theorem}
	\label{lem:mpt-lb-rand-cost}
	Let $\epsilon > 0$. Any (randomized) dynamic algorithm for finite metric spaces that can provide against an oblivious adversary a $(\Delta - \epsilon)$-approximation  to the cost of $k$-center, $k$-sum-of-radii, $k$-sum-of-diameters or a $(\Delta^p - \epsilon)$-approximation to the cost of $(1,p)$-clustering (where $p=1$ yields $k$-median and $p=2$ yields $k$-means) queries amortized $\Omega(k)$ distances in expectation after $O(k^2)$ operations, also if the algorithm needs to provide 
	only its center set and not an approximation $\apx$ to the cost.
\end{theorem}

\bibliographystyle{plainurl}%
\bibliography{literature}

\appendix

\section{Related work}
We summarize prior work on offline and dynamic results for the clustering problems we consider in this paper. Most prior work on fully dynamic clustering algorithms is for special metric spaces that are a generalization of the Euclidean metric space, namely in metric spaces \emph{with bounded doubling dimension}.

\paragraph{$k$-center clustering. }
There is a simple greedy 2-approximation algorithm for $k$-center~\cite{hsu1979easy}
which is best possible~\cite{DBLP:journals/tcs/Gonzalez85}. This
is the first clustering problem for which dynamic algorithms were
found. Indeed, already in 2004 Charikar et al.~\cite{CharikarCFM04}
gave two insertions-only algorithms: a deterministic 8-approximation
and a randomized 5.43-approximation, both with $O(k\log k)$ update
time per operation, which was improved to an (almost optimal) insertion-only
$(2+\epsilon)$-approximation algorithm with update time $O((k\log k))/\epsilon^{4})$~\cite{McCutchenK08}
for any small $\epsilon>0$. The best known fully dynamic algorithm
was given by Chan et al.~\cite{ChanGS18}. It is randomized and achieves
an approximation ratio of $2+\epsilon$ with an expected update time
of $O((k^{2}\log\Delta)/\epsilon)$ against an oblivious adversary.
In  metric spaces with bounded doubling dimension Goranci et al.~recently gave a $(2 + \epsilon)$-approximation algorithm with update time independent of $k$, namely $O((\log \Delta \log \log \Delta)/\epsilon)$ time per operation~\cite{Goranci20}. More generally, if the doubling dimension is
$\kappa$, the time per operation is $O((2/\eps)^{O(\kappa)} \log \Delta \log \log \Delta \log \eps^{-1})$.
Note that $\kappa$ can be $\Theta( \log n)$, in which case the above algorithm takes time polynomial in $n$.
Independently, a 16-approximation algorithm with $O(D^2  (\log \Delta)^2 \log n)$ time per operation was presented by Schmidt and Sohler in~\cite{abs-1908-02645} if the metric space is a $D$-dimensional Euclidean space. This algorithm also works for \emph{hierarchical} $k$-center clustering.

\paragraph{$k$-sum-of-radii and $k$-sum-of-diameters. }
The best known polynomial-time result for $k$-sum-of-radii is a $(3.504+\epsilon)$-approximation
algorithm~\cite{CharikarP04} and there is a QPTAS~\cite{gibson2010metric}.
Any solution for the $k$-sum-of-radii problem can be turned into
a solution for the $k$-sum-of-diameters problem with an increase in
the cost of at most a factor of 2 (and vice versa without any increase).
The $k$-sum-of-diameters problem is NP-hard to approximate with a
factor better than 2~\cite{DoddiMRTW00}.
Henzinger et al~\cite{DBLP:journals/algorithmica/HenzingerLM20} recently developed a fully dynamic algorithm for
a variant of the sum-of-radii problem which does not limit the number of centers. Instead there is a set
$\cal F$ of facilities and a set $\cal C$ of clients and the algorithm must assign a radius $r_i$ to each facility  $i$ such that every client is within distance at most $r_i$ for some facility $i$. The cost of the solution is the
sum of the radii $\sum_i r_i$. Additionally each facility can be assigned an opening cost $f_i$ and the algorithm then
tries to minimize $\sum_{i, r_i > 0} r_i + f_i.$
 In metric spaces with bounded doubling dimension the algorithm achieves  a constant approximation in $O(\log \Delta)$ time per operation. More generally, if the doubling dimension is $\kappa$, the algorithm maintains
a $O(2^{2 \kappa})$-approximation in time $O(2^{6\kappa}\log \Delta )$ per operation.
Note that $\kappa$ can be $\Theta( \log n)$, in which case the above algorithm takes time polynomial in $n$.
 
Note that the $k$-sum-of-radii problem is related to the \aw{\textsc{Set
		Cover}} problem where there is one set for each combination of a
center $c$ and a radius $r$ that contains all points $p$ with $d(p,c)\le r$.
However, the number of sets in a solution of \textsc{Set Cover}
is not limited and, thus, it cannot be used to solve the $k$-sum-of-radii
problem.

\paragraph{$k$-means and $k$-median.}
In practice, the $k$-median and $k$-means problems are very popular.
The best static polynomial-time approximation algorithm for $k$-median
achieves an approximation ratio of $2.675+\epsilon$~\cite{byrkaprst_17}
and the best such algorithm for $k$-means achieves a ratio of $9+\epsilon$~\cite{ahmadiannsw_17}.
Furthermore, there is a lower bound of $\Omega(nk)$ on the running
time of any (static) constant-factor approximation algorithm~\cite{MettuP04},
implying that there is no insertions-only algorithm with $o(k)$ time
per operation. %
For both problems Henzinger and Kale~\cite{DBLP:conf/esa/HenzingerK20}
gave randomized $O(1)$-approximation algorithms with $\tilde{O}(k^{2})$
worst-case update time against an adaptive adversary.
Cohen-Addad et al.~\cite{Cohen-AddadHPSS19} gave fully dynamic $O(1)$-approximation algorithms for 
$k$-means and $k$-median
with expected amortized update time of $\tilde{O}(n+k^{2})$. Note, however, that their algorithm is \emph{consistent}, i.e., it additionally tries to minimize the number of center changes.

\SetKwProg{Fn}{Function}{}{}
\SetKwRepeat{Do}{do}{while}
\DontPrintSemicolon

\def\N{\mathbb{N}}
\def\defeq{:=}
\NewDocumentCommand{\card}{m}{%
	\lvert #1 \rvert%
}
\NewDocumentCommand{\dmax}{}{%
	\Delta%
}

\newcommand{\costbound}{\ensuremath{\OPT'}}
\NewDocumentCommand{\cost}{o m O{k}}{%
	\mathrm{cost}_{#3}(#2\IfNoValueF{#1}{, #1})%
}
\NewDocumentCommand{\opt}{m O{k}}{%
	\mathrm{opt}_{#2}(#1)%
}

\def\adv/{\textsc{Adv}}
\def\alg/{\textsc{Alg}}
\NewDocumentCommand{\auxgraph}{m o o}{%
	#1%
	\IfNoValueF{#2}{_{#2%
		\IfNoValueF{#3}{,#3}%
	}}%
}
\NewDocumentCommand{\auxnum}{m o o}{%
	\lvert {#1}%
	\IfNoValueF{#2}{_{#2%
		\IfNoValueF{#3}{,#3}%
	}} \rvert%
}
\NewDocumentCommand{\kg}{o o}{%
	\auxgraph{G}[#1][#2]%
}
\NewDocumentCommand{\actno}{o o}{%
	\auxgraph{A}[#1][#2]%
}
\NewDocumentCommand{\actnum}{o o}{%
	\auxnum{A}[#1][#2]%
}
\NewDocumentCommand{\pasno}{o o}{%
	\auxgraph{P}[#1][#2]%
}
\NewDocumentCommand{\pasnum}{o o}{%
	\auxnum{P}[#1][#2]%
}
\NewDocumentCommand{\disno}{o o}{%
	\auxgraph{D}[#1][#2]%
}
\NewDocumentCommand{\disnum}{o o}{%
	\auxnum{D}[#1][#2]%
}
\NewDocumentCommand{\akg}{o o}{%
	\auxgraph{H}[#1][#2]%
}
\NewDocumentCommand{\s}{o}{%
	\sigma%
	\IfNoValueF{#1}{_{#1}}%
}
\NewDocumentCommand{\allpoints}{m}{%
	Q%
	\IfNoValueF{#1}{_{#1}}%
}
\NewDocumentCommand{\points}{m}{%
	P%
	\IfNoValueF{#1}{_{#1}}%
}
\NewDocumentCommand{\spr}{m}{%
	\s[\leq #1]%
}
\NewDocumentCommand{\sprs}{m}{%
	\s[< #1]%
}
\NewDocumentCommand{\q}{o o}{%
	q%
	\IfNoValueF{#1}{_{#1%
		\IfNoValueF{#2}{,#2}%
	}}
}
\NewDocumentCommand{\V}{m}{%
	V(#1)%
}
\NewDocumentCommand{\E}{m}{%
	E(#1)%
}
\NewDocumentCommand{\ngh}{o m}{%
	\Gamma%
	\IfNoValueF{#1}{_{#1}}%
	(#2)%
}
\NewDocumentCommand{\dg}{o m}{%
	\operatorname{deg}%
	\IfNoValueF{#1}{_{#1}}%
	(#2)%
}
\NewDocumentCommand{\dsp}{o m m}{%
	d%
	\IfNoValueF{#1}{_{#1}}%
	(#2, #3)%
}
\NewDocumentCommand{\ans}{m}{%
	\textrm{ans}(#1)%
}

\section{Lower bounds for algorithms against an adaptive adversary}
\label{sec:appendix-lower-adaptive}

We introduce some formal notation and definitions we use to revisit the adversarial strategy that generates an input stream and answers distance queries on the set of currently known points. Then, we derive lower bound constructions for the aforementioned problems that are based on the metric space defined by the stream and the answers to the algorithm.

\SetKwFunction{FnGenerate}{GenerateStream}
\SetKwFunction{FnAnswer}{AnswerQuery}

We describe a strategy for an adversary $\mathcal{A}$ that generates a stream of update operations~$\s$ and answers distance queries~$\q$ on pairs of points by any dynamic algorithm with a guarantee on its amortized  complexity. In the following presentation, the adversary constructs the underlying metric space ad hoc. More precisely, the adversary constructs two metric spaces simultaneously that cannot be distinguished by the algorithm and its queries. All subsequent lower bounds stem from the fact that the problem at hand has different optimal costs on the input for the two metric spaces. When the algorithm outputs a solution, the adversary can fix a metric space that induces high cost for the centers chosen by the algorithm.

During the execution of the algorithm, the adversary maintains a graph $\kg$.
Each point that was inserted by the adversary is represented by a node in $\kg$. All query answers given to the algorithm by the adversary can be derived from 
$G$
using the shortest path metric $\dsp[\kg]{\cdot}{\cdot}$ on $\kg$. We denote the algorithm's $i$\xth/ query after update operation $t$ by $\q[t][i]$, and the adversary's answer by $\ans{\q[t][i]}$. For $t > 0$, we denote the number of queries asked by the algorithm between the $t$\xth/ and the $(t+1)$\xth/ update operation by $c(t)$. If $t$ is clear from context, we simplify notation and write $c := c(t)$. Note that, in this section, we use a slightly extended notation when indexing graphs when compared to other sections. Details follow.

We number the update operations consecutively starting with 1 using index $t$ and after each update operation, we index the distance queries that the algorithm issues while processing the update operation and the immediately following value- or solution-queries using index $i$. Let $c > 0$ and let $\kg[0][c(0)]$ be the empty graph.  For every $t > 0, i \in [c(t)]$, consider $i$-th query issued by the algorithm processing the $t$-th update operation.  The graph $\kg[t][i]$ has the following structure. For every point $x$ that is inserted in the first $t$ operations, $\V{\kg[t][0]}$ contains a node $x$. All edges have length $1$, and it holds that $\V{\kg[t][i]} \supseteq \V{\kg[t][i-1]}$ and $\E{\kg[t][i]} \supseteq \E{\kg[t][i-1]}$.

Edges are inserted by the adversary as detailed below. Let $\preceq$ denote the predicate that corresponds to the lexicographic order. In particular, the adversary maintains the following invariant, which is parameterized by the update operation $t$ and the corresponding query $i$: for all $(t',j) \preceq (t,i)$ and $(u,v) \defeq \q[t'][j]$, $\ans{\q[t'][j]} = \dsp[\kg[t][i]]{u}{v}$. In other words, any query given by the adversary remains consistent with the shortest path metric on all versions of 
$G$
after the query was answered. The adversary distinguishes the following types of nodes in $\kg[t][i]$ to answer a query. Recall that $f \defeq f(k,n)$ is an upper bound on the amortized complexity per update operation of the algorithm, which is non-decreasing in $n$ for fixed $k$.

\begin{definition}[type of nodes]
	\label{def:node_types}
	Let $t \geq 0, i \in [c]$ and let $u \in \V{\kg[t][i]}$. If $u$ has degree less than $100f(k,i)$ for all $i \in [t]$, it is \emph{active} after update $t$, otherwise it is \emph{passive}. In addition, the adversary can mark passive nodes as \emph{off}. We denote the set of active, passive and off nodes in $\kg[t][i]$ by $\actno[t][i]$, $\pasno[t][i]$ and $\disno[t][i]$, respectively.%
\end{definition}

For operation $t$, the adversary answers the $i$\xth/ query $\q[t][i]$ according to the shortest path metric on $\kg[t][i-1]$ with the additional edge set $\actno[t][i-1] \times \actno[t][i-1]$. In other words, the adversary (virtually) adds edges between all active nodes in $\kg[t][i-1]$ and reports the length of a shortest path between the query points in the resulting graph. After the adversary answered query $\q[t][i]$, the resulting %
graph $\kg[t][i]$ is $\kg[t][i-1]$ plus all edges of the shortest path that was used to answer the query. A key element of our analysis is that \empty{all} answers up to operation $t$ and query $i$ are equal to the length of the shortest paths between the corresponding query points in $\kg[t][i]$. The generation of the input stream and the answers to all queries are formally given by \cref{alg:stream} and~\ref{alg:answer}, respectively.

\begin{algorithm}
	\Fn{\FnGenerate{$t$}}{
		\If{there exists a passive node $x \in \V{\kg[t-1][c]}$}{
			mark $x$ as off in $\kg[t][0]$\;
			\Return{$\langle$ delete $x$ $\rangle$}
		}
		\Else{
			let $x$ be a new point, i.e., that was not returned by the adversary before\;
			\Return{$\langle$ insert $x$ $\rangle$}
		}
	}
	\caption{\label{alg:stream} Construction of element $\s[t]$ of $\s$}
\end{algorithm}
\begin{algorithm}
	\Fn{\FnAnswer{$\q[t][i] = (x,y)$}}{
		let $\akg[t][i] = (\V{\kg[t][i-1]}, \E{\kg[t][i-1]} \cup (\actno[t][i-1] \times \actno[t][i-1]))$\;
		let $p = (e_1, \ldots, e_k)$ be a shortest path between $x$ and $y$ in $\akg[t][i]$\;
		set $\kg[t][i] \defeq \kg[t][i-1]$\;
		\ForEach{$e_i$}{
			insert $e_i$ into $\kg[t][i]$\;
		}
		\Return{length of $p$}
	}
	\caption{\label{alg:answer} Answer of the adversary to query $\q[t][i]$}
\end{algorithm}

\subsection{Adversarial strategy}

Let $n_t$ be the number of active and passive nodes, i.e., the number of current points for the algorithm, after operation $t$. In the whole section we use the notations  of $\kg[t'][i], \actno[t][i]$ etc.  from Definition~\ref{def:node_types}.

\subsubsection{Proof of \cref{lem:mpt-det-low-adv-prop}}

The next three lemmas prove the three claims in \cref{lem:mpt-det-low-adv-prop}.

\begin{lemma}[\cref{lem:mpt-det-low-adv-prop} (\ref{it::mpt-det-low-adv-prop-numac})]
	\label{lem:kg-actnum}
	For every $t > 0$, the number of active nodes in $\kg[t][0]$ is at least $96t / 100$.
\end{lemma}
\begin{proof}
	Recall that $f(k,n)$ is a positive function that is non-decreasing in $n$. We prove the claim by induction. By the properties of $f$, the case $t = 1$ follows trivially. Let $t \geq 2$. For any $i \in [t]$, the algorithm's query budget increases by $f(k, n_i)$ queries after the $i$\xth/ update. Since $f(k,i)$ is non-decreasing, nodes inserted after update operation $i-1$ can only become passive if their degrees increase to at least $100f(k,i)$. Therefore, it holds that $\pasnum[t][0] \leq \sum_{i \in [t]} 2f(k, n_i) / (100f(k, i)) \leq \sum_{i \in [t]} 2f(k, i) / (100f(k, i)) \leq 2t / 100$. It follows that the adversary will delete at most $2t / 100$ points in the first $t$ operations and insert points in the other at least $(1-2/100)t$ operations. The number of active nodes after update $t$ is $\actnum[t][0] \geq t - \pasnum[t][0] - 2t / 100 \geq 96t / 100$.

\end{proof}

\begin{lemma}[\cref{lem:mpt-det-low-adv-prop} (\ref{it::mpt-det-low-adv-prop-deg})]
	\label{lem:max-degree}
	For every $t > 0, i \in [c]$, all nodes have degree at most $100f(k,t)$ in $\kg[t][i]$.
\end{lemma}
\begin{proof}
	By definition, the claim is true for active nodes. Edges are only inserted into $G$
	if the algorithm queries for the distance between two nodes $x,y$ and the adversary determines a shortest path between $x$ and $y$ that contains edges that are not present in $\kg[t][i-1]$ (see \cref{alg:answer}). Since all such edges are edges between active nodes, only degrees of active nodes in $\kg[t][i]$ increase. The adversary finds a shortest path on a supergraph of $\actno[t][i] \times \actno[t][i]$. Therefore, any shortest path it finds contains at most one edge with two active endpoints. It follows that a query increases the degree of any active node in $\kg[t][i]$ by at most one, which may turn it into a passive node with degree $100f(k,t)$; however, \cref{alg:answer} never inserts edges that are incident to a passive node and passive nodes never become active again by \cref{def:node_types}.
\end{proof}

\begin{lemma}[\cref{lem:mpt-det-low-adv-prop} (\ref{it::mpt-det-low-adv-prop-clean})]
	\label{lem:no-passive-nodes}
	For every $t > 0$, there exists a clean update operation $t'$, $t < t' \le 2t$, i.e., $\kg[t'][0]$ only contains active and off nodes, but no passive nodes.
\end{lemma}
\begin{proof}
	We prove the claim by induction over the operations $t$ with the properties that $\kg[t][0]$ contains no passive node, but $\kg[t+1][0]$ contains at least one passive node. The claim is true for the initial (empty) graph $\kg[0][0]$. Let $t > 0$. We prove that in at least one operation $t' \in \{ t+1, \ldots, 2t \}$, the number of passive nodes is $0$. For the sake of contradiction, assume that for all $t'$, $t < t' \leq 2t$, the number of passive nodes is non-zero, i.e., $\pasnum[t'][0] > 0$. We call an active node \emph{semi-active} if it has degree greater than $50f(k,i)$ after some update $i \in [2t]$. Otherwise, we call it \emph{fully-active}. Recall that, similarly, a vertex becomes passive if it has degree at least $100f(k,i)$ after some update $i \in [2t]$ (and never becomes active again).
	
	For any $i \in [2t]$, the algorithm's query budget increases by $f(k, n_i)$ queries after the $i$\xth/ update. Since $f(k,i)$ is non-decreasing, nodes inserted after update operation $i-1$ can only become semi-active if their degrees increase to at least $50f(k,i)$ (resp. $100f(k,i)$) by the definition of semi-active (resp. passive). Also due to the monotonicity of $f$, the number of semi-active or passive nodes is maximized if the algorithm invests its query budget as soon as possible. It follows that the number of semi-active or passive nodes up to operation $2t$ is at most $\sum_{j \in [2t]} 2f(k,j) / (50f(k,j)) \leq 2t/50$. Without loss of generality, we may assume that all semi-active nodes are passive (so the algorithm does not need to invest budget to make them passive).
	
	After update operation $2t$, the algorithm's total query budget from all update operations is at most $\sum_{i \in [2t]} f(k,n_i) \leq \sum_{i \in [2t]} f(k,i) \leq 2t f(k,2t)$. The algorithm may use its budget to increase the degree of at most $2t f(k,2t) / (50f(k,2t)) \leq 2t/50$ fully-active nodes to at least $100f(k,2t)$, i.e., to make them passive nodes. 
	Recall our assumption that $\pasnum[i][0] > 0$ for all $i \in \{t+1, \ldots, 2t\}$. The adversary deletes one point corresponding to a passive node in each update operation from $\{t+1, \ldots, 2t \}$. Therefore, the number of passive nodes after update operation $2t$ is
	\begin{equation*}
		\pasnum[2t][0] \leq \frac{2t}{50} + \frac{2t}{50} - t \leq \frac{4t}{50} - t < 0.
	\end{equation*}
	This is a contradiction to the assumption.
\end{proof}

\subsubsection{Proofs of \cref{lem:metric-consistent,lem:muni-consistent,lem:muni-consistent,lem:mstar-consistent,lem:mstarrange-consistent}}

The following observation follows immediately from the properties of shortest path metrics.

\begin{observation}
	\label{lem:convex-compliance}
	Let $G=(V,E)$ and $G'=(V,E')$ be two graphs so that $E \subseteq E'$. If a sequence of queries is consistent with the shortest path metric on $G$ as well as on $G'$, then, for any $E''$, $E \subseteq E'' \subseteq E'$, it is also consistent with the shortest path metric on $(V, E'')$.
\end{observation}

The following lemma together with \cref{lem:convex-compliance} implies \cref{lem:metric-consistent,lem:muni-consistent,lem:muni-consistent,lem:mstar-consistent,lem:mstarrange-consistent,lem:mmulti-consistent} by setting $G = \kg[t][i-1]$ and $G' = (\V{\kg[t][i-1]}, \E{\kg[t][i-1]} \cup (\actno[t][i-1] \times \actno[t][i-1]))$ in \cref{lem:convex-compliance}.

\begin{lemma}
	\label{lem:sp-consistent}
	For any $t, t' > 0$, $i, i' \in [c]$ so that $(t',i') \prec (t,i)$, the answer given to query $\q[t'][i']$ is consistent with the shortest path metric on $G'$.
\end{lemma}
\begin{proof}
	Let $t' > 0$, $i' \in [c]$ so that $(t',i') \prec (t,i)$ and denote $(x,y) \defeq q \defeq \q[t'][i']$. We prove that $\ans{q} = \dsp[\kg[t][i]]{x}{y}$. As neither vertices nor edges are deleted after they have been inserted, for every $(t_1, i_1) \prec (t_2, i_2)$, $\kg[t_2][i_2]$ is a supergraph of $\kg[t_1][i_1]$. Thus, if there exists a path $P$ between $x$ and $y$ in $\kg[t'][i']$, a shortest path in $\kg[t][i]$ between $x$ and $y$ cannot be longer than $P$.

	It remains to prove $\ans{q} \leq \dsp[\kg[t][i]]{x}{y}$. For the sake of contradiction, assume that there exist $t'',i''$ 
	so that $\ans{q} = \dsp[\kg[t''][i''-1]]{x}{y}$, but $\ans{q} > \dsp[\kg[t''][i'']]{x}{y}$. By \cref{def:node_types}, passive nodes never become active. For any passive node $v \in \pasno[t''][i''-1]$, it follows that $\dsp[\kg[t''][i'']]{v}{\actno[t][i]} \geq \dsp[\kg[t''][i''-1]]{v}{\actno[t][i-1]}$ as \cref{alg:answer} only inserts edges between vertices in $\actno[t][i]$ into $\kg[t][i]$. Therefore, any shortest path between $x$ and $y$ in the graph $(\V{\kg[t''][i''-1]}, \E{\kg[t''][i''-1] \cup (\actno[t''][i''-1] \times \actno[t''][i''-1])}$ has length at least $\dsp[\kg[t''][i''-1]]{x}{y} = \ans{q}$.
\end{proof}

\subsection{Lower bounds for clustering}

Our lower bounds apply for the case that the algorithm is allowed to choose as centers any points that have ever been inserted as well as to the case where  centers must belong to the set of current points. In the whole section we use the notations  of $\kg[t'][i], \actno[t][i]$ etc. from the introduction of \cref{sec:appendix-lower-adaptive}.

\subsubsection{Proof of \cref{thm:det-1lb}}

For any $\ell\in\N_{0}$, we define a metric $M_{\ell}(P^{*})$
on a subset of points $P^{*}$ as the shortest path metric on the
following graph. For each pair of active vertices $u,v$ we add an edge if $d(u,P^{*})\geq\ell$
and $d(v,P^{*})\geq\ell$.

\begin{lemma} \label{lem:mmulti-consistent} For each clean update operation
	$t$, each subset of points $P^{*}\in V_{t}$, and each $\ell\in\N_{0}$
	the metric $M_{\ell}(P^{*})$ is consistent.
\end{lemma}
\begin{proof}
The metric $M_{\ell}(P^{*})$ is an augmented graph metric for $t$ and, thus, for a clean update operation $t$, it
is consistent by Lemma~\ref{lem:metric-consistent}.
\end{proof}

\begin{lemma} \label{lem:det-low-1pclus} 
Consider any dynamic algorithm for $(1,p)$-clustering that queries amortized $f(1,n)$ distances per operation, where $n$ is the number of current points,
 and outputs at most $1 \le g\leq n$ centers.
For any $t\geq1$ such that $t$ is a clean operation, the approximation factor of the algorithm's solution
(with respect to the optimal $(1,p)$-clustering cost) against an
adaptive adversary right after operation $t$ is at least $\left[\frac{\log(t/4g)}{p + \log(100f(1,t))}\right]^{p}/4$ and $96t/100 \le n \le t$. \end{lemma} 
\begin{proof}
Denote $G\defeq(V,E)\defeq\kg[t][0]$ and $A\defeq\actno[t][0]$.
By \cref{lem:kg-actnum}, the number of active nodes in $G$ is
at least $96t/100\geq t$, which implies that $|A| \ge 96t/100$. 
Thus,  after operation $t$,  $n \ge 96t/100$.

Let $C$ be the centers that are picked by the algorithm after operation $t$. By the assumption of the lemma, $|C| \le g$.
Consider $M_{\ell}(C)$, where $\ell\defeq\log(t/4g)/(p+\log(100 f(1,t)))$.
By \cref{lem:max-degree}, for any $s\in C$, the size of $\lvert\{x\mid x\in A\wedge d(x,s)<\ell\}\rvert$
is at most $\sum_{i\in[\ell-1]}(100 f(1,t))^{i}<(100 f(1,t))^{\ell}$. 
Let $V_{+}\defeq\{x\mid x\in A\wedge d(x,C)\geq\ell\}$.
Since $\lvert C\rvert\leq g$, it holds that $\lvert V_{+}\rvert>|A| -g(100 f(1,t))^{\ell}  \ge 96t / 100 -g(t/4g)^{\log ((100 f(1,t)) / (p + \log (100 f(t,1)))}\ge 96t/100 - t/4 \ge t/2$.
Since $d(s,V_{+})\geq\ell$ for any $s\in C$, the $(1,p)$-clustering
cost of $S$, and thus the cost of the algorithm, is at least $\lvert V_{+}\rvert\cdot\ell^{p}\geq t/2\cdot\ell^{p}$.

We will show that if instead a single point corresponding to a vertex of $V_{+}$ is picked as center, then the cost is at most $2t$, which provides an upper bound on the cost of the optimum solution.
It follows that the approximation factor achieved by the algorithm is at least $\ell^p / 4$.

To complete the proof consider a point $x$ whose corresponding point
$v_{x}$ belongs to $V_{+}$. \awtwo{One can easily show that for
each $\alpha\in\N$ it holds that $\alpha^{p}\le(100f(1,t)\cdot2^{p})^{2\alpha/3}$
since for each $\alpha\in\N$ it holds that $\alpha^{p}=2^{p\log\alpha}\le2^{p\cdot\frac{2\alpha}{3}}$}.
Thus, the $(1,p)$-clustering cost if a single point $x$ is chosen
as center is at most

\begin{alignat*}{1}
\sum_{i=0}^{\ell-1}g(100f(1,t))^{i}\cdot(\ell-i)^{p}+t\cdot1^{p} & \le g\left((100f(1,t))^{\ell}\cdot\sum_{i=1}^{\ell}\frac{1}{(100f(1,t))^{i}}\cdot i^{p}\right)+t\\
 & \le g\left((100f(1,t))^{\ell}\cdot\sum_{i=1}^{\ell}\frac{(100f(1,t)\cdot2^{p})^{2i/3}}{(100f(1,t))^{i}}\right)+t\\
 & \le g\left((100f(1,t))^{\ell}\cdot\sum_{i=1}^{\ell}\frac{2^{p\cdot2i/3}}{(100f(1,t))^{i/3}}\right)+t\\
 & \le g\left((100f(1,t))^{\ell}\cdot\sum_{i=1}^{\ell}\left(\frac{2^{\frac{2p}{3}}}{(100f(1,t))^{1/3}}\right)^{i}\right)+t\\
 & \le g(100f(1,t))^{\ell}2^{p\ell}+t\\
 & \le g(100f(1,t)\cdot2^{p})^{\ell}+t\\
 & \le5t/4\le2t
\end{alignat*}
using that $\ell=\frac{\log(t/4g)}{p+\log(100f(1,t))}=\frac{\log(t/4g)}{\log(2^{p}100f(1,t))}=\log_{2^{p}100f(1,t)}t/4g$.
This yields an approximation ratio of at least $\frac{t/2\cdot\ell^{p}}{2t}=\Omega(\ell^{p})$. 
\end{proof}

\begin{lemma}\label{lem:diam}
Consider any dynamic algorithm for computing the diameter of a dynamic point set  that queries amortized $f(1,n)$ distances per operation, where $n$ is the number of current points,
 and outputs at most $g\geq1$ centers.
For any $t\geq 2$ such that $t$ is a clean operation, the approximation factor of the algorithm's solution
(with respect to the correct diameter) against an
adaptive adversary right after operation $t$ is at least
$\log (96t/100) / (\log{102f(1,t)}) -1$ and $96t/100 \le n \le t$.
\end{lemma}
\begin{proof}
Denote $G\defeq(V,E)\defeq\kg[t][0]$ and $A\defeq\actno[t][0]$.
By \cref{lem:kg-actnum}, the number of active nodes in $G$ is
at least $96t/100\geq t$, which implies that $|A| \ge 96t/100$. 
Thus,  after operation $t$,  $n \ge 96t/100$.

Without loss of generality, we assume that $G[A]$ has exactly one connected component: If this is not the case, let $C_1, \ldots, C_s$ be the connected components of $G[A]$ and observe that we may insert a path connecting the connected components by inserting
into $G[A]$ edges  $(v_1,v_2), \ldots, (v_{s-1}, v_s)$ of length $1$, where $v_i \in C_i$ are arbitrary vertices. This increases the maximum degree of nodes in $G$ by at most $2$ and
the shortest-path metric $M$ on the resulting graph that is constructed in this way is an augmented graph metric for $t$. As $t$ is clean, Lemma~\ref{lem:metric-consistent} shows that $M$ is consistent.

	Let $x \in V$ be any node. By \cref{lem:max-degree}, the number $n^{(i)}$ of nodes that have distance $i$ to $x$ is at most $\sum_{j \in [i]} (100f(1,t)+2)^j \leq (100f(1,t)+2)^{i+1}$. 
Consider the largest $i$ such that $(100f(1,t)+2)^{i+1} < n$. It follows that there exists a node at distance $i+1$ to $x$. Furthermore $(100f(1,t)+2)^{i+2} \ge n^{(i+1)} \ge n$, which implies that
$i+2 \ge  \log_{100f(1,t)+2} n.$
As $f(1,t) \ge 1$ for all values of $t$, there exists a shortest path $P$ starting at $x$ of length at least $i+1 \geq \log_{100f(1,t)+2}n - 1 \geq (\log n/ \log{(102f(1,t))}) -1 =: \ell$. It follows that the diameter is at least $\ell$.
On the other hand, $M$ can be extended by adding an edge between any pair of active nodes, resulting in the
consistent metric  $M_{\mathrm{uni}}$. For this metric  the diameter of $G$ is 1.
As the algorithm cannot tell whether $\ell$ or 1 is the correct answer, and it always has to output a value that is as least as large as the correct answer, it will output at least $\ell$. Thus, the approximation ratio is at least $\ell \ge \log (96t/100)/ \log{102f(1,t)} -1$.
\end{proof}
Note that this implies a lower bound for the approximation ratio for $1$-center, $1$-sum-of-radii, and $1$-sum-of-diameter.

\begin{proof}[of \cref{thm:det-1lb}]
Let $t\in\N$. By Lemma~\ref{lem:mpt-det-low-adv-prop}, there is
a value $t'$ with $t<t'\le2t'$ such that $t'$ is a clean operation.
Let $n$ be the number of active points at iteration $t'$. By \cref{lem:kg-actnum}
we know that $t'\ge n\ge96t'/100$. Note that hence $t'\le2n$.

Recall that we assumed that the function $f(k,n)$ is non-decreasing
in $n$ (for any fixed $k$). Suppose that after operation $t'$ we
query the solution value of an algorithm for $1$-center, $1$-sum-of-radii,
or $1$-sum-of-diameter. By Lemma~\ref{lem:diam} its approximation
ratio is at least $\frac{\log(96t'/100)}{\log(102f(k,t'))} -1\ge\frac{\log(96n/100)}{\log(102f(k,2n))} -1=\Omega\left(\frac{\log(n)}{\log(f(k,2n))}\right)$.
Suppose that instead we query the solution from the algorithm for
$(1,p)$-clustering. By Lemma~\ref{lem:det-low-1pclus}, its approximation
ratio is at least

\begin{alignat*}{1}
\left[\frac{\log(t'/4g)}{p+\log(100f(1,t'))}\right]^{p}/4 & \ge\left[\frac{\log(n/4g)}{p+\log(100f(1,2n))}\right]^{p}/4\\
 & \ge\left[\frac{\log(n)-\log(4g)}{p+\log(100f(1,2n))}\right]^{p}/4\\
 & \ge\left[\frac{\log(n)}{1.1p+1.1\log(100f(1,2n))}\right]^{p}/4\\ 
 & \ge\left[\frac{\log(n)}{1.1p+1.1(7+\log(f(1,2n))}\right]^{p}/4\\ 
 & \ge\left[\frac{\log(n)}{1.1p+8+1.1\log(f(1,2n))}\right]^{p}/4\\ 
 & \ge\left[\frac{\log(n)}{2p+8+2\log(f(1,2n))}\right]^{p}/4\\ 
 & =\Omega\left(\left(\frac{\log n}{2p+8+2\log f(1,2n)}\right)^{p}\right)
\end{alignat*}
using that $g=O(1)$. Hence, for $k$-median and $k$-means if we
take $p=1$ and $p=2$, respectively, this yields bounds of $\Omega\left(\frac{\log n}{10+2\log f(1,2n)}\right)$
and $\Omega\left(\left(\frac{\log n}{12+2\log f(1,2n)}\right)^{2}\right)$,
respectively.

\end{proof}

\subsubsection{Proof of \cref{lem:mpt-det-low-kcen}}

\begin{lemma}[\cref{lem:mpt-det-low-kcen}]
Let $k \ge 2$.
Consider any dynamic algorithm for maintaining an approximate $k$-center solution of a dynamic point set  that (1) queries amortized $f(k,n)$ distances per operation, where $n$ is the number of current points,
 and (2) outputs at most $g(k) \in O(k)$ centers.
For any $t\geq 2$ such that $t$ is a clean operation, the approximation factor of the algorithm's solution
(with respect to an optimal $k$-center solution) against an
adaptive adversary right after operation $t$ is at least
$\Omega\left(\min\left\{ k,\frac{\log n}{k\log{f(k,2n)}}\right\} \right)$.
\end{lemma}
\begin{proof}
Denote $G\defeq(V,E)\defeq\kg[t][0]$ and $A\defeq\actno[t][0]$.
By \cref{lem:kg-actnum}, the number of active nodes in $G$ is
at least $96t/100\geq t$, which implies that $|A| \ge 96t/100$. 
Thus,  after operation $t$,  the number $n$ of current points is at least  $96t/100$.

Without loss of generality, we assume that $G[A]$ has exactly one connected component: If this is not the case, let $C_1, \ldots, C_s$ be the connected components of $G[A]$ and observe that we may insert a path connecting the connected components by inserting
into $G[A]$ edges  $(v_1,v_2), \ldots, (v_{s-1}, v_s)$ of length $1$, where $v_i \in C_i$ are arbitrary vertices. This increases the maximum degree of nodes in $G$ by at most $2$ and
the shortest-path metric $M$ on the resulting graph that is constructed in this way is an augmented graph metric for $t$. As $t$ is clean, Lemma~\ref{lem:metric-consistent} shows that $M$ is consistent.

Let $x \in V$ be any node. By \cref{lem:max-degree}, the number $n^{(i)}$ of nodes that have distance $i$ to $x$ is at most $\sum_{j \in [i]} (100f(1,t)+2)^j \leq (100f(k,t)+2)^{i+1}$. 
Consider the largest $\ell$ such that $(100f(k,t)+2)^{\ell+1} < n$. It follows that there exists a node $z$ at distance $\ell+1$ to $x$. Furthermore $(100f(k,t)+2)^{\ell+2} \ge n^{(\ell+1)} \ge n$, which implies that
$\ell+2 \ge  \log_{100f(k,t)+2} n.$

Let $S = \{ s_1, \ldots, s_{g(k)} \}$ be the solution of the algorithm. For $i \in [\ell+1]$, let us define $V^{(i)}$ to be the set of vertices $v$ in $G[A]$ with $d_{G[A]}(x,v) = i$. By pigeon hole principle, there must exist a consecutive sequence $(i_1, \ldots, i_m)$ so that $m \geq \ell / (g(k)+1)$ and, for all $i \in \{ i_1, \ldots, i_m \}$, $S \cap V^{(i)} = \emptyset$. Let $k' =  \min\{3k-1, m/2\}$. Consider the metric $M_{i_1,i_{k'}}(x)$ and let $y_{i_1}, \ldots, y_{k'}$ be elements from the respective sets $V^{(i_1)}, \ldots,V^{(i_{k'})}$.

The algorithm's solution $S$ has cost at least $k'$ because $\dsp[G]{S}{V_{i_{k'}}} \geq k'$. The solution $\{ y_{3j-2} \mid j \in \N \wedge 3j-2 \in [k'] \}$ is optimal and has cost $1$. It follows that the approximation factor of $S$ is greater than or equal to $k' = \min\{3k-1, m/2\}$.

Let $n$ be the number of points at iteration $t$. We calculate that
\begin{alignat*}{1}
\min\left\{ 3k/2,m/2\right\}  & \ge\Omega\left(\min\left\{ k,m\right\} \right)\\
 & \ge\Omega\left(\min\left\{ k,\frac{\ell}{g(k)+1}\right\} \right)\\
 & \ge\Omega\left(\min\left\{ k,\frac{1}{g(k)}\left(\frac{\log n}{\log{(102f(k,t))}}-1\right)\right\} \right)\\
 & \ge\Omega\left(\min\left\{ k,\frac{\log n}{k\log{(102f(k,2n))}}\right\} \right)\\
 & \ge\Omega\left(\min\left\{ k,\frac{\log n}{k\log102+k\log{f(k,2n)}}\right\} \right)\\
 & \ge\Omega\left(\min\left\{ k,\frac{\log n}{k\log{f(k,2n)}}\right\} \right)
\end{alignat*}
using that $g(k)=O(k)$.
\end{proof}

\section{Algorithm for $k$-Sum-of-Radii and $k$-Sum-of-Diameter}

\SetKwFunction{FnPrimalDual}{PrimalDual}
\SetKwFunction{FnPrune}{Prune}

In this section, we prove the correctness of \cref{alg:pd} for the $k$-sum-of-radii and the $k$-sum-of-diameter problem as described in \cref{sec:primal-dual}. To compute a solution in the static setting, the algorithm is invoked as $\FnPrimalDual(P, \emptyset, R, z, k, \epsilon, 0, 0, 0)$, where $P$ is the set of input points, $R$ is the set of radii, $z = \epsilon \OPT' / k$ is the facility cost and $k$ is the number of clusters. \Cref{alg:pd} is a pseudo-code version of the algorithm described in \cref{sec:primal-dual}.

\begin{algorithm}
	\SetKwData{null}{null}
	\KwData{point set $P$, unassigned point sets $U = \{ U_0, \ldots \}$, radii set $R$, facility cost $z$, number of clusters $k$, precision $\epsilon$, primal vector $x = \{ x^{(r)}_p \mid r \in R \wedge p \in P \}$, dual vector $Y = \{ y_p \mid p \in P \}$}
	\Fn{\FnPrimalDual($P, U, R, z, k, \epsilon, x, y, i$)}{
		\While{$U_i \neq \emptyset$}{
			$p_i \gets$ uniformly random point from $U_i$ \;
			$\delta_i \gets \max ( \{ 0 \} \cup \{\delta' \mid \delta' \in \R \wedge \forall r' \in R : \sum_{p' \in P : d(p_i, p') \le r'} y_{p'} + \delta' \leq r'/2 + z \} )$ \;
			\If{$\delta_i > 0$}{
				$r_i \gets \max \{r' \mid r' \in R \wedge \sum_{p' \in P : d(p_i, p')\le r'} y_{p'} + \delta_i = r'/2 + z \}$ \;
			}
			\Else{
				$r_i \gets \max \{r' \mid r' \in R \wedge \sum_{p' \in P : d(p_i, p')\le r'} y_{p'} \geq r'/2 + z \}$ \;
			}
			$y_{p_i} \gets \delta_i$; $x^{(2r_i)}_{p_i} \gets 1$ \;
			$U_{i+1} \gets U_i \setminus \{ p' \mid p' \in U_i \wedge d(p_i, p') \leq 2r_i \}$ \;
			$i \gets i + 1$ \;
			\If{$i > (2k/\epsilon)^2$}{
				\Return{``$\OPT' < \OPT$''} \;
			}
		}
		$\bar{S} \gets \emptyset$; $S \gets$ sort $(p_j, r_j)_{j \in [i-1]}$ non-increasingly according to $r_j$ \;
		\ForAll{$(p,r) \in S$}{
			\If{$\nexists (p_j,r_j) \in \bar{S} : d(p, p_j) \leq r + r_j$}{
				$\bar{S} \gets \bar{S} \cup \{ (p, 3r) \}$ \;
			}
		}
		\Return{$\bar{S}, U, x, y, r, i$} \;
	}
	\caption{\label{alg:pd} Pseudo code of the primal-dual algorithm for $k$-sum-of-radii as described in \cref{sec:primal-dual}.}
\end{algorithm}

\subsection{Proof of \cref{lem:pd-iteration-time}}

We bound the running time of a single primal-dual step. This implies \cref{lem:pd-iteration-time}.

\begin{lemma}[\cref{lem:pd-iteration-time}]
	The running time of the $i$\xth/ iteration of the while-loop in \FnPrimalDual is $O(ik / \epsilon + \lvert U_i \rvert)$.
\end{lemma}
\begin{proof}
	In each iteration, at most one entry of $y$, i.e., $y_{p_i}$, increases. Therefore, at most $i-1$ entries of $y$ are non-zero at the beginning of iteration $i$. By keeping a list of non-zero entries, each sum corresponding to a constraint of the dual program can be computed in time $O(i)$. Since $\lvert R \rvert \leq k / \epsilon$, it follows that $\delta_i$ and $r_i$ can be computed in time $O(ik/\epsilon)$. To construct $U_{i+1}$, it is sufficient to iterate over $U_i$ once.
\end{proof}

\subsection{Proof of \cref{lem:pd-iterations-bound}}

We show that our choice of $z = \epsilon \OPT' / k$ will result in a solution $\bar{S}$ if $\OPT \leq \OPT'$.

\lemPdIterationsBound*
\begin{proof}
	Since, by weak duality, $\sum_{p \in P} y_p$ is a lower bound to the optimum value of $(P)$, we bound the number of iterations that are sufficient to guarantee $\sum_{p \in P} y_p > \OPT'$. Call an iteration of the while-loop in \FnPrimalDual \emph{successful} if $\delta_i > 0$, and \emph{unsuccessful} otherwise. First, observe that for each successful iteration $i$, the algorithms increases $y_{p_i}$ by at least $z/2$. This is due to the fact that all values in $R$ are multiples of $z$, and therefore $r/2 + z$ is a multiple of $z/2$ for any $r \in R$. Since the algorithm increases $y_{p_i}$ as much as possible, all $y_{p_i}$ are multiples of $z/2$. Therefore, after $\OPT' / (z/2) +1= 2k / \epsilon +1$ successful iterations, it holds that $\OPT \geq \sum_{p \in P} y_p > \OPT'$.
	
	Now, we prove that for each successful iteration, there are at most $\lvert R \rvert$ unsuccessful iterations. Then, it follows that after $((2k/\epsilon)+1) \cdot \lvert R \rvert \leq (2k/\epsilon)^2 + (2k/\epsilon)$ iterations, $\OPT > \OPT'$. Let $i$ be an unsuccessful iteration. The crucial observation for the following argument is that for the maximum $r_i \in R$ so that $(p_i, r_i)$ is at least half-tight and for any $j < i$ so that $d(p_i, p_j) \leq r_i$ and $y_{p_j} > 0$, we have $r_j < r_i$: otherwise, $p_i$ would have been removed from $U_j$. On the other hand, such $p_j$ must exist because the dual constraint $(p_i, r_i)$ is at least half-tight but $y_{p_i} = 0$. Now, we \emph{charge} the radius $r_i$ to the point $p_j$ and observe that we will never charge $r_i$ to $p_j$ again. This is due to the fact that all points $p \in U_{i-1}$ with distance $d(p_j, p) \leq r_i$ are removed from $U_i$ because $d(p_i, p) \leq d(p_i, p_j) + d(p_j, p) \leq 2r_i$. Therefore, for any successful iteration $j$ and any $r \in R$, there is at most one $i > j$ so that $r = r_i$ is charged to $p_j$, each accounting for an unsuccessful iteration. 
\end{proof}

\subsection{Proof of \cref{lem:mpt-pd-pruning-time}}

We argue that the pruning step has running time $\textrm{poly}(k/\epsilon)$ to prove \cref{lem:mpt-pd-pruning-time}.

\begin{lemma}[\cref{lem:mpt-pd-pruning-time}]
	\label{lem:pd-pruning-time}
	The for-loop in \FnPrimalDual has running time $O((k/\epsilon)^4)$.
\end{lemma}
\begin{proof}
	Since $i \leq 2(2k/\epsilon)^2$, it holds that $\lvert S \rvert \leq \lvert \{ p \in P \mid \exists  r \in R : x^{(r)}_p > 0 \} \rvert \leq 2(2k/\epsilon)^2$. Therefore, sorting $S$ requires at most $O((2k/\epsilon)^2 \log (2k/\epsilon))$ time. In each iteration of the loop, it suffices to iterate over $\bar{S}$. Since $\lvert \bar{S} \rvert \leq \lvert S \rvert$, the claim follows.
\end{proof}

\subsection{Proof of \cref{lem:pd-offline-time}}

\lemPdOfflineTime*
\begin{proof}
	Recall that we sort points in $S$ non-increasingly. Therefore, for every point $p$, there exists only one $(p,r) \in S$ that is inserted into $\bar{S}$. It follows that $\lvert \bar{S} \rvert \leq \lvert S \rvert \leq (2k/\epsilon)^2$ and since $\lvert \hat{S} \rvert \leq k$, it suffices to iterate over $\bar{S}$ for every $(\hat{p}, \hat{r}) \in \hat{S}$.
\end{proof}

\subsection{Proof of \cref{lem:pd-cost-a}}

We turn to the feasibility and approximation guarantee of the solution that is computed by \cref{alg:pd}. First, we observe that the dual solution is always feasible.

\begin{lemma}
	\label{lem:always-dual-feasible}
	During the entire execution of \cref{alg:pd}, no dual constraint $(p, r)$ is violated.
\end{lemma}
\begin{proof}
	For the sake of contradiction, let $i$ be the first iteration of the while-loop in \FnPrimalDual after which there exists $(p,r)$ such that $\sum_{p' \in P : d(p, p') \leq r} > r + z$. From this definition, it follows that $d(p, p_i) \leq r$ as only $y_{p_i}$ has increased in iteration $i$. By the triangle inequality, for all $p' \in P$ so that $d(p,p') \leq r$, we have that $d(p_i,p') \leq 2r$. Therefore, the dual constraint $(p_i, 2r)$ is more than half-tight:
	\begin{align*}
		\sum_{p' \in P : d(p_i, p') \leq 2r} y_{p'}
		\geq \sum_{p' \in P : d(p, p') \leq r} y_{p'}
		> r + z
		= \frac{2r}{2} + z .
	\end{align*}
	This is a contradiction to the choice of $(p_i, r_i)$.
\end{proof}

The primal solution is also feasible to the LP, and its cost is bounded by $6\OPT'$.

\lemPdCostA*
\begin{proof}
	Since $U_i = \emptyset$, $x$ is feasible at the end of \cref{alg:pd}. By the construction of $\bar{S}$, it follows that $\bar{x}$ is feasible. By \cref{lem:always-dual-feasible}, $y$ is always feasibile for (D). It remains to bound the cost of $\bar{S}$.

	By the construction of $\bar{S}$, for any $(p, r) \in \bar{S}$, $x^{(r/3)}_p = 1$ and $(p, r/3)$ is at least half-tight in $y$. In other words, $r + z \le 6 \cdot (r/(2 \cdot 3) + z) = 6 \sum_{d(p,p') \leq r/3} y_p$. Let $\bar{S}(p,r) = \{ p' \in P \mid d(p, p') \leq r \}$. For all $(p_1,r_1), (p_2,r_2) \in \bar{S}$, $p_1 \neq p_2$, we have that $\bar{S}(p_1,r_1/3)$ and $\bar{S}(p_2,r_2/3)$ are disjoint due to the construction of $\bar{S}$. Therefore, it holds that
	\begin{equation*}
		\sum_{p\in P}\sum_{r\in R}\bar{x}_{p}^{(r)}(r+z)
		= \sum_{(p,r) \in \bar{S}} (r+z)
		\leq 6\cdot\sum_{p\in P}y_{p} .
	\end{equation*}
	By weak duality, $6\cdot\sum_{p\in P}y_{p} \leq 6\OPT'$.
\end{proof}

\subsection{Proof of \cref{lem:pd-cost-b}}

Finally, we prove that the pruned solution is a feasible $k$-sum-of-radii solution with cost bounded by $(13.008 + \epsilon)\OPT'$.

\lemPdCostB*
\begin{proof}
	First observe that the cost of $\hat{S}$ is bounded by $2\cdot 3.504 \cdot \OPT$ since we can construct a solution with cost at most $2\OPT$ that covers $\bar{C}$ using only centers from $\bar{C}$: for each point $p \in \bar{C}$, take the point $p' \in \OPT$ covering $p$ with some radius $r'$, and select $p$ with  radius $2r'$.

	Let $p \in P$, let $(\bar{p}, \bar{r}) \in \bar{S}$ so that $d(p, \bar{p}) \leq \bar{r}$ and let $(\hat{p}, \hat{r}) \in \hat{S}$ that was chosen to cover $\bar{p}$. By the triangle inequality, $d(p, \bar{p}) \leq \hat{r} + \bar{r}$. 
	Let $(\hat{p}, \tilde{r})$ be the corresponding tuple in $\tilde{S}$. By the construction of $\tilde{S}$, we have that $\hat{r} + \bar{r} \leq \tilde{r}$. Therefore, $\tilde{S}$ is feasible. It follows that the cost of $\tilde{S}$ is at most

	\begin{equation*}
		\sum_{(\hat{p},\hat{r})\in \hat{S}}\hat{r}+\sum_{p\in P}\sum_{r\in R}\bar{x}_{p}^{(r)}(r+z)
		\leq (2\cdot3.504 + 6 + \epsilon) \OPT'. \qedhere
	\end{equation*}
\end{proof}

\subsection{Proof of \cref{thm:pd-main}}

\thmPdMain*
\begin{proof}
	Consider a fixed choice of $\OPT'$. This assumption will be removed at the end of the proof. We invoke \cref{alg:pd} on the initial point set as for the static setting. Consider an operation $t$, and let $\breve{S}, \breve{U}, \breve{x}, \breve{y}, \breve{r}, \breve{i}$ be the state before this operation.

	If a point $p$ is inserted, the algorithm checks, for every $j \in \{ 1, \ldots, i-1 \}$, if $d(p_j, p) \leq 2r_i$. If this is not the case for any $j$, $p$ is added to $\breve{U_j}$ and the algorithm proceeds. Otherwise, the algorithm stops. If the last check fails and $i > (2k/\epsilon)^2$, the algorithm stops, too. Otherwise, it runs $\FnPrimalDual(P, \breve{U}, R, z, k, \epsilon, x, y, i)$ with the updated $\breve{U}$. In any case, the point $p$ deposits a budget of $2k/\epsilon$ tokens for each possible $U_i$, $i \in [2k/\epsilon]$, i.e., $(2k / \epsilon)^2$ tokens in total. By \cref{lem:pd-iteration-time,lem:pd-pruning-time}, the total running time is $O((2k/\epsilon)k/\epsilon + 0 + (k/\epsilon)^4 + (2k / \epsilon)^2) = O((k/\epsilon)^4)$.

	If a point $p$ is deleted, the algorithm deletes $p$ from all $\breve{U_i}$ it is contained in. If for all radii $r \in R$ we have that $\breve{x}^{(r)}_p = 0$ then the algorithms stops. Otherwise, let $j$ be so that $p \in \breve{U}_j \setminus{\breve{U}}_{j+1}$, i.e., $p$ is the center of the $j$\xth/ cluster. The algorithm sets, for all $r \in R$ and all $j' \geq j$, $\breve{x}^{(r)}_{p_{j'}} = 0$, $\breve{y}_{p_{j'}} = 0$ and, for all $j' > j$, $\breve{U}_{j'} = \emptyset$. Then, it calls $\FnPrimalDual(P, \breve{U}, R, z, k, \epsilon, \breve{x}, \breve{y}, j)$. By \cref{lem:pd-iteration-time,lem:pd-pruning-time}, the total running time is $O((2k/\epsilon)^2 k/\epsilon + (2k/\epsilon)\lvert U_j \rvert + (k/\epsilon)^4)$.
	
	The correctness of the algorithm follows from the fact that if $\OPT' \geq \OPT$, the algorithm produces a feasible solution irrespective of the choice of the $p_i$, $i \in [2k / \epsilon]$, by \cref{lem:pd-cost-b} and observing that the procedure described above simulates a valid run of the algorithm for $P = P_{t-1} \cup \{ p \}$ and $P = P_t \setminus \{ p \}$, respectively. Finally, we prove that the expected time to process all deletions up to operation $t$ is bounded by $O(t \cdot 2k/\epsilon)$. The argument runs closely along the running time analysis in~\cite{ChanGS18}.
	
	Let $t' \leq t$, $j \in [2k/\epsilon]$ and let $\bar{U}^{(t')}_j$ be the set $U^{(t')}_i$ that was returned by \FnPrimalDual after the last call that took place before operation $t'$ so that the argument $i$ is such that $i \leq j$. Note that this is the last call to \FnPrimalDual before operation $t'$ when $U_j$ is reclustered. We decompose $U^{(t')}_j$ into $A^{(t')}_j = U^{(t')}_j \setminus \bar{U}^{(t')}_j$ and $B^{(t')}_j = U_j \cap \bar{U}^{(t')}_j$ and define the random variable $T^{t'}_i$, where $T^{t'}_i = \lvert B^{(t')} \rvert$ if operation $t'$ deletes center $p_i$ and $T^{t}_i = 0$ otherwise. Next, we bound $E[\sum_{t' < t} \sum_{i \in [2k / \epsilon]} T^{(t')}_i]$. For $t' <t$ and $i \in [2k / \epsilon]$, consider $E[T^{(t')}_i]$. Since $p_i$ was picked uniformly at random from $B^{(t')}$, the probability that operation $t'$ deletes $p_i$ is $1 / \lvert B^{(t')} \rvert$. Therefore, $E[T^{t'}_i] = 1$. By linearity of expectation, $E[\sum_{t' < t} \sum_{i \in [2k / \epsilon]} T^{(t')}_i] \leq 2k/\epsilon$. If operation $t'$ deletes $p_j$, $U_j$ is reclustered at operation $t'$ and any point in $A$ is not in $A^{(T^{t'}_i)}$ for any $t'' \geq t'$. Therefore, each point $p \in A$ can pay $(2k/\epsilon)$ tokens from its insertion budget if $p_j$ is deleted. The expected amortized cost for all operations up to operation $t$ is therefore at most $O((2k/\epsilon)^2 k/\epsilon + (2k/\epsilon)^2 + (k/\epsilon)^4) = O((k/\epsilon)^4)$.

	Now, we remove the assumption that $\OPT'$ is known. Recall that $d(x,y) \geq 1$ for every $x,y \in X$, and $\dmax = \max_{x,y \in X} d(x,y)$. For every $\Gamma \in \{ (1+\epsilon)^i \mid i \in [\lceil \log_{1+\epsilon} (k \dmax) \rceil]$, the algorithm maintains an instance of the LP with $\costbound = (1+\epsilon)^\Gamma$. After every update, the algorithm determines the smallest $\Gamma$ for which a solution is returned. Recall that the algorithm from \cite{CharikarP04} takes time $O(n^{O(1/\epsilon)})$. The total expected amortized cost is $O(k^{O(1/\epsilon)} \log \dmax)$.
\end{proof}

\section{Algorithm for $k$-center}

\subsection{Proof of \cref{lem:mpt-rnd-kcen-linup}}

Our strategy is to maintain a set of at most $k+1$ points $C\subseteq P$
so that any two points $c,c'\in C$ satisfy that $d(p,p')>2\OPT'$.
If $|C|=k+1$ then this asserts that there can be no solution with
value $\OPT'$. If $|C|\le k$ then we will ensure that $C$ forms
a feasible solution of value $2\OPT'$. 

The main idea is as follow: We say that two points $p,p'\in P, p \ne p',$ are \emph{neighbors} if $d(p,p')\le2\OPT'$. We dynamically maintain a counter $a_{p}\in\{0,...,k+1\}$ for each point $p\in P$ such that $a_{p}$ denotes the number of neighbors of $p$ in $C$. 
As shown in the next lemma, this is sufficient to maintain a suitable set $C$ and leads to an algorithm with  amortized update time $O(n+k)$. 

\lemRndKcenLinup*
\begin{proof}
The algorithm maintains a set of at most $k+1$ points $C\subseteq P$
so that any two points $c,c'\in C$ satisfy that $d(p,p')>2\OPT'$.
If $|C|=k+1$ then this asserts that there can be no solution with
value $\OPT'$. If $|C|\le k$ then we will ensure that $C$ forms
a feasible solution of value $2\OPT'$. 
Our algorithm maintains as data structure for each point $p \in P$ a counter
$a_{p}\in\{0,...,k+1\}$ such that $a_{p}$
denotes the number of neighbors of $p$ in $C$. In particular, if
$p\in C$ then it will always hold that $a_{p}=0$.

\emph{Insertions.} When a new point $p$ is inserted into $P$, we compute $a_{p}$
by determining its distance to all centers. If $a_{p}=0$ and $|C|\le k$ then we add $p$ into
$C$. Then, we increment (by one) the counter $a_{p'}$ for each neighbor
$p'\in P$ of $p$.

\emph{Deletions.} Suppose that a point $p\in P$ is deleted. If $p\not\in C$
we do not change anything else. If $p\in C$, we decrement
(by one) the counter of each neighbor $p'\in P$ of $p$. 
Then, we iterate over $P\setminus C$ in arbitrary order.
If for a point $p'\in P\setminus C$ we have that $a_{p'}=0$ and $|C|\le k$,
then we add $p'$ into $C$ and increment the counter $a_{p''}$ for
each neighbor $p''\in P$ of $p'$. Then we consider the next point
in $P\setminus C$. Note that it could be that for a point $p'\in P\setminus C$ its counter
$a_{p'}$ equals to 0 immediately after $p$ was deleted, but by the time that 
$p'$ is considered, i.e. $a_{p'}$ is checked, it holds that $a_{p'}>0$. In this case $p'$ does, of course, not become a center.

\emph{Preprocessing.} When started on a non-empty set $P$, the algorithm simply inserts every element of $P$ into an initially empty data structure.

\emph{Correctness.} 
Note that the algorithm maintains the following invariant.

\emph{(INV) If $|C|\le k$, then for all points $p \in P \setminus C$ it holds that $a_p >0$.
}

This can be shown by an easy induction over the number of operations.
It follows that whenever $|C|\le k$, then $C$ forms a feasible solution with value at most $2\OPT'$ as every point in $P$ is within distance at most $2\OPT'$ of a point in $C$. 
If $|C|= k+1$ then all points in $C$ are at distance more than $2\OPT'$ of each other, and hence we can assert that $\OPT>\OPT'$.

\emph{Running time analysis.} 
We now show that our amortized update time is $O(n+k)$. Let  $c_t$ be the number of centers after update operation $t$ and $n_t$ be the size of $P$ after update operation $t$. Operation 0 is the preprocessing and $n_0$ is the size of the initial set $P$ and $c_0$ the number of centers after preprocessing.

\emph{Worst-case running time.}
We first analyze the worst-case running time of update operation $t$ for $t \ge 0$.
When update operation $t$ inserts a point $p$, we need 
$\Theta(k)$
time	to initialize the counter of $p$ by counting its neighbors in $C$.
Additionally, if $p$ becomes a center, it incurs a cost of $\Theta(n_t)$ to update the counters of its neighbors.
Thus insertions have $\Theta(n_t)$ worst-case time.
Furthermore, only insertions that create a new center take time $\Theta(n_t)$, 
all other insertions take time 
$\Theta(k)$.
This implies that the worst-case preprocessing time is 
$\Theta(n_0 k)$.

Consider next the case that update operation $t$ deletes a point $p$.
When a point $p\in P\setminus C$ is deleted, the running time is constant. If a center $p$ is deleted, the algorithm incurs a cost of $O(n_t)$ to update the counters of its neighbors. Additionally, all the neighbors $q$ with $a_q = 0$ are placed on a queue and processed one after the other. 
	If a point $q$ has $a_q = 0$ when it is pulled off the queue, it becomes a center and there is cost
	of $O(n_t)$ to update the counters of all neighbors of $q$. If $a_q > 0$, $q$ does not become a center and only $O(1)$ time is spent on $q$. Thus, the worst-case time per delete is $O((1+\delta_t )\cdot n_t)$, where $\delta_t = c_t - c_{t-1}$. 

\emph{Token-charging scheme.}	
	We next show how to pay for these operations using amortized analysis. We use a token-base approach such that each token can be used to pay for $O(1)$ amount of work.
	The preprocessing phase is charged $(2 k+1) n_0$ tokens, $k n_0$ of its tokens are used to pay for the preprocessing time and the remaining $(k+1) n_0$ of the tokens are placed on the bank account.
	It follows that the amortized preprocessing time is $\Theta(k n_0)$, as is its worst-case running time.
	
	Consider the $t$-th update operation.
	Each update operation is charged $n_t + k$ tokens, where $n$ is the current number of points. For insertions we use at most $n_t$ tokens to pay for the operation and place $k$ tokens on a bank account. 
	It follows that the amortized insertion time is $O(n_t + k) = O(n + k).$
	
	Assume next that the $t$-th update operation is a deletion. 
	We will show below that for any $t \ge 1$ the bank account contains at least $\delta_t n_{t-1}$ tokens right before operation $t$.
	We use the $n_t+k$ tokens charged to the deletion plus $\delta_t n_t$ tokens from the bank account to pay for the deletions and we put any leftover tokens on the bank account. 
	As the worst-case running time of a deletion is $O((1+\delta_t )\cdot n_t)$, it follows that
	the amortized time of the deletion is $O((1+\delta_t )\cdot n_t - \delta_t n_t + (n_t + k)) = O(n_t + k) = O(n + k)$.

It remains to prove that the bank account contains at least $\delta_t n_{t-1}$ tokens right before the $t$-th operation. Note that $\delta_t = c_{t} - c_{t-1} \le k +1- c_{t-1}$.
Thus it suffices to show that the bank account  contains at least $(k+1-c_{t-1}) \cdot n_{t-1}$ tokens before operation $t$ for $t \ge 1$, i.e. at least $(k+1-c_{t}) \cdot n_{t}$ tokens \emph{after} operation $t$ for $t \ge 0$. 
To show this we perform an induction on $t$.
	For $t = 0$ note that preprocessing placed $(k+1) n_0\ge (k+1 - c_0) n_0 $ token on the bank account and thus, the claim holds.
	Assume the claim was true after $t-1$ operations and we want to show that it holds after $t$ operation.	
We consider the following cases.

	(a) If operation $t$ is an insertion, then the number of tokens on the bank account increases by  $k$, while $c_t$ might be unchanged or increased by 1. Thus $c_t \ge c_{t-1}$ and $n_t = n_{t-1} + 1$.
	It follows that the number of tokens on the bank account is at least $(k +1- c_{t-1}) \cdot n_{t-1} +k\ge (k+1 - c_{t}) \cdot (n_t-1)  +k\ge (k +1-c_t) \cdot n_t + k - (k+1 - c_t) \ge (k +1-c_t) \cdot n_t$ as $c_t \ge 1$ after the insertion.
	
	(b) If operation $t$ is a deletion, then $n_t = n_{t-1}-1$,
	the operation is charged $n_t + k$ and the operation consumes $(1+\delta_t) n_t$ tokens.
	Recall that $c_{t-1} = c_{t} - \delta_t$ and that $c_{t-1} \le k +1 \le 2k+1$.
	Thus the number of tokens on the bank account after operation $t$ is at least $(k+1 - c_{t-1}) \cdot n_{t-1} + n_t + k - (1+\delta_t) n_t= (k+1 - c_t + \delta_t +1 - 1 - \delta_t)\cdot n_t + 2k +1- c_{t-1}
	\ge (k+1 - c_t) \cdot n_t$.
	
	This completes the induction and, thus, the running time analysis.
	
	It follows directly by construction of the algorithm that any point in $P$ becomes a center at most once, and then stays a center until it is deleted. Note that after a point $p$ has become a center, it might be that at some point $\OPT>\OPT'$ is reported and no center is output. However, if later on again a solution is output, then $p$ is again a center in this solution (unless $p$ has already been deleted).
\end{proof}

\subsection{Proof of \cref{lem:mpt-rnd-kcen-delonly}}

\lemRndKcenDelOnly*
\begin{proof}
	Suppose that the current solution is queried. If each bucket is small
	(which happens if we stopped the last partial rebuild because
	$U_{i}=\emptyset$ for some $i$), we output $C$. By construction,
	in this case each point $p\in P$ is contained in some cluster $P(p')$
	for some point $p'$, and we output one point from each such cluster
	$P(p')$. Each cluster has a diameter of at most $4\OPT'$, and therefore
	for each point $p\in P$ there is some center $c\in C$ such that
	$d(p,c)\le4\OPT'$. In addition, we have that $|C|\le O(k\log n)$.
	On the other hand, if there is a large bucket $C^*$ we report that
	$\OPT>\OPT'$. This is justified since $C^*$ contains at least
	$k+1$ points such that for any two such distinct points $c,c'\in C^*$
	it holds that $d(c,c')>2\OPT'$.
\end{proof}

\section{Lower bound for randomized algorithms against an oblivious adversary}

\aw{In this section we prove \cref{lem:mpt-lb-rand-cost}, i.e., we show that if an algorithm returns a
$(\Delta-\epsilon)$-approximation (for some $\epsilon > 0$) of the cost of the optimal solution to $k$-center, then} it needs to ask
$\Omega(k)$ queries to the adversary for each insertion after the first $k$ insertions. \aw{Then we argue that we get similar lower bounds also for $k$-median, 
$k$-means, $k$-sum-of-radii, and $k$-sum-of-diameters.
}

To avoid confusion between the query operations that the adversary asks to the algorithm and the distance queries that the algorithm asks the adversary, we formalize the latter kind of queries using the notion of
distance oracles.
A \emph{distance oracle} is a black-box that returns the distance between two points that are given as input in constant time and its output is controlled by the adversary.
This is the only way that the algorithm can get information about the metric
space, with each distance query ``costing'' constant time for the algorithm. As the distance oracle needs to give consistent answers, the distance between two queried points has been \emph{fixed}, i.e. the adversary cannot change it later.
The answers of the distance oracle are determined by an \emph{adversary} who also determines the sequence of operations given to the algorithm. 
Thus there are two types of queries that should not be confused: (1) the $k$-center-cost query operations issued by the adversary and (2) the distance-oracle queries issued by the algorithm. 
The goal of the adversary is to maximize the running time of the algorithm. 
The general approach of the lower bound construction is to give an adversary that gives a sequence of operations and reveals the  metric in such a way that the algorithm has to ask $\Omega(k)$ distance oracle queries per operation.

We first present the general idea:
Assume the adversary uses a metric space $(X,d)$  with $\min_{x,y \in X} d(x,y) = 1$ and $\max_{x,y \in X} d(x,y) = \Delta$.
Suppose by contradiction that there is a dynamic randomized algorithm for $k$-center
for which the update time is at most $k/4$, and in addition the algorithm
needs at most $k/4$ time to report the value $\ALG$.

Consider the following sequence of operations and queries. As usual,
we assume that the adversary needs to define the operations without
seeing the random bits of the algorithm.  First, the adversary introduces
$k$ points $P_{0}$ that are at pairwise distance $\Delta$ to each other.
For simplicity, we assume that the algorithm knows their pairwise
distances, without having to query the distance between any pair of
them. Thus, their distances are fixed and cannot be changed by the adversary anymore.
Then, the adversary introduces a point $p_{1}$ and afterwards it issues as $k$-center-cost
query. 

If the algorithm issues at least $k$ distance-oracle queries for the insertion of $p_1$  with probability 1, then the claim that the algorithm spends time $\Omega(k)$ per insertion holds already. 
Otherwise there is some probability $q>0$ such that if $d(p_{1},p)=\Delta$
for each $p\in P_{0}$, then the algorithm queries the distance between
$p_{1}$ and at most $k-1$ points $P'_{0}\subseteq P_{0}$. 
Thus there is at least one point $p_1^*$ in $P'_{0}\subseteq P_{0}$ such that the distance 
$d(p_{1},p^{*}_1)$ has not yet been queried and, thus, not yet been fixed.
Let $\ALG$ denote the value of the approximate solution that the algorithm reports
in this case.  Consider the following two cases:

{\emph Case 1:} If $\ALG\ge \Delta$, then the adversary takes the point $p^{*}_1$ %
and defines $d(p_{1},p^{*}_1)=1$ and $d(p_{1},p)=\Delta$ for each $p\in P_{0}\setminus\{p^{*}_1\}$.
(Note that for this metric, with probability at least $q$ the algorithm
makes the same queries as for the metric in which $d(p_{1},p)=\Delta$
for each $p\in P_{0}$, i.e.~where $d(p_{1},p^{*}_1)=\Delta$.)
It follows that the optimum solution consists of all points in $P_0$ and has cost 1, i.e., 
$OPT=1\le\ALG/\Delta$ and hence the
approximation ratio of the algorithm is at least $\Delta$ (with probability
at least $q>0$).

{\emph Case 2:} If $\ALG<\Delta$ then the adversary defines that $d(p_{1},p)=\Delta$
for each $p\in P_{0}$, also for $p_1^*$ and, thus, the cost of the optimal solution
is $\Delta>\ALG$. Therefore, in this case the algorithm reports a wrong
upper bound for its cost which is a contradiction. 

Thus, the approximation
ratio of the algorithm is again at least $\Delta$. 
Finally, the adversary
deletes $p_{1}$ again and repeats the above with a new point $p_{2}$
and so on. 
For each new point $p_{1},p_{2},...$, there are three operations of
the adversary, namely an insertion, a $k$-center cost query, and a deletion. 
The running time of the algorithm is at most $k/4$
per operation, and thus, at most $3k/4$ per point $p_{i}$. 

We formalize this bound and its proof in the following theorem.

\begin{theorem}
	Let $\epsilon > 0$. Any (randomized) dynamic algorithm for finite metric spaces that can provide a $(\Delta - \epsilon)$-approximation to the $k$-center cost of its input issues amortized $\Omega(k)$ distance-oracle queries in expectation after $O(k^2)$ operations generated by an oblivious adversary.
\end{theorem}
\begin{proof}
	We use Yao's principle to prove the lower bound, i.e., we give a distribution of the input and determine the worst-case complexity of any deterministic algorithm for such an input distribution. Let $\ALG$ be any deterministic algorithm that \begin{inenum} \item takes as input a sequence of $t$ point insertions, deletions and cost queries, \item outputs, for each cost-query operation $t'$, a $(\Delta - \epsilon)$-approximation to the $k$-center cost of the set of input points after operation $t'-1$ and \item issues at most $tk / 256$ distance-oracle queries in expectation \end{inenum}. Since any dynamic algorithm that beats the claimed lower bound in the statement of the lemma can be turned into $\ALG$ by using it as a black-box, proving that $\ALG$ does not exist proves the claim. 
	
	We define the uniform distribution over a set $S$ of sequences, each consisting of $t \defeq k + 30k^2$ operations. For the sake of analysis, we describe $S$ by iteratively constructing an element from it according to the uniform distribution. Let $P_0$ be a set of $k$ points that have pairwise distance $\Delta$. The first $k$ operations insert the points in $P_0$. For each $i \in \{0, \ldots, 10k^2-1\}$, we choose a uniformly random point $p^{*}_i \in P_0$. Operation $k + 3i$ inserts a point $p_i$, where, for every $p \in P_0 \setminus \{p^{*}_i\}$, we have $d(p_i, p) = \Delta$. With probability $1/2$, we set $d(p_i, p^{*}_i) = 1$, and $d(p_i, p^{*}_i) = \Delta$ otherwise. Operation $k + 3i + 1$ is a $k$-center cost query, and operation $k + 3i + 2$ removes $p_i$. This way, every sequence $\sigma \in S$ is naturally partitioned into an \emph{initialization block} of $k$ operations and $10k^2$ \emph{small blocks} of $3$ operation each. As a sequence is determined by $10k^2$ many $2k$-ary choices, it follows that $\lvert S \rvert = (2k)^{10k^2}$.

	For the sake of simplicity, we reveal the identity of $P_0$ and the pairwise distances of points in $P_0$ to the algorithm. This gives additional information to the algorithm without any cost for the algorithm, i.e., it only ``helps'' the algorithm. By an averaging argument, the algorithm must query at most $tk / 128$ distances on at least half of the sequences from $S$. Let $T \subseteq S$ denote this subset of sequences. By another averaging argument, for at least half of the small blocks $i$ in a sequence from $T$, $\ALG$ queries at most $(tk/128) / k^2 \leq 32k^3 / (128k^2) = k / 4$ distances between $p_i$ and points in $P_0$. 
	
	Therefore, there must exist a $\sigma \in S$ and a corresponding small block $j$ so that $\ALG$ queries at most $k/4$ distances between $p_j$ and points in $P_0$. Let $\sigma' \in S$ be any sequence that equals $\sigma$ on the first $k + 3j$ operations (i.e., on its points and distances). When the $j$\xth/ cost query on $\sigma$ or $\sigma'$ is issued, the optimal cost can differ by a factor of $\Delta$: If $d(p_j, p^{*}_j) = 1$, the optimal cost of the instance after operation $k + 3i + 1$ is $1$; if $d(p_j, p^{*}_j) = \Delta$, the optimal cost of the instance is $\Delta$. To distinguish these two cases, the algorithm needs to determine whether there exists $p^{*}_j$ so that $d(p_j, p^{*}_j) = 1$. Since $p^{*}_j$ is chosen uniformly at random and the events $d(p_j, p^{*}_j) = 1$ and $d(p_j, p^{*}_j) = \Delta$ both have non-zero probability, the probability that $\ALG$ distinguishes these two events is at most $k/4 \cdot 1 / (\lvert P_0 \rvert - k/4) \leq k / 4 \cdot 4 / (3k) < 1$ by the union bound. This contradicts the assumption on $\ALG$.

	Finally, we observe that each of the $k^2$ points $p_i$ is chosen from a set of $k$ points, and therefore, a space of size $\lvert P _0 \rvert + k \cdot k^2  \in O(k^3)$ suffices to construct $S$. Choosing, for any $i,j$, $d(p_i, p_j) = 1$ if $p^{*}_i = p^{*}_j$ and $d(p_i, p_j) = \Delta$ otherwise, the space is metric.
	
Also note that Yao's minmax principle gives a lower bound against an oblivious adversary.
\end{proof}

Using a similar construction, we can obtain a lower bound on the approximation ratio of any center set that is computed by an algorithm 
(rather than an approximation on the cost of $OPT$ computed by an algorithm)
that queries amortized $o(k)$ distances $d(p,q)$ in expectation, where $p,q \in X$ (note that a $\Omega(k)$ bound on the running time follows already from the output complexity of a center set for a suitably chosen input). Here, we fix a set $P_0$ of $k-1$ points at pairwise distance $\Delta$ and insert two points $q_i$ and $p_i$ in each iteration. The point $p_i$ has distance $\Delta$ to all other points, while $q_i$ has distance $1$ to a uniformly random $p^{(*)}_i \in P_0$ and distance $\Delta$ to every other point. Intutively, it is hard for the algorithm to decide whether it should place the $k$\xth/ center on $q_j$ or $p_j$. Choosing $P_0 \cup \{p_i\}$ or $P_0 \setminus \{p^{*}_i\} \cup \{q_j,p_j\}$ gives a solution with cost $1$, while any other choice gives a solution with cost $\Delta$.

\begin{theorem}
	Let $\epsilon > 0$. Any (randomized) dynamic algorithm for finite metric spaces that can provide a center set whose cost is a $(\Delta - \epsilon)$-approximation to the $k$-center cost of its input issues amortized $\Omega(k)$ distance-oracle queries in expectation after $O(k^2)$ operations generated by an oblivious adversary.
\end{theorem}
\begin{proof}
	We use Yao's principle to prove the lower bound. Let $\ALG$ be any deterministic algorithm that \begin{inenum} \item takes as input a sequence of $t$ point insertions, deletions and cost queries, \item outputs, for each center-query operation $t'$, a center set whose cost is a $(\Delta - \epsilon)$ to the $k$-center cost of the set of input points after operation $t'-1$ and \item queries at most $tk / 512$ distances in expectation over the input distribution\end{inenum}. Since any dynamic algorithm that beats the claimed lower bound in the statement of the lemma can be turned into $\ALG$ by using it as a black-box, proving that $\ALG$ does not exist proves the claim. We define the uniform distribution over a set $S$ of sequences, each consisting of $t \defeq k + 50k^2$ operations. For the sake of analysis, we describe $S$ by iteratively constructing an element from it according to the uniform distribution. Let $P_0$ be a set of $k-1$ points that have pairwise distance $\Delta$. The first $k' \defeq k-1$ operations insert the points in $P_0$. For each $i \in \{0, \ldots, 10k^2-1\}$, we choose a uniformly random point $p^{*}_i \in P_0$. Operations $k' + 5i$ and $k' + 5i + 1$ insert points $q_i, p_i$ at distance $\Delta$ to each other. For every $p \in P_0 \setminus \{p^{*}_i\}$, we have $d(q_i, p) = \Delta$, and we set $d(q_i, p^{*}_i) = 1$. The distance between $p_i$ and any point $p \in P_0$ is $\Delta$. Operation $k' + 5i + 2$ queries the centers of the algorithm's solution, and operations $k' + 5i + 3$ and $k' + 5i + 4$ remove $q_i, p_i$. This way, every sequence $\sigma \in S$ is naturally partitioned into an \emph{initialization block} of $k'$ operations and $10k^2$ \emph{small blocks} of $5$ operations each. As a sequence is determined by $10k^2$ many $k$-ary choices, it follows that $\lvert S \rvert = k^{10k^2}$.

	For the sake of simplicity, we reveal the identity of $P_0$ and the pairwise distances of points in $P_0$ to the algorithm. This gives additional information to the algorithm without any cost for the algorithm, i.e., it only ``helps'' the algorithm. By an averaging argument, the algorithm must query at most $tk / 256$ distances on at least half of the sequences from $S$. Let $T \subseteq S$ denote this subset of sequences. By another averaging argument, for at least half of the small blocks $i$ in a sequence from $T$, $\ALG$ queries at most $(tk/256) / k^2 \leq 64k^3 / (256k^2) = k / 4$ distances between $\{q_i, p_i\}$ and points in $P_0$.
	
	Therefore, there must exist a $\sigma \in S$ and a corresponding small block $j$ so that $\ALG$ queries at most $k/4$ distances between $\{q_j,p_j\}$ and points in $P_0$. Let $\sigma' \in S$ be any sequence that equals $\sigma$ on the first $k' + 5j$ operations (i.e., on its points and distances). When the $j$\xth/ cost query on $\sigma$ or $\sigma'$ is issued, only the center sets $P_0 \cup \{p_j\}$ and $P_0 \setminus \{p^{*}_i\} \cup \{q_j,p_j\}$ have cost $1$, all other center sets have cost $\Delta$. To distinguish these two cases, the algorithm needs to determine $q_j$ or $p^{*}_j$, i.e., issue a query that reveals $d(q_j, p^{*}_j)$. Since $p^{*}_j$ is chosen uniformly at random, the probability that $\ALG$ distinguishes these two sequences is at most $k/4 \cdot 1 / (2\lvert P_0 \rvert - k/4) \leq k / 4 \cdot 4 / (7k) < 1$ by the union bound. This contradicts the assumption on $\ALG$.
	
	Finally, we observe that each of the $2k^2$ points $q_i, p_i$ are chosen from a set of $2k$ points, and therefore, a space of size $\lvert P _0 \rvert + 2k \cdot 2k^2  \in O(k^3)$ suffices to construct $S$. Choosing, for any $i,j$, $d(p_i, p_j) = 1$ if $p^{*}_i = p^{*}_j$ and $d(p_i, p_j) = \Delta$ otherwise, the space is metric.
\end{proof}

With exactly the same construction we obtain lower bounds of $\Delta$
for $k$-sum-of-radii, $k$-sum-of-diameters, $k$-median, a lower bound of
$\Delta^{2}$ for $k$-means, and a lower bound of $\Delta^p$ for $(k,p)$-clustering for each $p>0$. The reason is that in each of the possible metrics above, the cost of $\OPT$ is defined by $d(p_{i},p)$ for one point
$p_{i}$ and one point $p\in P_{0}$, and the other distances do not
contribute anything to $\OPT$. Similarly, for the second construction, the cost of a center set is determined by the choice of $p^{(*)}_j$.

\section{Algorithm for $k$-center against an adaptive adversary}
\label{sec:app-det-upper}

\NewDocumentCommand{\subtree}{O{T} m}{%
	#1(#2)
}
\NewDocumentCommand{\subtreepoints}{m}{%
	\mathcal{P}(#1)
}
\NewDocumentCommand{\detcen}{m}{%
	C_{#1}%
}

\SetKwFunction{FnInsertNode}{InsertIntoNode}
\SetKwFunction{FnDeleteNode}{DeleteFromNode}
\SetKwFunction{FnTryCenter}{TryMakeCenter}
\SetKwFunction{FnInsert}{InsertPoint}
\SetKwFunction{FnDelete}{DeletePoint}
\SetKwData{iscen}{isCenter}
\SetKwData{islowb}{lowerBoundWitness}

Given an upper bound $\costbound$, \cref{alg:det-guess} maintains a hierarchy on the input represented by a binary tree $T$, which we call \emph{clustering tree}. The main property of the clustering tree is that the input points are stored in the leaves and each inner node stores a $k$-center set for the $k+k$ centers that are stored at its two children.

\begin{definition}[clustering tree]
	\label{def:clustree}
	Let $\costbound > 0$, let $P$ be a set of points and let $T$ be a binary tree. We call $T$ a \emph{clustering tree} on $P$ with \emph{node-cost} $\costbound$ if the following conditions hold:
	\begin{enumerate}
		\item each node $u$ stores at most $2k$ points from $P$, denoted $\points{u}$, \label{enum:ct-leaf-node}
		\item for each node $u$, at most $k$ points, denoted $\detcen{u}$, are marked as centers. Either, their $k$-center cost is at most $\costbound$ on all points stored in $u$, or $u$ is marked as a witness that there is no center set with cost at most $\costbound/2$, \label{enum:ct-cost}
		\item each inner node stores the at most $2k$ centers of its children. \label{enum:ct-innernode}
	\end{enumerate}
\end{definition}

For each node $u$ in $T$, the algorithm maintains a corresponding graph on the at most $2k$ points $\points{u}$ it stores, which is called \emph{blocking graph}. Without loss of generality, we assume that $T$ is a full binary tree with $n/2k$ leaves. We explain in the proof of \cref{lem:det-main-kcenter} how to get rid of this assumption. For the sake of simplicity, we identify a node $u$ with its associated blocking graph $N=(V,E)$ in the following. For each node $N=(V,E)$, at most $k$ points are marked as centers (\iscen in \cref{alg:det-guess-node}), and the algorithm maintains the invariant that two centers $u,v \in V$ have distance at least $\costbound$ by keeping record of \emph{blocking} edges in the blocking graph between centers and points that have distance less than $\costbound$ to one of these centers. We say that a center $u$ \emph{blocks} a point $v$ (from being a center) if there is an edge $(u,v)$ in the blocking graph. In addition, the algorithm records whether $N$ contains more than $k$ points with pairwise distance greater than $\costbound$ (\islowb in \cref{alg:det-guess-node,alg:det-guess}).

\paragraph*{Insertions (see \FnInsert).} When a point $u$ is inserted into $T$, a node $N=(V,E)$ with less than $2k$ points is selected and it is checked whether $d(v,u) \leq \costbound$ for any center $v \in V$. If this is the case, the algorithm inserts an edge $(u,v)$ for every such center $v$ into $E$ and terminates afterwards. Otherwise, the algorithm checks whether the number of centers is less than $k$. If this is the case, it marks $u$ as a center and inserts an edge $(u,w)$ for \emph{each point} $w \in V$ with $d(u,w) \leq \costbound$ and recurses on the parent of $N$. Otherwise, if there are more than $k-1$ centers in $u$, the algorithm marks $N$ as witness and terminates.

\paragraph*{Deletions (see \FnDelete).} When a point $u$ is deleted from $T$, the point is first removed from the leaf $N$ (and the blocking graph) where it is stored. If $u$ was not a center, the algorithm terminates. Otherwise, the algorithm checks whether any points were unblocked (have no adjacent node in the blocking graph) and, if this is the case, proceeds by attempting to mark these points as centers and inserting them into the parent of $N$ one by one (after marking the first point as center, the remaining points may be blocked again). Afterwards, the algorithm recurses on the parent of $N$.

\begin{algorithm}
	\SetKwData{neigh}{neighbors}
	\SetKwData{newcen}{newCenters}
	\KwData{$\iscen$ is a boolean array on the elements of $V$, $\islowb$ is a boolean array on the nodes of $T$}
	\Fn{\FnInsertNode{$N = (V,E), p, \costbound$}}{
		insert $p$ into $V$ \;
		\ForEach{$v \in V \setminus \{ p \}$}{
			\If{$\iscen[v] \wedge d(p,v) \leq \costbound$}{
				insert $(v,p)$ into $E$ \;
			}
		}
		\Return \FnTryCenter($G, p, \costbound$) \;
	}
	\Fn{\FnDeleteNode{$N = (V,E), p, \costbound$}}{
		$\newcen \gets \emptyset$ \;
		$\neigh \gets \ngh{p}$ \;
		delete $p$ from $N$ \;
		\If{$\iscen[p] = true$}{
			\ForEach{$u \in \neigh$}{
				$\newcen \gets \newcen \cup \FnTryCenter{N, u, \costbound}$ \;
			}
		}
		\Return{$\newcen$} \;
	}
	\Fn{\FnTryCenter{$N = (V,E), p, \costbound$}}{
		\If{$\dg{p} = 0 \wedge \lvert \{ v \mid v \in V \wedge \iscen[v] \} \rvert < k$}{
			$\iscen[p] \gets true$ \;
			\ForEach{$v \in V \setminus \{ p \}$}{
				\If{$d(p,v) \leq \costbound$}{
					insert $(p,v)$ into $E$ \;
				}
			}
			$\islowb[N] \gets false$ \;
			\Return{$\{ p \}$} \;
		}
		\ElseIf{$\dg{p} = 0$}{
			$\islowb[N] \gets true$ \;
		}
		\Return{$\emptyset$} \;
	}
	\caption{\label{alg:det-guess-node} Insertion and deletion of a point $P$ in a node of the clustering tree $T$ (represented by a blocking graph $G$).}
\end{algorithm}

\begin{algorithm}
	\SetKwData{cen}{centers}
	\SetKwData{newcen}{newCenters}
	\SetKwData{failed}{failed}
	\KwData{$\islowb$ is a boolean array on the nodes of $T$}
	\Fn{\FnInsert{$T, p, \costbound$}}{
		$N \gets$ leaf in $T$ that contains less than $2k$ elements \;
		\Do{$\cen = \{ p \}$}{
			$\cen \gets \FnInsertNode{$N, p$}$ \;
			$N \gets N.parent$ \;
		}
	}
	\Fn{\FnDelete{$T, p, \costbound$}}{
		$N \gets$ leaf in $T$ that contains $p$ \;
		$\cen \gets \emptyset; \failed \gets false$ \;
		\Do{$N \neq null$}{
			$\cen \gets \cen \cup \FnDeleteNode{$N, p$}$ \;
			$\newcen \gets \emptyset$ \;
			\ForEach{$v \in \cen$}{
				$\newcen \gets \newcen \cup \FnInsertNode{N.parent, v}$ \;
			}
			$\cen \gets \newcen$ \;
			$N \gets N.parent$ \;
		}
	}
	\caption{\label{alg:det-guess} Insertion and deletion of a point $p$ in the clustering tree $T$.}
\end{algorithm}

Given a rooted tree $T$ and a node $u$ of $T$, we denote the subtree of $T$ that is rooted at $u$ by $\subtree{u}$. For a clustering tree $T$, we denote the set of all points stored at the leaves of $\subtree{u}$ by $\subtreepoints{u}$. Recall that the points directly stored at $u$ are denoted by $\points{u}$. Observe that for each node $u$ in a clustering tree, $\subtree{u}$ is a clustering tree of $\subtreepoints{u}$.

\subsection{Feasibility}

We show that \cref{alg:det-guess} maintains a clustering tree.

\begin{lemma}
	\label{lem:det-clustree-insert}
	Let $T$ be a clustering tree on a point set $P$. After calling $\FnInsert(T, p)$ (see \cref{alg:det-guess}) for some point $p \notin P$, $T$ is a clustering tree on $P \cup \{ p \}$.
\end{lemma}
\begin{proof}
	We prove the statement by induction over the recursive calls of \FnInsertNode in \FnInsert. Let $N=(V,E)$ be the leaf in $T$ where $p$ is inserted. Condition~\ref{enum:ct-leaf-node} in \cref{def:clustree} is guaranteed for $N$ by \FnInsert. The algorithm \FnInsertNode ensures that $p$ is marked as a center only if it is not within distance $\costbound$ of any other center. Let $C$ be the center set of $N$ before inserting $p$. By condition~\ref{enum:ct-cost}, $C$ has $\cost[V]{C} \leq \costbound$. Let $C'$ be the center set of any optimal solution on $V$. If $\cost[V]{C'} \leq \costbound / 2$, each center of $C$ covers at least one cluster of $C'$. By pigeonhole principle, $p$ is not blocked by a center in $C$ if and only if $\lvert C \rvert < k$. Otherwise, if $\cost[V]{C'} > \costbound/2$, $p$ is chosen if no center in $C$ covers $p$ and $\lvert C' \rvert < k$, or $N$ is marked as witness. It follows that condition~\ref{enum:ct-cost} is still satisfied for $N$ after \FnInsertNode terminates.

	Now, let $N=(V,E)$ be any inner node in a call to $\FnInsertNode(N,p)$. We note that such call is only made if $p$ was marked as a center in its child $N'$ on which $\FnInsertNode$ was called before by $\FnInsert$. Thus, $p$ is inserted into $V$ if and only if $p$ is a center in $N'$. Therefore, condition~\ref{enum:ct-cost} is satisfied for $N$. By the above reasoning, condition~\ref{enum:ct-innernode} is also satisfied.
\end{proof}

\begin{lemma}
	Let $T$ be a clustering tree on a point set $P$. After calling $\FnDelete(T, p)$ (see \cref{alg:det-guess}) for some point $p \in P$, $T$ is a clustering tree on $P \setminus \{ p \}$.
\end{lemma}
\begin{proof}
	We prove the statement by induction over the loop's iterations in \FnDelete. Let $N=(V,E)$ be the leaf in $T$ where $p$ was inserted. \FnDeleteNode deletes $p$ from $N$ and iterates over all unblocked points to mark them as centers one by one. After deleting $p$, condition~\ref{enum:ct-leaf-node} holds for $N$. Let $U$ be the set of unblocked points after removing $p$, let $C$ be the center set after removing $p$ from $V$, and let $C'$ be the center set of any optimal solution on $V \setminus \{ p \}$. If $\cost[V]{C'} \leq \costbound / 2$, each unblocked point from $U$ covers at least one cluster of $C'$ with cost $\costbound$. Therefore, any selection of $k - \lvert C \rvert$ points from $U$ that do not block each other together with $C$ is a center set for $V$ with cost $\costbound$. Otherwise, if $\cost[V]{C'} > \costbound / 2$, a set of at most $k - \lvert C \rvert$ points from $U$ is chosen, or $N$ is marked as witness. It follows that condition~\ref{enum:ct-cost} is still satisfied for $N$ after \FnDeleteNode terminates. Let $C''$ be the center set that is returned by \FnDeleteNode. \FnDelete inserts all points from $C''$ into the parent node. This reinstates condition~\ref{enum:ct-innernode} on the parent node. Then, \FnDeleteNode recurses on the parent.
\end{proof}

\subsection{Approximation guarantees}

We use the following notion of super clusters and its properties to prove the $O(k)$ upper bound on the approximation ratio of the algorithm.

\begin{definition}[super cluster]
	\label{def:supercluster}
	Let $P$ be a set of points and let $C$ be a center set with $\cost[P]{C} \leq 2\costbound$. Consider the graph $G = (C, E)$, where $E = \{ (u,v) \mid d(u,v) \leq 2\costbound \}$. For every connected component in $G$, we call the union of clusters corresponding to this component a \emph{super cluster}.
\end{definition}

\begin{lemma}
	\label{lem:clustree-supercluster}
	Let $T$ be a clustering tree constructed by \cref{alg:det-guess} with node-cost $\costbound$ and no node marked as witness. For any node $N$ in $T$, $N$ contains one point from each super cluster in $\subtreepoints{N}$ that is marked as center.
\end{lemma}
\begin{proof}
	Let $S$ be a supercluster of $P$. Let $N$ be a node in $T$ that contains a point $p \in S$. For the sake of contradiction, assume that there exists no point $q \in S \cap N$ that is marked as center. By \cref{def:supercluster}, the $k$-center clustering cost of the centers in $N$ is greater than $\costbound$. By \cref{def:clustree}.\ref{enum:ct-cost}, this implies that a point must be marked as witness. This is a contradiction to the assumption that no node is marked as witness.
\end{proof}

The following simple observation leads to the bound of $O(\log (n/k))$ on the approximation ratio.

\begin{observation}
	\label{lem:kcenter-of-cluster}
	Let $c > 0$, let $P, Q$ be sets of points and let $C, C'$ so that $\cost[P]{C} \leq c$ and $\cost[Q]{C'} \leq c$. For every $k$-center set $C''$ with $\cost[C \cup C']{C''} \leq c'$ we have $\cost[P \cup Q]{C''} \leq c + c'$ by the triangle inequality.
\end{observation}

We combine the previous results and obtain the following approximation ratio for our algorithm.

\begin{lemma}
	\label{lem:det-kcenter-cost}
	Let $T$ be a clustering tree on a point set $P$ that is constructed by \cref{alg:det-guess} with node-cost $\costbound$. Let $C$ be the points in the root of $T$ that are marked as centers. If no node in $C$ is marked as witness, $ \cost[P]{C} \leq \min\{k, \log(n/k)\} \cdot 4 \opt{P}$. Otherwise, $\opt{P} \geq \costbound / 2$.
\end{lemma}
\begin{proof}
	Assume that $\opt{P} \leq \costbound / 2$, as otherwise, there exists a node that stores at least $k+1$ points that have pairwise distance $\costbound$, which implies the claim. First, we prove $\cost[P]{C} \leq \log(n/k) \cdot \costbound$. It follows from \cref{lem:kcenter-of-cluster} that $C$ has $k$-center cost $2\costbound$ on the points stored in the root's children. Since $T$ has depth at most $\log(n/k)$, it follows by recursively applying \cref{lem:kcenter-of-cluster} on the children that $C$ also has cost $\log(n/k) \cdot \costbound$ on $P$.

	Now, we prove $\cost[P]{C} \leq k$. Let $p \in P$, and let $S$ be the super cluster of $p$ with corresponding optimal center set $C'$. We have $d(p,C') \leq \costbound / 2$. By \cref{lem:clustree-supercluster}, there exists a center $q \in C$ so that $d(q,C') \leq \costbound / 2$. By the definition of super clusters, for every $x,y \in C'$, $d(x,y) \leq (k-1) \cdot 2\costbound$. It follows from the triangle inequality that $d(p,q) \leq 2k \costbound$.
\end{proof}

\subsection{Running time}

\begin{lemma}
	\label{lem:det-time}
	The amortized update time of \cref{alg:det-guess} is $O(k \log (n/k)$.
\end{lemma}
\begin{proof}
	To maintain blocking graphs efficiently, the graphs are stored in adjacency list representation and the degrees of the vertices as well as the number of centers are stored in counters. When a point $p$ is inserted, $5k \log (n/k)$ tokens are paid into the account of $p$. Each token pays for a (universally) constant amount of work.

	Each point is inserted at most once in each of the $\log (n/k)$ nodes from the leaf where it is inserted up to the root. The key observation is that it is marked as a center in each of these nodes at most once (when it is inserted, or later when a center is deleted): Marking a point as center is irrevocable until it is deleted. For each node $N$ and point $p \in N$, it can be checked in constant time whether $p$ can be marked as a center in $N$ by checking its degree in the blocking graph. Marking $p$ as a center takes time $O(2k)$ because it is sufficient to check the distance to all other at most $2k$ points in $N$ and insert the corresponding blocking edges. For each point $p$, we charge the time it takes to \emph{mark} $p$ as a center to the account of $p$. Therefore, marking $p$ as center withdraws a most $2k \log (n/k)$ tokens from its account in total.

	Consider the insertion of a point $p$ into a node. As mentioned, $p$ is inserted in at most $\log (n/k)$ nodes, and in each of these nodes, it is inserted at most once (when it is inserted into the tree, or when it is marked as a center in a child node). Inserting a point into a node $N$ requires the algorithm to check the distance to all at most $k$ centers in $N$ to insert blocking edges, which results in at most $k \log (n/k)$ work in total.

	It remains to analyze the time that is required to update the tree when a point $p$ is deleted. For any node $N$, if $p$ is not a center in $N$, deleting $p$ takes constant time. Otherwise, if $p$ is a center, the algorithm needs to check its at most $2k$ neighbors in the blocking graph one by one whether they can be marked as centers. Checking a point $q$ takes only constant time, and marking $q$ as a center has already been charged to $q$ by the previous analysis. All centers that have been marked have to be inserted into the parent of $N$, but this has also been charged to the corresponding points. Therefore, deleting $p$ consumes at most $2k \log (n/k)$ tokens from the account of $p$.
\end{proof}

\subsection{Main result}

It only remains to combine all previous results to obtain \cref{thm:mpt-det-upper}.

\begin{theorem}[\cref{thm:mpt-det-upper}]
	\label{lem:det-main-kcenter}
	Let $\epsilon, k > 0$. There exists a deterministic algorithm for the dynamic $k$-center problem that has amortized update time $O(k \log(n) \log(\dmax) / \log(1+\epsilon))$ and approximation factor $(1+\epsilon) \cdot \min\{4k, 4\log(n/k)\}$.
\end{theorem}
\begin{proof}
	Since $P = (\points{1}, \ldots, \points{n})$ is a dynamic point set, its size $n_i := \lvert \points{i} \rvert$ can increase over time. Therefore, we need to remove the assumption that the clustering tree has depth $\log(\max_{t \in [n]} n_t)$. First, we note that we can insert and delete points so that the clustering tree $T$ is a complete binary tree (the inner nodes induce a full binary tree, and all leaves on the last level are aligned left): We always insert points into the left-most leaf on the last level of $T$ that is not full; when a point is deleted from a leaf $N$ that is not the right-most leaf $N'$ on the last level of $T$, we delete an arbitrary point $q$ from $N'$ and insert $q$ into $N$. Deleting and reinserting points this way can be seen as two update operations, and therefore, it can only increase time required to update the tree by a factor of $3$. Furthermore, any leaf can be turned into an inner node by adding a copy of itself as its left child and adding an empty node as its right child. Vice versa, an empty leaf and its (left) sibling can be contracted into its parent. This way, the algorithm can guarantee that the depth of the tree is between $\lfloor \log(n_t/k) \rfloor$ and $\lceil \log(n_t/k) \rceil$ at all times $t \in [n]$.

	Recall that $d(x,y) \geq 1$ for every $x,y \in X$, and $\dmax = \max_{x,y \in X} d(x,y)$. For every $\Gamma \in \{ (1+\epsilon)^i \mid i \in [\lceil \log_{1+\epsilon} (\dmax) \rceil]$, the algorithm maintains an instance $T_\Gamma$ of a clustering tree with node-cost $\costbound = (1+\epsilon)^\Gamma$ by invoking \cref{alg:det-guess}. After every update, the algorithm determines the smallest $\Gamma$ so that no node in $T_\Gamma$ is marked as witness, and it reports the center set of the root of $T_\Gamma$. The bound on the cost follows immediately from \cref{lem:det-kcenter-cost}. Since there are at most $\log(\dmax) / \log(1+\epsilon)$ instances, the bound on the time follows from \cref{lem:det-time}.

\end{proof}

\end{document}

%% file: preamble.tex
\declaretheoremstyle[
  bodyfont=\normalfont\itshape,
  ]{normal_style}
\declaretheoremstyle[
  headfont=\normalfont\itshape,
  notefont=\normalfont\itshape,
  qed=\qedsymbol
  ]{proof_style}

\declaretheorem[
  name=Theorem,
  style=normal_style]{theorem}

\declaretheorem[
  name=Lemma,
  numberlike=theorem,
  style=normal_style]{lemma}

\declaretheorem[
  name=Observation,
  numberlike=theorem,
  style=normal_style]{observation}
\declaretheorem[
  name=Definition,
  numberlike=theorem,
  style=normal_style]{definition}
\makeatletter
\let\proof\@undefined
\let\endproof\@undefined
\makeatother
\declaretheorem[
  name=Proof,
  numbered=no,
  style=proof_style]{proof}

\def\xth/{%
	\textsuperscript{th}%
}

%% file: dynamic-clustering.bbl
\begin{thebibliography}{10}

\bibitem{ahmadiannsw_17}
Sara Ahmadian, Ashkan Norouzi{-}Fard, Ola Svensson, and Justin Ward.
\newblock Better guarantees for k-means and euclidean k-median by primal-dual
  algorithms.
\newblock In {\em 58th {IEEE} Annual Symposium on Foundations of Computer
  Science, {FOCS} 2017, Berkeley, CA, USA, October 15-17, 2017}, pages 61--72.
  {IEEE} Computer Society, 2017.

\bibitem{ahmadian2016approximation}
Sara Ahmadian and Chaitanya Swamy.
\newblock Approximation algorithms for clustering problems with lower bounds
  and outliers.
\newblock In {\em 43rd International Colloquium on Automata, Languages, and
  Programming (ICALP 2016)}. Schloss Dagstuhl-Leibniz-Zentrum fuer Informatik,
  2016.

\bibitem{CharikarCFM04}
Moses~Charikar an~Chandra~Chekuri, Tom{\'{a}}s Feder, and Rajeev Motwani.
\newblock Incremental clustering and dynamic information retrieval.
\newblock {\em {SIAM} J. Comput.}, 33(6):1417--1440, 2004.
\newblock \href {https://doi.org/10.1137/S0097539702418498}
  {\path{doi:10.1137/S0097539702418498}}.

\bibitem{AwasthiCKS15}
Pranjal Awasthi, Moses Charikar, Ravishankar Krishnaswamy, and Ali~Kemal Sinop.
\newblock The hardness of approximation of euclidean k-means.
\newblock In Lars Arge and J{\'{a}}nos Pach, editors, {\em 31st International
  Symposium on Computational Geometry, SoCG 2015, June 22-25, 2015, Eindhoven,
  The Netherlands}, volume~34 of {\em LIPIcs}, pages 754--767. Schloss Dagstuhl
  - Leibniz-Zentrum f{\"{u}}r Informatik, 2015.
\newblock \href {https://doi.org/10.4230/LIPIcs.SOCG.2015.754}
  {\path{doi:10.4230/LIPIcs.SOCG.2015.754}}.

\bibitem{DBLP:journals/corr/abs-2004-08432}
Aaron Bernstein, Jan van~den Brand, Maximilian~Probst Gutenberg, Danupon
  Nanongkai, Thatchaphol Saranurak, Aaron Sidford, and He~Sun.
\newblock Fully-dynamic graph sparsifiers against an adaptive adversary.
\newblock {\em CoRR}, abs/2004.08432, 2020.
\newblock URL: \url{https://arxiv.org/abs/2004.08432}, \href
  {http://arxiv.org/abs/2004.08432} {\path{arXiv:2004.08432}}.

\bibitem{byrkaprst_17}
Jaroslaw Byrka, Thomas~W. Pensyl, Bartosz Rybicki, Aravind Srinivasan, and Khoa
  Trinh.
\newblock An improved approximation for \emph{k}-median and positive
  correlation in budgeted optimization.
\newblock {\em {ACM} Trans. Algorithms}, 13(2):23:1--23:31, 2017.

\bibitem{can1987dynamic}
Fazli Can and E~Ozkarahan.
\newblock A dynamic cluster maintenance system for information retrieval.
\newblock In {\em Proceedings of the 10th annual international ACM SIGIR
  conference on Research and development in information retrieval}, pages
  123--131, 1987.

\bibitem{ChanGS18}
T.{-}H.~Hubert Chan, Arnaud Guerqin, and Mauro Sozio.
\newblock Fully dynamic \emph{k}-center clustering.
\newblock In Pierre{-}Antoine Champin, Fabien~L. Gandon, Mounia Lalmas, and
  Panagiotis~G. Ipeirotis, editors, {\em Proceedings of the 2018 World Wide Web
  Conference on World Wide Web, {WWW} 2018, Lyon, France, April 23-27, 2018},
  pages 579--587. {ACM}, 2018.
\newblock \href {https://doi.org/10.1145/3178876.3186124}
  {\path{doi:10.1145/3178876.3186124}}.

\bibitem{CharikarP04}
Moses Charikar and Rina Panigrahy.
\newblock Clustering to minimize the sum of cluster diameters.
\newblock {\em J. Comput. Syst. Sci.}, 68(2):417--441, 2004.
\newblock \href {https://doi.org/10.1016/j.jcss.2003.07.014}
  {\path{doi:10.1016/j.jcss.2003.07.014}}.

\bibitem{Cohen-AddadHPSS19}
Vincent Cohen{-}Addad, Niklas Hjuler, Nikos Parotsidis, David Saulpic, and
  Chris Schwiegelshohn.
\newblock Fully dynamic consistent facility location.
\newblock In Hanna~M. Wallach, Hugo Larochelle, Alina Beygelzimer, Florence
  d'Alch{\'{e}}{-}Buc, Emily~B. Fox, and Roman Garnett, editors, {\em Advances
  in Neural Information Processing Systems 32: Annual Conference on Neural
  Information Processing Systems 2019, NeurIPS 2019, 8-14 December 2019,
  Vancouver, BC, Canada}, pages 3250--3260, 2019.
\newblock URL:
  \url{http://papers.nips.cc/paper/8588-fully-dynamic-consistent-facility-location}.

\bibitem{DoddiMRTW00}
Srinivas Doddi, Madhav~V. Marathe, S.~S. Ravi, David~Scot Taylor, and Peter
  Widmayer.
\newblock Approximation algorithms for clustering to minimize the sum of
  diameters.
\newblock {\em Nord. J. Comput.}, 7(3):185--203, 2000.

\bibitem{DBLP:conf/icalp/EvaldFGW21}
Jacob Evald, Viktor Fredslund{-}Hansen, Maximilian~Probst Gutenberg, and
  Christian Wulff{-}Nilsen.
\newblock Decremental {APSP} in unweighted digraphs versus an adaptive
  adversary.
\newblock In Nikhil Bansal, Emanuela Merelli, and James Worrell, editors, {\em
  48th International Colloquium on Automata, Languages, and Programming,
  {ICALP} 2021, July 12-16, 2021, Glasgow, Scotland (Virtual Conference)},
  volume 198 of {\em LIPIcs}, pages 64:1--64:20. Schloss Dagstuhl -
  Leibniz-Zentrum f{\"{u}}r Informatik, 2021.
\newblock \href {https://doi.org/10.4230/LIPIcs.ICALP.2021.64}
  {\path{doi:10.4230/LIPIcs.ICALP.2021.64}}.

\bibitem{feder1988optimal}
Tom{\'a}s Feder and Daniel Greene.
\newblock Optimal algorithms for approximate clustering.
\newblock In {\em Proceedings of the twentieth annual ACM symposium on Theory
  of computing}, pages 434--444, 1988.

\bibitem{gibson2010metric}
Matt Gibson, Gaurav Kanade, Erik Krohn, Imran~A Pirwani, and Kasturi
  Varadarajan.
\newblock On metric clustering to minimize the sum of radii.
\newblock {\em Algorithmica}, 57(3):484--498, 2010.

\bibitem{DBLP:journals/tcs/Gonzalez85}
Teofilo~F. Gonzalez.
\newblock Clustering to minimize the maximum intercluster distance.
\newblock {\em Theor. Comput. Sci.}, 38:293--306, 1985.
\newblock \href {https://doi.org/10.1016/0304-3975(85)90224-5}
  {\path{doi:10.1016/0304-3975(85)90224-5}}.

\bibitem{Goranci20}
Gramoz Goranci, Monika Henzinger, Dariusz Leniowski, and Alexander Svozil.
\newblock Fully dynamic k-center clustering in doubling metrics.
\newblock {\em CoRR}, abs/1908.03948, 2019.
\newblock URL: \url{http://arxiv.org/abs/1908.03948}, \href
  {http://arxiv.org/abs/1908.03948} {\path{arXiv:1908.03948}}.

\bibitem{DBLP:conf/soda/GutenbergW20a}
Maximilian~Probst Gutenberg and Christian Wulff{-}Nilsen.
\newblock Decremental {SSSP} in weighted digraphs: Faster and against an
  adaptive adversary.
\newblock In Shuchi Chawla, editor, {\em Proceedings of the 2020 {ACM-SIAM}
  Symposium on Discrete Algorithms, {SODA} 2020, Salt Lake City, UT, USA,
  January 5-8, 2020}, pages 2542--2561. {SIAM}, 2020.
\newblock \href {https://doi.org/10.1137/1.9781611975994.155}
  {\path{doi:10.1137/1.9781611975994.155}}.

\bibitem{hassanzadeh2009framework}
Oktie Hassanzadeh, Fei Chiang, Hyun~Chul Lee, and Ren{\'e}e~J Miller.
\newblock Framework for evaluating clustering algorithms in duplicate
  detection.
\newblock {\em Proceedings of the VLDB Endowment}, 2(1):1282--1293, 2009.

\bibitem{DBLP:conf/esa/HenzingerK20}
Monika Henzinger and Sagar Kale.
\newblock Fully-dynamic coresets.
\newblock In Fabrizio Grandoni, Grzegorz Herman, and Peter Sanders, editors,
  {\em 28th Annual European Symposium on Algorithms, {ESA} 2020, September 7-9,
  2020, Pisa, Italy (Virtual Conference)}, volume 173 of {\em LIPIcs}, pages
  57:1--57:21. Schloss Dagstuhl - Leibniz-Zentrum f{\"{u}}r Informatik, 2020.
\newblock \href {https://doi.org/10.4230/LIPIcs.ESA.2020.57}
  {\path{doi:10.4230/LIPIcs.ESA.2020.57}}.

\bibitem{DBLP:journals/algorithmica/HenzingerLM20}
Monika Henzinger, Dariusz Leniowski, and Claire Mathieu.
\newblock Dynamic clustering to minimize the sum of radii.
\newblock {\em Algorithmica}, 82(11):3183--3194, 2020.
\newblock \href {https://doi.org/10.1007/s00453-020-00721-7}
  {\path{doi:10.1007/s00453-020-00721-7}}.

\bibitem{DBLP:journals/corr/abs-2011-00977}
Monika Henzinger and Pan Peng.
\newblock Constant-time dynamic weight approximation for minimum spanning
  forest.
\newblock {\em CoRR}, abs/2011.00977, 2020.
\newblock URL: \url{https://arxiv.org/abs/2011.00977}, \href
  {http://arxiv.org/abs/2011.00977} {\path{arXiv:2011.00977}}.

\bibitem{hochbaum1985best}
Dorit~S Hochbaum and David~B Shmoys.
\newblock A best possible heuristic for the k-center problem.
\newblock {\em Mathematics of operations research}, 10(2):180--184, 1985.

\bibitem{hochbaum1986unified}
Dorit~S Hochbaum and David~B Shmoys.
\newblock A unified approach to approximation algorithms for bottleneck
  problems.
\newblock {\em Journal of the ACM (JACM)}, 33(3):533--550, 1986.

\bibitem{hsu1979easy}
Wen-Lian Hsu and George~L Nemhauser.
\newblock Easy and hard bottleneck location problems.
\newblock {\em Discrete Applied Mathematics}, 1(3):209--215, 1979.

\bibitem{DBLP:conf/stoc/JainMS02}
Kamal Jain, Mohammad Mahdian, and Amin Saberi.
\newblock A new greedy approach for facility location problems.
\newblock In John~H. Reif, editor, {\em Proceedings on 34th Annual {ACM}
  Symposium on Theory of Computing, May 19-21, 2002, Montr{\'{e}}al,
  Qu{\'{e}}bec, Canada}, pages 731--740. {ACM}, 2002.
\newblock \href {https://doi.org/10.1145/509907.510012}
  {\path{doi:10.1145/509907.510012}}.

\bibitem{jain2002new}
Kamal Jain, Mohammad Mahdian, and Amin Saberi.
\newblock A new greedy approach for facility location problems.
\newblock In {\em Proceedings of the thiry-fourth annual ACM symposium on
  Theory of computing}, pages 731--740, 2002.

\bibitem{lee2017improved}
Euiwoong Lee, Melanie Schmidt, and John Wright.
\newblock Improved and simplified inapproximability for k-means.
\newblock {\em Information Processing Letters}, 120:40--43, 2017.

\bibitem{McCutchenK08}
Richard~Matthew McCutchen and Samir Khuller.
\newblock Streaming algorithms for k-center clustering with outliers and with
  anonymity.
\newblock In Ashish Goel, Klaus Jansen, Jos{\'{e}} D.~P. Rolim, and Ronitt
  Rubinfeld, editors, {\em Approximation, Randomization and Combinatorial
  Optimization. Algorithms and Techniques, 11th International Workshop,
  {APPROX} 2008, and 12th International Workshop, {RANDOM} 2008, Boston, MA,
  USA, August 25-27, 2008. Proceedings}, volume 5171 of {\em Lecture Notes in
  Computer Science}, pages 165--178. Springer, 2008.
\newblock \href {https://doi.org/10.1007/978-3-540-85363-3\_14}
  {\path{doi:10.1007/978-3-540-85363-3\_14}}.

\bibitem{MettuP04}
Ramgopal~R. Mettu and C.~Greg Plaxton.
\newblock Optimal time bounds for approximate clustering.
\newblock {\em Mach. Learn.}, 56(1-3):35--60, 2004.
\newblock \href {https://doi.org/10.1023/B:MACH.0000033114.18632.e0}
  {\path{doi:10.1023/B:MACH.0000033114.18632.e0}}.

\bibitem{DBLP:conf/stoc/NanongkaiS17}
Danupon Nanongkai and Thatchaphol Saranurak.
\newblock Dynamic spanning forest with worst-case update time: adaptive, las
  vegas, and o(n\({}^{\mbox{1/2 - {\(\epsilon\)}}}\))-time.
\newblock In Hamed Hatami, Pierre McKenzie, and Valerie King, editors, {\em
  Proceedings of the 49th Annual {ACM} {SIGACT} Symposium on Theory of
  Computing, {STOC} 2017, Montreal, QC, Canada, June 19-23, 2017}, pages
  1122--1129. {ACM}, 2017.
\newblock \href {https://doi.org/10.1145/3055399.3055447}
  {\path{doi:10.1145/3055399.3055447}}.

\bibitem{DBLP:books/sp/Overmars83}
Mark~H. Overmars.
\newblock {\em The Design of Dynamic Data Structures}, volume 156 of {\em
  Lecture Notes in Computer Science}.
\newblock Springer, 1983.
\newblock \href {https://doi.org/10.1007/BFb0014927}
  {\path{doi:10.1007/BFb0014927}}.

\bibitem{qian2010case}
Feng Qian, Abhinav Pathak, Yu~Charlie Hu, Zhuoqing~Morley Mao, and Yinglian
  Xie.
\newblock A case for unsupervised-learning-based spam filtering.
\newblock {\em ACM SIGMETRICS performance evaluation review}, 38(1):367--368,
  2010.

\bibitem{abs-1908-02645}
Melanie Schmidt and Christian Sohler.
\newblock Fully dynamic hierarchical diameter k-clustering and k-center.
\newblock {\em CoRR}, abs/1908.02645, 2019.
\newblock URL: \url{http://arxiv.org/abs/1908.02645}, \href
  {http://arxiv.org/abs/1908.02645} {\path{arXiv:1908.02645}}.

\bibitem{sheikhalishahi2015fast}
Mina Sheikhalishahi, Andrea Saracino, Mohamed Mejri, Nadia Tawbi, and Fabio
  Martinelli.
\newblock Fast and effective clustering of spam emails based on structural
  similarity.
\newblock In {\em International Symposium on Foundations and Practice of
  Security}, pages 195--211. Springer, 2015.

\bibitem{DBLP:conf/stoc/Wajc20}
David Wajc.
\newblock Rounding dynamic matchings against an adaptive adversary.
\newblock In Konstantin Makarychev, Yury Makarychev, Madhur Tulsiani, Gautam
  Kamath, and Julia Chuzhoy, editors, {\em Proccedings of the 52nd Annual {ACM}
  {SIGACT} Symposium on Theory of Computing, {STOC} 2020, Chicago, IL, USA,
  June 22-26, 2020}, pages 194--207. {ACM}, 2020.
\newblock \href {https://doi.org/10.1145/3357713.3384258}
  {\path{doi:10.1145/3357713.3384258}}.

\end{thebibliography}
